\documentclass[envcountsame]{llncs}

\usepackage{soul}
\usepackage{url}
\usepackage[hidelinks]{hyperref}
\usepackage[utf8]{inputenc}
\usepackage{caption}
\usepackage{graphicx}
\usepackage{amsmath}
\usepackage{amsthm}
\usepackage{booktabs}
\usepackage{algorithm}
\usepackage{algorithmic}
\usepackage{paralist}
\usepackage{amsfonts}
\usepackage{latexsym, amssymb}
\usepackage{stmaryrd}
\usepackage{wasysym}
\usepackage{bbm}
\usepackage{rotating}
\usepackage{pifont}
\usepackage{array}
\usepackage{multirow}
\usepackage{tikz}
\usepackage{xcolor}
\usepackage{enumitem}
\usepackage{xspace}
\usepackage{thm-restate}
\usepackage{phonetic}
\usepackage{bm}
\usepackage{multicol}

\newcommand{\p}{\varphi}

\newcommand{\eset}{\emptyset}

\newcommand{\Int}{\ensuremath{\mathcal{I}}\xspace}

\newcommand{\tr}{^{\dagger}}

\newcommand{\red}{^{\dagger}}

\newcommand{\sttr}[2]{\ensuremath{{\pi}_{#1}(#2)}\xspace}

\newcommand{\sig}[1]{\ensuremath{\Sigma_{#1}}\xspace}

\newcommand{\sub}[1]{\ensuremath{\mathsf{for}(#1)}\xspace}

\newcommand{\con}[1]{\ensuremath{\mathsf{con}(#1)}\xspace}

\newcommand{\conseq}[1]{\ensuremath{\mathit{Cons}(#1)}\xspace}

\newcommand{\assign}{\ensuremath{\sigma}\xspace}

\newcommand{\tvalue}{\ensuremath{\bar{\assign}}\xspace}

\newcommand{\dom}{\ensuremath{{\sf dom}}\xspace}

\newcommand{\monodic}{\ensuremath{\mathop{\ooalign{$\Box$ \cr \kern0.57ex \raisebox{0.2ex}{\scalebox{0.55}{$1$}}}\rule{0pt}{1.5ex} \kern-0.7ex}}\xspace}

\newcommand{\Lang}{\ensuremath{\mathcal{L}}}

\newcommand{\eq}[1]{\ensuremath{[{#1}]}\xspace}

\newcommand{\ra}{\ensuremath{{\cal R }_{\sf A}}}

\newcommand{\NC}{\ensuremath{{\sf N_C}}\xspace}
\newcommand{\NI}{\ensuremath{{\sf N_I}}\xspace}

\newcommand{\NR}{\ensuremath{{\sf N_R}}\xspace}

\newcommand{\ALC}{\ensuremath{\smash{\mathcal{ALC}}}\xspace}
\newcommand{\ALCO}{\ensuremath{\smash{\mathcal{ALCO}}}\xspace}

\newcommand{\ALCOu}{\ensuremath{\smash{\mathcal{ALCO}_{\!u}}}\xspace}

\newcommand{\ALCOud}{\ensuremath{\smash{\mathcal{ALCO}_{\!u}^{\defdes}}}\xspace}

\newcommand{\ALCOd}{\ensuremath{\smash{\mathcal{ALCO}^{\defdes}}}\xspace}

\newcommand{\EL}{\ensuremath{\smash{\mathcal{EL}}}\xspace}
\newcommand{\ELO}{\ensuremath{\smash{\mathcal{ELO}}}\xspace}
\newcommand{\ELOu}{\ensuremath{\smash{\mathcal{ELO}_{\!u}}}\xspace}

\newcommand{\ELOd}{\ensuremath{\smash{\mathcal{ELO}^{\defdes}}}\xspace}
\newcommand{\ELOud}{\ensuremath{\smash{\mathcal{ELO}_{\!u}^{\defdes}}}\xspace}

\newcommand{\NKRd}{\ensuremath{\smash{\mathcal{ALCO}^{\defdes \ast}}}\xspace}

\newcommand{\tp}{\ensuremath{t}\xspace}

\newcommand{\fincanmod}{\ensuremath{\Imc_{{\mathsf A},\Omc}}\xspace}

\newcommand{\Ex}{\ensuremath{\mathsf{Ex}}\xspace}

\newcommand{\exrel}{\ensuremath{\downarrow\Ex}\xspace}

\newcommand{\OutDom}{\ensuremath{\Delta}\xspace}
\newcommand{\InnDom}{\ensuremath{\mathfrak{d}}\xspace}

\newcommand{\defdes}{\ensuremath{\smash{\iota}}\xspace}

\newcommand{\concleqone}{\ensuremath{C^{\leq 1}_{\tau}}\xspace}

\newcommand{\topex}{\boldsymbol{{\mathsf{T}}}}

\newcommand{\simul}{\ensuremath{\leq}}

\newcommand{\PTime}{\textsc{PTime}}

\newcommand{\ExpTime}{\textsc{ExpTime}}
\newcommand{\TwoExpTime}{\ensuremath{2\textsc{ExpTime}}}

\newcommand{\Amc}{\ensuremath{\mathcal{A}}\xspace}
\newcommand{\Bmc}{\ensuremath{\mathcal{B}}\xspace}
\newcommand{\Cmc}{\ensuremath{\mathcal{C}}\xspace}
\newcommand{\Dmc}{\ensuremath{\mathcal{D}}\xspace}

\newcommand{\Imc}{\ensuremath{\mathcal{I}}\xspace}
\newcommand{\Jmc}{\ensuremath{\mathcal{J}}\xspace}

\newcommand{\Lmc}{\ensuremath{\mathcal{L}}\xspace}

\newcommand{\Omc}{\ensuremath{\mathcal{O}}\xspace}

\newcommand{\Rmc}{\ensuremath{\mathcal{R}}\xspace}
\newcommand{\Smc}{\ensuremath{\mathcal{S}}\xspace}

\title{On Free Description Logics \\ with Definite Descriptions}

\author{%
Alessandro Artale$^1$\and
Andrea Mazzullo$^1$\and
Ana Ozaki$^2$\and
Frank Wolter$^3$
}

\institute{
$^1$Free University of Bozen-Bolzano\quad
$^2$University of Bergen \quad
$^3$University of Liverpool}

\begin{document}

\maketitle

\begin{abstract}

	\emph{Definite descriptions} are phrases of the form `the $x$ such that $\p$', used to refer to single entities in a context.
	They are often more meaningful to users than individual names alone, in particular when modelling or querying data over ontologies.
	%
	%
	We investigate
	\emph{free description logics} 
	with both individual names and definite descriptions as
	terms of the language,
	while also accounting for their possible lack of denotation.
	%
	We focus on the extensions of $\mathcal{ALC}$ and, respectively, $\mathcal{EL}$ with nominals, the universal role, and definite descriptions.
	We show that standard reasoning in these extensions 
    is not harder than in the original languages,
    %
	and
	we
	characterise the expressive power of concepts relative to first-order formulas using a suitable notion of bisimulation.
	Moreover, we
	lay the foundations for automated support
	for definite descriptions generation
	by studying the complexity of deciding the existence of definite descriptions for an individual under an ontology. 
	Finally, we provide a polynomial-time reduction of reasoning in
	other
	free description logic languages
	based on dual-domain semantics to
	the case of
	partial interpretations.
\end{abstract}

\section{Introduction}
\label{sec:intro}
Noun phrases that can be used to refer to a single object in a context are known in linguistics as \emph{referring expressions} (REs).
These include
both \emph{individual names},
such as
`KR 2021',
and \emph{definite descriptions},
such as
`the General Chair of KR 2021'~\cite{Nea90,Can93}.
Compared to individual names alone, REs provide increased flexibility in the description and the identification of objects, representing also a natural tool
to
transmit
this
kind of
information in a semantically transparent way.
In the context of
information and
knowledge base
(KB)
management
systems, REs have been proposed to address the problem of object identifiers that
remain
obscure to end-users, such as blank node identifiers in RDF or system-generated \texttt{ref} expressions in
object-oriented
databases~\cite{BorEtAl16,BorEtAl17}.

However, with the recent exception of the work by~\cite{NeuEtAl20}
discussed   below, most of the ontology languages considered in the literature have not included definite descriptions as first-class \emph{terms}, on a par with individual names.
%
To this goal, another feature of
REs has to be taken into account:
that of possibly \emph{failing to denote} any object at all.
For instance, `KR 2019' is a non-denoting individual name, since no KR conference took place in 2019, while
`the Program Chair of KR 2020'
and
`the banquet of KR 2020'
are non-denoting definite descriptions, because this conference had
two Program Chairs and no banquet
in 2020.
%
This is not easily captured
in
classical first-order logic (FO),
where
an individual name
is always assigned to an element of the domain by the interpretation function,
and
definite descriptions are not included
among the terms of the language~\cite{Rus05}.
Logics that allow for possibly non-denoting terms
are known as \emph{free logics}~\cite{Ben02,Leh02}.

In this work,
we introduce and study a family of description logic (DL) 
languages with both individual names and definite descriptions, that we call \emph{free DLs with definite descriptions}, or \emph{free DLs}, for short.
Syntactically,
they
extend the classical
ones
with nominals of the form $\{ \defdes C \}$, where $\defdes C$ is a term standing for the definite description `the object that is $C$' and $C$ is a concept.
We denote the resulting DLs
with an upperscript $\defdes$,
focussing in particular on
$\ALCOud$ and $\ELOud$, which are, respectively, $\ALC$ and $\EL$ with nominals, the universal role, and definite descriptions.
Their semantics is based on \emph{partial interpretations},
that generalise the classical ones by letting the interpretation function to be \emph{partial} on 
individual names, meaning that only a \emph{subset} of all the individual names has its elements assigned to objects of the domain.
Moreover, the extension of $\{ \defdes C \}$ in a partial interpretation 
coincides with that of the concept $C$, if $C$ is interpreted as a singleton, and it is empty otherwise.
%
%
%
Nominals involving definite descriptions can
be used to form concept inclusions (CIs)
with different satisfaction conditions.
E.g.,
\begin{align*}
\{ \defdes  \exists {\sf isPCof}. \{ {\sf dl20 }\} \} & \sqsubseteq \exists {\sf reportsTo}.\{\defdes \exists {\sf isGCof}.\{ {\sf dl20} \} \}
\end{align*}
states
that
whoever (if anyone) is the Program Chair of DL 2020 reports to the General Chair of DL 2020:
if there is exactly one
object in
$\exists {\sf isPCof}. \{ {\sf dl20 }\}$,
then
$\exists {\sf isGCof}.\{ {\sf dl20} \}$ is forced to have exactly one element as well, 
but this
CI
is
(vacuously)
satisfied
also in interpretations without, or with more than one,
object in $\exists {\sf isPCof}. \{ {\sf dl20 }\}$.
%

We show
that reasoning
in free DLs with definite descriptions
can be performed at no additional costs.
For
(extensions of)
$\ALCOud$, we
employ a polynomial time reduction
(via a translation that can be applied to other constructors as well)
to
the corresponding
language
without definite descriptions,
so that efficient off-the-shelf reasoners can be used.
%
Moreover, we show that entailment in $\ELOud$
ontologies
remains tractable, using a modified version of the algorithm for classical $\ELO$~\cite{BaaEtAl05}. 

We next characterise the expressive power of $\ALCOud$ concepts relative to
FO
on partial interpretations, using an appropriate notion of \emph{$\ALCOud$ bisimulations}. This result is of interest in its own right, but also serves as an important technical tool for the remainder of the article.


Having designed
a suitable DL language,
we further
consider the task of
constructing ontologies with definite descriptions.
As a step
in this direction,
we
study
the
problem of \emph{finding} meaningful REs for an individual
under a given ontology.
This
is
related to \emph{RE generation}
in natural language processing,
concerned with the
automatic production of such noun phrases,
possibly extracted from a non-linguistic source, e.g. a database~\cite{ReiDal00,KraEtAl03,KraDee12}.
%
Towards a better understanding of this problem,
we investigate the complexity of deciding the \emph{existence} of an RE for an individual within a given language and signature and with respect to an ontology.
%
%
%
The signature
allows
the user
to
specify
the features of interest for describing
an
object,
and
by deciding
this problem
it can be determined
whether alternative characterisations of an individual are available.
%
For example,
consider
the
following
$\ALCOud$
ontology $\Omc$, about KR events
held
between 2018 and 2020:
\begin{align*}
	\mathsf{ KRConf } \equiv \ & \{ \mathsf{ kr18, kr19, kr20 } \}, \\
	\mathsf{ KRWork } \equiv \ & \{ \mathsf{ dl18, dl19, dl20 } \}, \\
	\mathsf{ KREvent } \equiv \ & \mathsf{ KRConf \sqcup KRWork }, \\
\{ \mathsf{kr19} \} \sqsubseteq \bot, \ & \
\{ \mathsf{kr18} \} \equiv \exists \mathsf{hasRC}.\top, \\
\mathsf{ \{ dl18, kr18, dl19, kr20 \} } \equiv \ & \exists \mathsf{hasPCM}.\{ \defdes \exists \mathsf{isGCof}.\{ \mathsf{dl20} \} \}, \\
	\mathsf{ \{ kr20, dl20 \} } \equiv \ & \mathsf{ \exists hasLoc.VirtualLoc }.
\end{align*}
%
The first three CIs define, respectively, the concepts of KR Conference, Workshop and Event, while the
next
two
state, respectively, that `KR 2019' does not denote, and that KR 2018 is the one and only object that has a Registration Chair.
The
subsequent
CI asserts that the KR events having as PC Member the General Chair of DL 2020 are exactly DL 2018, KR 2018, DL 2019 and KR 2020. Finally, the last one expresses that the
objects
having a virtual location are exactly KR 2020 and DL 2020. 
Focussing on location-based characterisations of KR events,
the nominal $\{ \mathsf{ kr20 } \}$
has an
$\ELO$ RE in terms of the signature $\Sigma_{1} = \{ \mathsf{ KRConf, hasLoc, VirtualLoc } \}$,
since
\[
	\Omc \models \mathsf{ \{ kr20 \} } \equiv \mathsf{KRConf} \sqcap \exists \mathsf{hasLoc}.\mathsf{VirtualLoc},
\]
whereas there is no RE for $\{ \mathsf{ kr20 } \}$ in any language if we consider the signature $\Sigma_{2} = \{ \mathsf{ KREvent, hasLoc, VirtualLoc } \}$.
If we instead choose to refer to KR events in light of their organising members, we have that $\{ \mathsf{ kr20 } \}$ has no RE under
$\Sigma_{3} = \{ \mathsf{ KRConf, hasPCM, isGCof, dl20 } \}$,
while it
can be described
in $\ALCOud$ (but not in $\ELOud$)
in terms of
$\Sigma_{4} = \{ \mathsf{KRConf, hasRC, hasPCM, isGCof, dl20} \}$, since it is equivalent under $\Omc$ to the $\ALCOud$ concept
\[
	\mathsf{KRConf} \sqcap
	\exists \mathsf{hasPCM}.\{ \defdes \exists \mathsf{isGCof}.\{ \mathsf{dl20} \} \} \sqcap \lnot \exists \mathsf{hasRC}.\top.
\]
%

We show that deciding the existence of REs is $\TwoExpTime$-complete if the ontology and the RE are both in $\ALCOud$.
The problem is in \PTime{} if the ontology and the RE are both in $\ELOud$, under the additional assumption that the individual name one aims to describe denotes in every model of the ontology. Without this assumption, the complexity remains open. If
FO
expressions are allowed as REs,
the first problem becomes $\ExpTime$-complete and the latter is still in $\PTime$, because of the projective Beth definability property~\cite{Bet56,ChaKei90} of
FO,
also on partial interpretations,
and because reasoning in $\ALCOud$ and $\ELOud$ are $\ExpTime$- and, respectively, $\PTime$-complete.
In this case, for $\ELOud$, no restriction is needed regarding the denotation
of the individual name.
If instead
the ontology is in $\ALCOud$ (even $\ALCO$), but one asks for
an $\ELOud$ RE,
the problem becomes undecidable.
%
%

Finally, we compare our
framework
to
the \emph{positive} and \emph{negative semantics}
for
free DLs
proposed by~\cite{NeuEtAl20}
based on \emph{dual-domain interpretations}.
For both these semantics, we provide a polynomial time reduction to reasoning in $\ALCOud$ on partial interpretations.

\section{Free Description Logics}
\label{sec:freedl}

%
We introduce basic notions for free DLs
(with definite descriptions)
by presenting the syntax and semantics of  $\ALCOud$,
which we define as a free DL based on the classical
$\ALCOu$~\cite{BaaEtAl03a}, and other related languages.

\subsection{Syntax}

Let \NC, \NR and \NI be countably infinite and pairwise disjoint sets of \emph{concept names}, \emph{role names}, and \emph{individual names}, respectively.
The $\ALCOud$ \emph{terms} $\tau$ and \emph{concepts} $C$ are constructed by mutual induction as follows:
\[
	\tau ::= a \mid \defdes C,
	\
	C ::= A \mid
	\{ \tau \} \mid
	\lnot C \mid (C \sqcap C) \mid \exists r.C
	\mid
	\exists u.C,
\]
where $a\in \NI$, $A \in \NC$, $r \in \NR$, and $u$ is the \emph{universal role}.
%
A term of the form $\defdes C$ is called a \emph{definite description}, with the concept $C$ being the \emph{body of $\defdes C$}, and a concept $\{ \tau \}$ is called a \emph{(term) nominal}.
An \emph{$\ALCOud$ axiom} is either an $\ALCOud$ \emph{concept
  inclusion} (\emph{CI}) of the form $(C \sqsubseteq D)$ or an $\ALCOud$
\emph{assertion} of the form
$C(\tau)$
or
$r(\tau_1,\tau_2)$, where $C, D$ are 
concepts,
$r \in \NR$, and $\tau, \tau_1, \tau_2$ are
terms.
An \emph{$\ALCOud$ ontology $\Omc$} is a finite set of CIs and assertions.


%
All the usual syntactic abbreviations and conventions are assumed. 
In particular,
for concepts, we set
$\bot = A \sqcap \lnot A$,
$\top = \lnot \bot$, $C \sqcup D = \lnot (\lnot C \sqcap \lnot D)$,
$C \Rightarrow D = \lnot C \sqcup D$, and
$\forall s. C = \lnot \exists s. \lnot C$, with $s \in \NR \cup \{ u \}$,
while
a \emph{concept equivalence} (\emph{CE}) $C \equiv D$ abbreviates $C \sqsubseteq D, D \sqsubseteq C$.

In the rest of this paper, we will consider other DL languages with nominals, that we introduce briefly here.
%
We define the classical $\ALCO$ as $\ALCOud$ without neither definite descriptions nor the universal role, while $\ALCOd$ and $\ALCOu$ are defined as $\ALCO$ with the addition of either definite descriptions or the universal role, respectively.
%
Moreover, the language $\ELOud$ is obtained from $\ALCOud$ by allowing only for $\bot$, $\top$ 
(considered now as primitive logical symbols), concept names, term nominals, conjunctions
and existential restrictions.
%
Finally, $\ELO$, $\ELOd$ and $\ELOu$ are similarly defined sublanguages of $\ELOud$.
%
%

Given a DL
$\Lmc$,
the \emph{signature} of an $\Lmc$ ontology $\Omc$, $\sig{\Omc}$, is the set of all concept, role and individual names occurring in $\Omc$,
while $\con{\Omc}$ is the set of all subconcepts occurring in $\Omc$.
For a \emph{signature} $\Sigma \subseteq \NC \cup \NR \cup \NI$, an \emph{$\Lmc(\Sigma)$ ontology} 
$\Omc$ is an $\Lmc$ ontology such that $\sig{\Omc} \subseteq \Sigma$ (analogous notions are given 
for $\Lmc$ concepts, where in particular $\con{C}$ is the set of subconcepts occurring in $C$).
\subsection{Semantics}
\label{sec:semantics}


For the DL languages with nominals considered in this work,
we generalise
their semantics
through the notion of partial interpretation.
%
A \emph{partial interpretation} 
is a pair
$\Int = (\Delta^{\Int}, \cdot^{\Int})$,
where
$\Delta^{\Int}$ is a non-empty set, called \emph{domain} of $\Int$,
and $\cdot^{\Int}$ is a function that maps
every $A \in \NC$ to a subset of $\Delta^{\Imc}$,
every $r\in\NR$ to a subset of $\Delta^{\Imc} \times \Delta^{\Imc}$,
the universal role $u$ to the set $\Delta^{\Imc} \times \Delta^{\Imc}$ itself,
and every $a$ in
a
\emph{subset} of $\NI$
to an element in $\Delta^{\Imc}$.
In other words, 
$\cdot^{\Imc}$ is a total function on $\NC \cup \NR$ and a partial function on $\NI$.
A \emph{total interpretation} is a partial interpretation
$\Int = (\Delta^{\Int}, \cdot^{\Int})$
in which $\cdot^{\Int}$ is also total on $\NI$. 
%
The \emph{value} $\tau^{\Int}$ of a term $\tau$  in $\Int$
and the \emph{extension} $C^{\Int}$ of a concept $C$ in $\Int$ are defined by mutual induction:
\begin{gather*}
		(\defdes C)^{\Int}  =
			\begin{cases}
				d, & \text{if} \ C^{\Int} = \{ d \}, \ \text{for some} \ d \in \Delta^{\Imc}; \\
				\text{undefined}, & \text{otherwise}.
			\end{cases}
\end{gather*}
We say that $\tau$ \emph{denotes}
in $\Int$ iff $\tau^{\Int} = d$, for a $d \in \Delta^{\Imc}$.
Thus, in particular, an individual name $a$ denotes in $\Imc$ iff $a^{\Imc}$ is defined.
In addition, where $s \in \NR \cup \{ u \}$:
	\begin{gather*}
		(\neg C)^{\Int} = \Delta^{\Imc} \setminus C^{\Int}, \qquad 
		(C \sqcap D)^{\Int} = C^{\Int} \cap D^{\Int}, \\
		(\exists s.C)^{\Int} = \{d \in \Delta^{\Imc} \mid \exists  e \in C^{\Int}: (d,e) \in s^{\Int}\}.
	\end{gather*}
Moreover,
we set
$\{ \tau \}^{\Int} = \{ \tau^\Int \}$, if $\tau$ denotes in $\Int$, and $\{ \tau \}^{\Int} =\emptyset$, otherwise.
%
%
%

A concept $C$ is \emph{satisfied in $\Int$} iff $C^{\Int} 
\neq \eset$, and it is \emph{satisfiable} iff there is
a
partial interpretation in which it is satisfied.
Given an axiom $\alpha \in \Omc$, the \emph{satisfaction of $\alpha$ in $\Int$}, written
$\Int \models \alpha$, is defined as follows:
\begin{gather*}
	\Int \models C\sqsubseteq D \  \text{iff} \  C^{\Int} \subseteq D^{\Int}, \\
	\Int  \models C(\tau)  \ \text{iff}  \ \tau \text{ denotes in $\Int$ and } \tau^{\Int} \in C^{\Int}, \\
	\Int \models r(\tau_{1}, \tau_{2})  \ \text{iff}  \  \tau_1, \tau_2 \text{ denote in $\Int$ and } (\tau_1^{\Int}, \tau_2^{\Int}) \in r^{\Int}.
\end{gather*} 
%
We say that $\Omc$ is \emph{satisfied} in a partial interpretation $\Int$ 
(or that $\Int$ \emph{satisfies}, or is a \emph{model} of, $\Omc$), written $\Imc \models \Omc$, 
iff $\Int \models \alpha$, for every $\alpha \in \Omc$,
and it is \emph{satisfiable} iff
it is satisfied in some partial interpretation.
A concept $C$ is \emph{satisfiable w.r.t. an ontology} \Omc
if both $C$ and \Omc are satisfied in some partial interpretation.
%
Moreover, $\Omc$ \emph{entails}
an axiom $\alpha$, written $\Omc \models \alpha$, if every partial interpretation that satisfies $\Omc$ satisfies also $\alpha$.
%
%
Finally,
we say that an ontology $\Omc'$ is a \emph{conservative extension} of an ontology $\Omc$ if every model of $\Omc'$ is a model of $\Omc$, and every model of $\Omc$ can be turned into a model of $\Omc'$ by modifying the interpretation of symbols in $\sig{\Omc'} \setminus \sig{\Omc}$, while keeping fixed the interpretation of symbols in $\sig{\Omc}$. 
We also consider these
notions
for total
(that is, classical)
interpretations
and write  `on total interpretations' explicitly
whenever this is the case. 
\subsection{Basic Properties}
\label{sec:observations}
We discuss some properties of free DLs,
where
$\Lmc \in \{ \ALCOud, \ELOud \}$
 in the following.
%

%
%
%
%
$\textbf{(1)}$ 
An $\Lmc$ term $\tau$ denotes in a partial
interpretation $\Int$ iff
$\Int \models \top \sqsubseteq \exists u.\{ \tau \}$.
Furthermore, an $\Lmc$ ontology $\Omc$ entails
$\top \sqsubseteq \exists u.\{ \tau \}$ iff $\Omc \models \top (\tau)$ and
this happens iff
$\tau$ denotes in all the partial interpretations that are
models of $\Omc$.
We say that an individual name $a$ \emph{denotes w.r.t.~an ontology} 
$\Omc$ if $\Omc\models \top \sqsubseteq \exists u.\{a\}$.
%
%
By adding such a CI to an ontology for each individual name 
occurring in it,
we immediately obtain
that
the $\Lmc$ ontology satisfiability
and entailment problems
on total interpretations can
be reduced in polynomial time to the
corresponding problems
on partial
interpretations.
%
The converse polynomial time reduction 
can be defined by substituting every individual name $a$ 
occurring in an ontology or axiom with a fresh concept name $B_{a}$, 
and adding the CI $B_{a} \sqsubseteq \{ b_{a} \}$, 
for a fresh individual name $b_{a}$.




\textbf{(2)}
On partial interpretations, an $\Lmc$
assertion $C(\tau)$ is \emph{not} equivalent to
$\{\tau\}\sqsubseteq C$.  Indeed, while
terms
occurring in assertions are forced by the
semantics to always denote,
the CI $\{\tau\}\sqsubseteq C$ is satisfied in any partial
interpretation where $\tau$ is not denoting.  Nevertheless, assertions are just syntactic sugar. 
One can replace
\begin{itemize}
\item $C(\tau)$ by 
$\{\tau\}\sqsubseteq C$, $\top \sqsubseteq \exists u.\{\tau\}$; and 
\item $r(\tau_1,\tau_2)$ by 
$\{\tau_1\}\sqsubseteq \exists r.\{\tau_2\}$, $\top \sqsubseteq \exists u.\{\tau_{1}\}$.
\end{itemize}
This
encoding
yields an equivalent ontology. 
Thus, from
now on, we may assume
w.l.o.g.
that $\Lmc$ ontologies do not contain assertions.
\textbf{(3)}
For every $\Lmc$ ontology $\Omc$, concept $C$ and term $\tau$, we have that $\Omc \models \{ \tau \} \equiv C$ implies $\Omc \models \{ \tau \} \equiv \{ \defdes C \}$.
In $\ELOud$, under the assumption that
an $\ELOud$ term $\tau$
denotes in every model of an
$\ELOud$
ontology $\Omc$, we also have the following, for every $\ELOud$ concept $C$:
if $\Omc \models \{ \tau \} \equiv C$,
then $\Omc \models \{ \tau \} \equiv C'$, where $C'$ is obtained from $C$ by substituting every
$\{ \defdes D \}$ occurring in $C$ with the concept $D$.

\textbf{(4)}
Given an $\Lmc$ ontology $\Omc$,
we can obtain a conservative extension $\overline{\Omc}$ of $\Omc$ in \emph{flattened form}, that is, such that all occurrences of definite descriptions in $\overline{\Omc}$ are of the form $\defdes B$, where $B$ is a concept name.
Indeed, let $\defdes C_1, \ldots, \defdes C_n$ be all the definite descriptions in $\Omc$ that do not occur in the body of another definite description $\defdes C'$.
%
We define
$\overline{\Omc}$ as
\[
	\Omc'
	\cup \bigcup_{1 \leq i \leq n} \overline{\{ B_{C_{i}} \equiv C_i\}},
\]
where $\Omc'$ is obtained from $\Omc$ by substituting the bodies $C_1, \ldots, C_n$ of $\defdes C_1, \ldots, \defdes C_n$ with fresh concept names $B_{C_{1}}, \ldots, B_{C_{n}}$, respectively, and $\overline{ \{B_{C_{i}} \equiv C_{i}\} }$ is the ontology obtained by recursively applying the procedure just described to the ontology $\{ B_{C_{i}} \equiv C_{i} \}$.
\section{Reasoning in Free DLs}
\label{sec:complexity}

We
study
the complexity of reasoning in $\ALCOud$ and 
$\ELOud$.

\subsection{Satisfiability in $\ALCOud$}
\label{sec:sat}





We prove that satisfiability in $\ALCOud$ is \ExpTime-complete.
To show this result, we provide a polynomial size equisatisfiable translation
into $\ALCOu$. 
%

An
$\ALCOud$ ontology $\Omc$
is in \emph{normal form} if it is in flattened form and 
all the CIs
in $\Omc$ are either of the form
$E \sqsubseteq F$, where $E, F$ are
$\ALC_{u}$  (i.e., $\ALC$ with the universal role) concepts,
or
$\{ \tau \} \sqsubseteq A$, or
$A \sqsubseteq \{ \tau \}$,
with $A\in\NC$.
It can be seen that an $\ALCOud$ ontology can be transformed in polynomial time into an $\ALCOud$ ontology in normal form that is a conservative extension of the original.

We now define a translation
of an $\ALCOud$ ontology $\Omc$ in normal form into an $\ALCOu$ ontology $\Omc\red$. While the translation
preserves
symbols in
$\NC \cup \NR$, nominals
$\{\tau\}$ are translated as follows:
%
\begin{align*}
  \{\tau\}\red = \{\tau\}^{+} \sqcap \concleqone,
\end{align*}
where: $\{ \tau \}^{+} = A_{b}$, with fresh $A_{b} \in \NC$,
if $\tau = b \in \NI$; $\{ \tau \}^{+} = B$, if $\tau = \defdes B$, with $B \in \NC$;
and $\concleqone$ stands for the concept
%
  $\forall u. (\{ \tau \}^+ \Rightarrow \{ a_{\tau} \})$,
%
with fresh $a_{\tau} \in \NI$.
%
We now define  
%
\[
(E \sqsubseteq F)\red  = E \sqsubseteq F,
\]
\[
  (\{ \tau \} \sqsubseteq A)\red  = \{ \tau \}\red \sqsubseteq A,
  \quad
  (A \sqsubseteq \{ \tau \})\red  = A \sqsubseteq \{ \tau \}\red,
\]
%
where $E, F$ are $\ALC_{u}$ (i.e., $\ALC$ with the universal role) concepts, and $A$ is a concept name.
Finally, we set $\Omc\red$ as
\[
\bigcup_{C \sqsubseteq D \in \Omc}
\!\!\!\!
\{ (C \sqsubseteq D)\red \}
\
\cup
\!\!\!\!\!\!
\bigcup_{\{\tau\} \in \con{\Omc}}
\!\!\!\!\!\!\!
\{  \{\tau\}^+ \sqsubseteq
  \forall u. (\{ a_{\tau} \} \Rightarrow \{ \tau \}^+) \}.
\]
%
%
We then obtain the following.
\begin{lemma}\label{prop:alcodtoalcou}
An $\ALCOud$ ontology $\Omc$ in normal form is satisfiable iff the $\ALCOu$ ontology $\Omc\red$ is satisfiable on total interpretations.
\end{lemma}

It follows from a 
result in Propositional Dynamic Logic 
extended with nominals and the universal modality~\cite[Corollary 7.7]{PasTin91} that 
the $\ALCOu$ ontology satisfiability problem on total interpretations is in $\ExpTime$. The matching lower bound comes from the $\ALC$ ontology satisfiability problem on total interpretations~\cite{GabEtAl03}.
Since the $\ALCOud$ ontology satisfiability
problem on total interpretations
is
reducible 
in polynomial time
%
to
its counterpart
on partial interpretations 
(cf. Point~(1) in Section~\ref{sec:observations}),
the following holds.

\begin{theorem}
$\ALCOud$ ontology satisfiability 
(both on partial and total interpretations) 
is $\ExpTime$-complete.
\end{theorem}

The reduction we presented can be easily adapted to deal with  more expressive DLs, e.g. 
extensions of $\ALCOud$ with inverse roles and number restrictions.
\subsection{Reasoning in \ELOud}\label{subsec:elo}


We prove that 
entailment in $\ELOud$ ontologies is \PTime-complete.
%
To show this result, we assume
w.l.o.g.
that the
assertions 
are encoded within the CIs in the ontology
(cf. Point~(2) in Section~\ref{sec:observations})
and adapt the completion
algorithm for \ELO ontologies~\cite{BaaEtAl05}.  The main idea is to
add a copy of each concept name in an ontology and remove it only if
its extension is a singleton in any model.
%
Even though $\ELOud$ admits a mild form of disjunction ($\{\defdes A\} \sqsubseteq B$ states 
that the extension of $A$ contains at least two elements \emph{or} $A\sqsubseteq B$),  
the logic remains `Horn'
in the sense that (if
an ontology is satisfiable, then) 
minimal models exist. 

Any \ELOud ontology can be converted
in polynomial time into a conservative extension in \emph{normal form}, that is, 
an $\ELOud$ ontology $\Omc$ in flattened form 
where all CIs
have one of the
following forms:
\begin{equation*}
C_1\sqcap C_2\sqsubseteq D, \
\exists r.C\sqsubseteq D, \
C \sqsubseteq \exists r.D, \
\{\tau\}\sqsubseteq D, \
C\sqsubseteq \{\tau\},
%
\end{equation*}
where $C_{(i)}\in \NC\cup\{\top\}$,
$D\in  \NC\cup\{\top,\bot\}$ and all terms $\tau$
in \Omc are either of the form $\{a\}$, with $a\in\NI$, or of the form
$\{\iota A\}$, with $A\in\NC$. 

Let \Omc be in normal form. 
We denote by $\Bmc\Cmc_\Omc$ the union of
$\{\top\}$, the set of all concept names occurring in \Omc, and the
set of all concepts $\{a\}\in \con{\Omc}$. 
Also,
we denote by $\Bmc\Cmc^{+}_\Omc$ the union of $\Bmc\Cmc_\Omc$
with
$\{\bot\}\cup \{\{\iota A\}\mid \{\iota A\} \in \con{\Omc} \}$ 
and by $\Rmc_\Omc$ the set including $u$ and the role names occurring in \Omc.  
%
%
%
%
Given $A,B\in\NC$, we may write $A \sqsubseteq B$ instead of
$A \sqcap A\sqsubseteq B$.
If $\{\iota A\} \in \con{\Omc}$,
we
assume
w.l.o.g.
that $\{\iota A\}\sqsubseteq A\in\Omc$.
Moreover, we write ${\sf A}$ to denote
a concept name which we aim at checking whether
$\Omc\models {\sf A}\sqsubseteq B$ (see Lemma~\ref{lem:completion} and
Table~\ref{tab:rules}).
%
The \emph{classification graph} for \Omc and    ${\mathsf A}$ is a tuple
$(V, S,R)$ where
\begin{itemize}
\item $V=\Bmc\Cmc_\Omc\cup\{A^c\mid A\in(\Bmc\Cmc_\Omc\cap\NC)\}$,
  with each $A^c\in\NC$ fresh; 
\item $S$ is a function mapping nodes in $V$ to subsets of
  $\Bmc\Cmc^{+}_\Omc$; 
\item $R$ is a function mapping edges in $V\times V$ to (possibly
  empty) subsets of $\Rmc_\Omc$, 
  where $r$ is in $\Rmc_\Omc$.
\end{itemize}
%
%
Intuitively, a concept name of the form $A^c$ represents a second element in the
extension of $A$, and it is removed from the classification graph if $A$ has at most
one object in its extension.
Initially, we set $S(C) := \{C,\top\}$, for all nodes $C \in V$, and
$R(C,D) := \emptyset$, for all edges $(C, D) \in (V \times V)$.
%
If $C\in V\setminus \Bmc\Cmc_\Omc$ is of the form $A^c$, with
$A\in\NC$, then we add $A$ to $S(A^c)$. 
Given 
$C,D\in \Bmc\Cmc_\Omc$, we
write $C\leadsto_R D$ iff there are $C_1,\ldots, C_k\in\Bmc\Cmc_\Omc$
such that $C_1\in S(C)$; $r\in R(C_j,C_{j+1})$, for some
$r\in \Rmc_\Omc$, for all $1\leq j < k$; $D\in S(C_k)$.
The completion rules are given in Table~\ref{tab:rules}.  Assume that
rules are only applied if $S$ or $R$ or $V$ change after the rule
application. This bounds the number of rule applications to a
polynomial in the number of concept and role names in $\Omc$. Thus, the
resulting \emph{completed classification graph} for
\Omc can be constructed in polynomial time
with respect to
the size of \Omc.

\begin{table*}[tb]
\centering
\begin{tabular*}{0.95\textwidth}{l@{}l@{}l@{}}
\toprule
 $\quad$ \textbf{if}   & $\quad\quad$ & \textbf{then}  
\\
\hline
\hline
${\sf R}_1$:   $C\sqcap D\sqsubseteq B\in\Omc$, $C,D\in S(E)$ &   &  add $B$ to $S(E)$
\\
${\sf R}_2$:   ${\sf A}\leadsto_R E$,  
$C\sqsubseteq \exists r.D\in\Omc$, $C\in S(E)$ &  &  add $r$ to $R(E,D)$, $R(E^{c},D)$ 
\\
${\sf R}_3$:  
$ \exists r.C \sqsubseteq D \in\Omc$, $C\in S(B), \ r\in R(E,B)$& &  add $D$ to $S(E)$
\\
${\sf R}'_3$:  $\exists u.C \sqsubseteq D \in\Omc$, ${\sf A}\leadsto_R
  C$,  $E\in V$& &  add $D$ to $S(E)$
  \\
${\sf R}_4$:   $ \{\tau\}\in S(E)\cap S(D)$, ${\sf  A} \leadsto_R D$;
&  &   $S(E):=S(E)\cup S(D)$ 
\\
${\sf R}_5$:   $r \in R(E,D)$, $\bot\in S(D)$ 
&   & add $\bot$ to $S(E)$
\\
${\sf R}_6$:    $\{\tau\}\sqsubseteq D\in\Omc$, $\{\tau\}\in S(E)$ 
&    & add $D$ to $S(E)$
\\
${\sf R}_7$:   $C\sqsubseteq \{\tau\}\in\Omc$, $C\in S(E)$
&    & add $\{\tau\}$ to $S(E)$
\\
${\sf R}_8$:    $\{\tau\}\in S(B)$, $B\in\NC$
&    & $V:=V\setminus \{B^c\}$
\\
${\sf R}_{9}$:    $B\in S(E)$, $B^c\not\in V$ 
&    & add $\{\iota B\}$ to $S(E)$
\\
${\sf R}_{10}$:    ${\sf A}\leadsto_R C$, 
$\{\iota B\}\in S(C)$
&  &   $V:=V\setminus \{B^c\}$ %
\\
\bottomrule
\end{tabular*}
\caption{Completion rules for subsumption in \ELOud with respect to ontologies.
}
\label{tab:rules}
\end{table*}

%
%

%
%
\begin{restatable}{lemma}{Eloentailment}\label{lem:completion}
  Given an \ELOud ontology \Omc in normal form, let $S$ be the node function of a completed
  classification graph for \Omc
  (cf. rules in Table~\ref{tab:rules}), ${\sf A}\in \NC$  and $B\in
  \Bmc\Cmc _\Omc \cup \{\bot\}$.
   Then, $\Omc\models {\sf A}\sqsubseteq  B$
  iff $S({\sf A})\cap \{B,\bot\}\neq \emptyset$.
\end{restatable}

Thanks to Lemma~\ref{lem:completion}, given arbitrary \ELOud concepts
$C,D$ and an \ELOud ontology \Omc, one can decide in polynomial time
whether $C\sqsubseteq D$ is entailed by \Omc by adding ${\sf A}\equiv
C$ and $B\equiv D$ to \Omc,
converting it in normal form,
and then checking whether $S({\sf A})\cap
\{B,\bot\}\neq \emptyset$, where ${\sf A},B$ are fresh concept names.
As 
an immediate consequence of 
Lemma~\ref{lem:completion} and the polynomial size of a completed classification graph we
obtain the following complexity result.
\begin{theorem}\label{th-eloud-sat}
  Entailment in $\ELOud$ (both on partial and total interpretations)
  is \PTime-complete.
\end{theorem}

The completed  
classification graph can be used to define
a polynomial size canonical model for an $\ELOud$ ontology (if
it is satisfiable).
%
Let $\Omc$ be in normal form 
and let $(V, S,R)$ be the
completed classification graph for \Omc and a concept name
${\mathsf A}$.
Consider the following sets:
    %
    \begin{align*}
      \Rmc_{{\mathsf A}} & = \{C\in \Bmc\Cmc_\Omc \mid {{\mathsf A}}\leadsto_R C\},\\
      \Rmc_{{\mathsf A}}^c & = \{C^c\in V\mid C\in \Rmc_{{\mathsf A}} \cap\NC\},
    \end{align*}
    over which we define the relation $\thicksim$, where 
  \begin{center}
  $C \thicksim D$ iff $C=D$ or $\{\tau\}\in S(C)\cap S(D)$, for
    some   $\tau$.
  \end{center}
  %
  %
  It can be seen that $\thicksim$ is an equivalence relation, whose
  equivalence classes are denoted by $\eq{C}$.
Assume  
$\bot\notin S({\mathsf A})$ (otherwise, by  Lemma~\ref{lem:completion}, no model of $\Omc$ and ${\mathsf A}$ exists). 
The \emph{polynomial size canonical model of $\Omc$ and ${\mathsf A}$} is the
partial interpretation $\fincanmod =(\Delta^{\fincanmod},
\cdot^{\fincanmod})$ such that: 
  \begin{itemize}
  \item $\Delta^{\fincanmod} =\{\eq{C} \mid C\in\Rmc_{{\mathsf A}} \cup \Rmc_{{\mathsf A}}^c \}$;
  \item $r^{\fincanmod} =\{(\eq{C},\eq{D}) \in \Delta^{\fincanmod} \times
    \Delta^{\fincanmod} \mid \exists D'\in [D]. r\in R(C,D')\}$, for all
    $r\in\NR\cap \Rmc_\Omc$;
  \item $D^{\fincanmod} = \{\eq{C}\in \Delta^{\fincanmod} \mid D\in S(C)\}$, for all $D\in \NC$;
  \item $a^{\fincanmod} =\eq{C}$, for some $C\in \Rmc_{{\mathsf A}}$, if $\{a\}\in S(C)$, for all $a\in\NI$.
    %
  \end{itemize}
  %


We are now ready to state the main property of the canonical
model, used in Section~\ref{sec:bdp}.
\begin{restatable}
  {theorem}{eloufincanmod}\label{theo:eloufincanmod} Let $\Omc$ be an $\ELOud$
  ontology in normal form and ${\mathsf A}$ a concept name satisfiable w.r.t. \Omc. Then $[{\mathsf A}] \in  C^{\fincanmod}$ iff
  $\Omc \models {\mathsf A} \sqsubseteq C$, for every $\ELOu$ concept $C$.
\end{restatable} 
We note that the equivalence is stated for $\ELOu$ concepts and not for
$\ELOud$ concepts. In fact, it is an interesting open problem whether
polynomial size canonical models exist that satisfy the equivalence for
$\ELOud$ concepts. 

\section{Bisimulations and Expressive Power}
\label{sec:alcoiotabisim}

Here we discuss the expressive power of free DLs. 
In particular,
we define a notion of bisimulation for \ALCOud
that we use
to characterise the expressive
power of concepts relative to
FO
formulas interpreted on partial
interpretations. 
The definitions are standard in the literature~\cite{AreEtAl01,Ten05,LutEtAl11}, but have to be adapted to partial interpretations and definite descriptions.

Let $\Int$ and $\Jmc$ be partial interpretations, and let
$\Sigma$
be a signature.
An \emph{$\ALCO(\Sigma)$ bisimulation} between $\Int$ and $\Jmc$
is a 
relation
$Z \subseteq \Delta^{\Int} \times \Delta^{\Jmc}$ such
that, for every $d\in \Delta^{\Int}$ and $e\in\Delta^{\Jmc}$ with $(d, e) \in Z$,
every concept name or nominal $X$ formulated within $\Sigma$, 
and every role name $r$ in $\Sigma$: 
%
(\textit{atom})
$d \in X^{\Imc}$ iff $e \in X^{\Jmc}$;
%
%
%
(\textit{forth})
if $(d, d') \in r^{\Int}$ then there is 
$e' \in \Delta^{\Jmc}$ such that $(e, e') \in r^{\Jmc}$ and
$(d', e') \in Z$; and 
%
(\textit{back})
if $(e, e') \in r^{\Jmc}$ then there is
$d' \in \Delta^{\Int}$ such that $(d, d') \in r^{\Int}$ and
$(d', e') \in Z$. 
For \emph{pointed partial interpretations} $(\Imc, d)$ and $ (\Jmc, e)$, we say that
$(\Imc, d)$ is \emph{$\ALCO(\Sigma)$ bisimilar} to $(\Jmc, e)$ and write
$(\Int, d) \sim^{\ALCO}_{\Sigma} (\Jmc, e)$
if there is an
$\ALCO(\Sigma)$ bisimulation $Z$ between $\Int$ and $\Jmc$
such that $(d, e) \in Z$.
$\ALCO(\Sigma)$ bisimulations characterise the expressive power of $\ALCO(\Sigma)$ concepts in the sense that an FO formula $\varphi$ is
preserved under $\ALCO(\Sigma)$ bisimulations iff it is equivalent to
an $\ALCO(\Sigma)$ concept. To characterise $\ALCOud(\Sigma)$ we
add a condition that reflects its ability to count up to one and also add
totality conditions
that reflect the addition of the universal role.

An \emph{$\ALCOud(\Sigma)$ bisimulation} $Z$ between \Imc and \Jmc is an 
$\ALCO(\Sigma)$ bisimulation
that is \emph{total}, meaning that $\Delta^\Imc$ and $\Delta^{\Jmc}$ are the domain and range of the relation, and that satisfies, for all $(d,e)\in Z$:
\begin{enumerate}
[label=$(\defdes.\arabic*)$, align=left, leftmargin=*, labelsep=0pt]
	\item[($\defdes$)] there exists $d' \in \Delta^{\Imc}$ such that $d \neq d'$ and $(\Imc, d) \sim^{\ALCO}_{\Sigma} (\Imc, d')$ iff there exists $e' \in \Delta^{\Jmc}$ such that $e' \neq e$ and $(\Jmc, e) \sim^{\ALCO}_{\Sigma} (\Jmc, e')$.
\end{enumerate}
We write $(\Imc, d) \sim^{\ALCOud}_{\Sigma} (\Jmc, e)$ if there exists
an $\ALCOud(\Sigma)$ bisimulation $Z$ between $\Imc$ and $\Jmc$ containing $(d,e)$, and we write $(\Imc, d) \equiv^{\ALCOud}_{\Sigma} (\Jmc, e)$ if
$d\in C^{\Imc}$ iff $e\in C^{\Jmc}$, for all $\ALCOud(\Sigma)$ concepts $C$. The definition of $\omega$-saturated partial interpretation is the obvious generalisation of that one for total interpretations \cite{ChaKei90}, and it is given in the appendix.
 

\begin{restatable}{theorem}{alcoiotabisimtoequiv}\label{thm:alcoiotabisimtoequiv}
	\label{thm:equivomegatobisim}
	For all signatures $\Sigma$ and all pointed partial interpretations $(\Imc,d)$ and $(\Jmc,e)$,
	\begin{enumerate}
		\item
		if
		$(\Imc, d) \sim^{\ALCOud}_{\Sigma} (\Jmc, e)$,
		then
		$(\Imc, d) \equiv^{\ALCOud}_{\Sigma} (\Jmc, e)$;
		\item
		if
		$(\Imc, d) \equiv^{\ALCOd}_{\Sigma} (\Jmc, e)$ and $\Int, \Jmc$ are $\omega$-saturated, then $(\Int, d) \sim^{\ALCOd}_{\Sigma} (\Jmc, e)$.
	\end{enumerate}
\end{restatable}
The following example illustrates how bisimulations can be used to prove the inexpressibility of certain concepts.

\begin{example}
	The concept $C=\exists u.\{\defdes A\}$ states that the extension of $A$ has cardinality one,
	that is
	$d\in C^{\Imc}$ iff $|A^{\Imc} |=1$, for any interpretation $\Imc$ and $d\in \Delta^{\Imc}$.
	The concept $D = \exists u.(A \sqcap \lnot \{ \defdes A \})$
	is such that
	$d\in D^{\Imc}$ iff $|A^{\Imc} | \geq 2$,
	hence it states that $A$ has cardinality greater 
	or equal to two.
	However,
	there
	is no
	$\ALCOud$ concept stating that the extension of $A$ has cardinality two: the pointed interpretations $(\Imc,d)$ and $(\Jmc,e)$ depicted below are $\ALCOud(\{A\})$-bisimilar (witnessed by $Z=\Delta^{\Imc}\times \Delta^{\Jmc}$), but $|A^{\Imc}|=2<|A^{\Jmc}|$.
%
\end{example}  
%
%
\begin{figure}[h]
\centering
\begin{tikzpicture}
\tikzset{
dot/.style = {draw, fill=black, circle, inner sep=0pt, outer sep=0pt, minimum size=2pt}
}

\draw (-1.75,-0.15) node[label=north:$\mathcal{I}$] {};
\draw (0,0) node[dot, label=east:$A$] {};
\draw (-1,0) node[dot, label=east:$A$, label=west:$d$] {};

\draw  (0.75,0.25) -- (0.75,-0.25);

\draw (4.25,-0.15) node[label=north:$\mathcal{J}$] {};
\draw (1.35,0) node[dot, label=east:$A$,label=west:$e$] {};
\draw (2.35,0) node[dot, label=east:$A$] {};
\draw (3.35,0) node[dot, label=east:$A$] {};
\end{tikzpicture}
%
\label{fig:bisim}
\end{figure}

We next state that \ALCOud is the fragment of
FO
on partial interpretations that is invariant under $\ALCOud$-bisimulations.
The \emph{standard translation} of an $\ALCOud$ concept $C$ into an FO formula $\sttr{x}{C}$ with one free variable $x$ is defined as expected,
where
for nominals of the form
$\{\iota C\}$ we
set:
\begin{align*}
	\sttr{x}{\{ \defdes C \}} = & \ \exists x \sttr{x}{C} \land \forall x \forall y (\sttr{x}{C} \land \sttr{y}{C} \to x = y) \\
 & \land \forall y (\sttr{y}{C} \to x = y).
\end{align*}
An
FO
formula $\p(x)$ is
\emph{invariant under $\sim^{\ALCOud}_{\Sigma}$}
 iff, for every $(\Imc, d)$ and $(\Jmc, e)$ such that $(\Imc, d) \sim^{\ALCOud}_{\Sigma} (\Jmc, e)$, we have $\Imc \models \p(d)$ iff $\Jmc \models \p(e)$.
%
%

\begin{restatable}{theorem}{alcofirstorder}\label{thm:alcofirstorder}
Let $\Sigma$ be a signature, and let $\p(x)$
a first-order formula such that $\sig{\p(x)} \subseteq \Sigma$.
The following conditions are equivalent: 
\begin{enumerate}
	\item there exists an $\ALCOud(\Sigma)$ concept $C$ such that $\sttr{x}{C}$ is logically equivalent to $\p(x)$; 
	\item $\p(x)$ is invariant under $\sim^{\ALCOud}_{\Sigma}$.
\end{enumerate}
\end{restatable}

%

We next consider $\ELOud$. In contrast to $\ALCOud$, we do not have a 
 model-characterisation of $\ELOud$ that generalises 
 the one for $\EL$~\cite{LutWol10,LutEtAl11}.  
   The fundamental problem is to constrain simulations (the basic notion used to characterise $\EL$) in such a way that they capture the expressivity of $\ELOud$ concepts.

To obtain
preliminary
results on
$\ELOud$
REs
in the next section, we remind the reader of the standard simulations between interpretations and how they characterise $\ELOu$.
A relation $Z \subseteq \Delta^{\Imc} \times \Delta^{\Jmc}$ is an \emph{$\ELO(\Sigma)$ simulation} from $\Imc$ to $\Jmc$ iff it satisfies~(\textit{atom$_{R}$}),
i.e.,
the `only if' direction of the Condition~(\textit{atom}), and the Condition~(\textit{forth}) given above.
An \emph{$\ELOu(\Sigma)$ simulation} from $\Imc$ to $\Jmc$ is an $\ELO(\Sigma)$ simulation from $\Imc$ to $\Jmc$ that
is \emph{left total}, meaning that $\Delta^{\Imc}$ is the domain of the relation.
%
We write $(\Imc, d) \simul^{\ELOu}_{\Sigma} (\Jmc, e)$ if there exists
an $\ELOu(\Sigma)$ simulation $Z$ from $\Imc$ to $\Jmc$ with $(d,e) \in Z$.
Given
a DL $\Lmc$,
a partial interpretation $\Imc$
with
$d \in \Delta^{\Imc}$, and a signature $\Sigma$,
we call the \emph{$\Lang(\Sigma)$ type of $d$ in $\Imc$} the set $\tp^{\Imc}_{\Lang(\Sigma)}(d)$ of
$\Lang(\Sigma)$ concepts $C$ such that
$d \in C^{\Imc}$.

\begin{restatable}{theorem}{elousimequiv}
\label{thm:simequiv}
\label{thm:elousimtoequiv}
\label{thm:elousimequiv}
For all signatures $\Sigma$ and all partial pointed interpretations $(\Imc,d)$ and $(\Jmc,e)$,
\begin{enumerate}
	\item
	if
	$(\Imc, d)  \simul^{\ELOu}_{\Sigma} (\Jmc, e)$,
 	then
$\tp^{\Imc}_{\ELOu(\Sigma)}(d) \! \subseteq  \! \tp^{\Jmc}_{\ELOu(\Sigma)}(e)$;
	\item
	if
	$\tp^{\Imc}_{\ELOu(\Sigma)}(d) \subseteq  \tp^{\Jmc}_{\ELOu(\Sigma)}(e)$ and $\Jmc$
	is
	$\omega$-saturated,
	then
	$(\Int, d) \simul^{\ELOu}_{\Sigma} (\Jmc, e)$.
\end{enumerate}
\end{restatable}

\section{Existence of Referring Expressions}
\label{sec:bdp}


One of the main motivations for enriching DLs with definite descriptions 
comes from the observation 
that individual names used in databases, ontologies, or other forms of
KBs,
are very often completely
meaningless
to
the human user~\cite{BorEtAl16}. 
Introducing semantically meaningful referring expressions (REs)
in addition to
individual names via ontologies enables a more informative naming of individuals, and thus a more 
user-friendly modelling of domains.
In this section, we address the problem of providing support for the generation of such expressions for individual names that
occur in an ontology.
Thus,
for
an individual name $a$ and an ontology $\Omc$, the goal is to support 
the generation of a concept $C$ so that
$$
\Omc \models \{a\} \equiv C,
$$
if such a concept exists. One can then replace $a$ by $\defdes C$ in the 
ontology or add an explicit definition of $\{a\}$ to the ontology 
and possibly remove other inclusions that become redundant 
(see~\cite{TenEtAl06,DBLP:journals/jair/CateFS13} 
for this approach applied to concept names rather than individual names).
Other than
being used
to improve an ontology,
REs may also be regarded as answer to queries about individuals~\cite{BorEtAl16}.

To support the targeted generation of
REs,
we consider two types of restrictions on $C$:
$(i)$
the restriction of the signature of $C$ to some subset $\Sigma$ of the signature of $\Omc$;
and
$(ii)$
restrictions
on
the DL constructors used in $C$.
As an
initial
step,
we focus on the
DLs
introduced in this paper and on the complexity of deciding the existence of
an RE.
We also discuss briefly what happens if one admits
FO
formulas as
REs.
The algorithm deciding the existence of
REs
can then inform the development of
a
generating algorithm
for REs,
although
this is
beyond the scope of
the present
contribution. 

Formally, given a pair $(\Lmc,\Lmc_{R})$ of logics, \emph{$(\Lmc,\Lmc_{R})$
RE
existence} is the problem of deciding, for an $\Lmc$ ontology $\Omc$, an individual name $a$, and a signature $\Sigma$, 
whether \emph{$a$ is explicitly $\Lmc_{R}(\Sigma)$ definable under $\Omc$}, that is,
whether
there exists an $\Lmc_{R}(\Sigma)$ concept $C$ such that $\Omc \models \{a \} \equiv C$.
If $\Lmc_{R}=\text{FO}$, then we ask whether there is an FO formula $\varphi(x)$ over $\Sigma$ such that $\Omc \models \forall x ((x=a) \leftrightarrow \varphi(x))$.
Such a concept or formula is called an \emph{$\Lmc_{R}(\Sigma)$ 
RE
for $a$ under $\Omc$}.
If $\Lmc=\Lmc_{R}$, then we simply speak of \emph{$\Lmc$
RE
existence}.
\begin{theorem}\label{thm:mainrefexp} On partial and total interpretations:
	\begin{enumerate}
   \item $(\ALCOud,\text{FO})$
   	RE
	existence is \ExpTime-complete;
   
   \item $(\ELOud,\text{FO})$
      RE
      existence is in \PTime;
   
   \item $\ALCOud$
   	RE
	existence is 2\ExpTime-complete;

   \item $(\ALCOud,\ELOud)$
   RE
   existence is undecidable;

   \item $\ELOud$
   RE
   existence is in \PTime, for individuals $a$ that denote w.r.t.~the ontology.
   \end{enumerate}
\end{theorem}
%
%
%
%
%
We first comment on Points~(1) and (2), which are consequences of the fact that
FO has the
projective Beth definability property
on total and partial interpretations, and the \ExpTime{} and \PTime{} upper bounds for reasoning 
in \ALCOud and \ELOud, respectively. In detail, we say that an individual name $a$
is \emph{implicitly definable} from a signature $\Sigma$ under an ontology $\Omc$ if for all (partial or total, respectively) models $\Int$ and $\Jmc$ of $\Omc$ such that $\Delta^{\Int} = \Delta^{\Jmc}$ and $X^{\Int} = X^{\Jmc}$, for all $X \in \Sigma$, we have
$a^{\Int} = a^{\Jmc}$. 
Clearly, if $a$ is explicitly $\Lmc(\Sigma)$-definable under $\Omc$, then it 
is implicitly definable from $\Sigma$ under $\Omc$.
We say that $\Lmc$ has the
\emph{projective Beth definability property (PBDP) for individuals}
if the converse holds.
FO is
known to have the PBDP for individuals (and, in fact, for arbitrary relations) under total interpretations~\cite{ChaKei90}.
It is straightforward to extend
 the model-theoretic proof given, e.g., in~\cite{ChaKei90} to show that FO enjoys the PBDP for individuals (and, again, arbitrary relations) under partial interpretations as well. 
From an algorithmic viewpoint, the PBDP is important because it implies that explicit definability can be checked using implicit definability, and the latter reduces to a standard reasoning problem: $a$ is implicitly definable from $\Sigma$ under $\Omc$ iff $\Omc \cup \Omc' \models \{a\} \equiv \{a'\}$, where $\Omc'$ is obtained from $\Omc$ by
uniformly substituting every symbol $X$ not in $\Sigma$ by a fresh symbol $X'$.
We
thus proved Points~(1) and (2). 

FO is arguably too expressive as a language for
REs
for individuals in DL ontologies, and Points (3) to (5) address the problem of
finding
REs
within the DLs considered in this paper.
Unfortunately, $\ALCOud$ does not enjoy the PBDP for individuals.
\begin{example}\label{exm:jiuntcons}
Let $\Sigma = \{A,r\}$ and assume that $\Omc$ consists of the following CIs:
\begin{align*}
\{a\} \sqsubseteq & \ \exists r.\{a\},
\\
\neg \{a\} \sqcap A \sqsubseteq \ & \forall r.(\neg \{a\} \Rightarrow \neg A), \\
\neg \{a\}\sqcap \neg A \sqsubseteq \ & \forall r.(\neg \{a\} \Rightarrow A).
\end{align*}
Then $a$ is implicitly definable from $\Sigma$ under $\Omc$ since $\Omc\models \forall x ((x=a) \leftrightarrow r(x,x))$.
However, the figure depicted below (for which we assume $a^{\Imc}=a$) 
shows
a model $\Imc$ of $\Omc$ such that $(\Imc, a^{\Imc}) \sim^{\ALCOud}_{\Sigma} (\Imc, e)$, with $e \neq a^{\Imc}$. Thus, any $\ALCOud(\Sigma)$ concept that applies to $a^{\Imc}$ applies to $e$ in $\Imc$ and so $a$ cannot be explicitly
$\ALCOud(\Sigma)$ definable under $\Omc$.
	

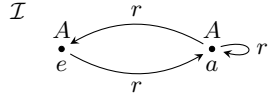
\begin{figure}[h]
\centering
\begin{tikzpicture}
\tikzset{
dot/.style = {draw, fill=black, circle, inner sep=0pt, outer sep=0pt, minimum size=2pt}
}

\draw (-0.75,0.5) node[label=west:$\mathcal{I}$] {};
\node (a) at (1.5,0) {};
\draw (a) node[dot, label=south:$a$, label=north:$A$] {};
\node (b) at (-0.5,0) {};
\draw (b) node[dot, label=north:$A$,  label=south:$e$] {};

\draw[->, >=stealth]  (a) edge [loop right] node[] {$r$} ();
\path[->, >=stealth, bend right] (a) edge (b);
\path[->, >=stealth, bend right] (b) edge (a);
\node (x) at (0.5,0.5) {$r$};
\node (y) at (0.5,-0.5) {$r$};

%
%


\end{tikzpicture}
\label{fig:nopbdp}
\end{figure}

\end{example}
We now come to the proof of (3). We adapt the proof given in~\cite{ArtEtAl21}
that
RE
existence in $\ALCOu$ is \TwoExpTime-complete. The proof is based on the bisimulation characterisation given in the previous section. The changes required to the proof in~\cite{ArtEtAl21} are subtle and non-trivial, however, as one now has to take care of the Condition ($\defdes$) for definite descriptions.
Assume $\Omc$, $\Sigma$, and an individual $a$ are given. Then we say that $\Omc,\{a\}$ and $\Omc,\neg\{a\}$ are \emph{jointly consistent modulo $\ALCOud(\Sigma)$ bisimulations} iff there exist models $\Imc$ and $\Jmc$ of $\Omc$ and $e\in \Delta^{\Jmc}$ such that $e\not=a^{\Jmc}$ and  $(\Imc,a^{\Imc}) \sim^{\ALCOud}_{\Sigma} (\Jmc,e)$.
Example~\ref{exm:jiuntcons} above illustrates this definition. The following
lemma reduces the
RE
existence problem to the problem of deciding the complement of joint consistency modulo $\ALCOud(\Sigma)$
bisimulations.
\begin{restatable}{lemma}{expdefchar}
	\label{thm:expdefchar}
	Let $\Omc$ be an $\ALCOud$ ontology, $a$ an individual name, and $\Sigma$ a signature. Then the following conditions are equivalent:
	\begin{enumerate}
		\item there exists an $\ALCOud(\Sigma)$ RE for $a$ under $\Omc$;
		\item $\Omc,\{a\}$ and $\Omc, \lnot \{a\}$ are not jointly consistent modulo $\ALCOud(\Sigma)$ bisimulations.
	\end{enumerate}
\end{restatable}
The following lemma states the main technical result from which Point~(3) follows.
\begin{restatable}{lemma}{bothdir}
	\label{lem:bothdir}
		For \ALCOud ontologies $\Omc$, signatures $\Sigma$, and individual names $a$, joint consistency of $\Omc,\{a\}$ and $\Omc,\neg\{a\}$ modulo $\ALCOud(\Sigma)$-bisimulations 
	is 2\ExpTime-complete. The lower bound holds already if $a$ is the only
	individual name in $\Omc$.
\end{restatable}
\begin{proof}[Sketch]
	The proof of the upper bound is mosaic-based, where mosaics are pairs $(T_{1},T_{2})$ of sets $T_{1}$ and $T_{2}$ of $\Omc$-types such that
	there exist models $\Imc_{1}$ and $\Imc_{2}$ of $\Omc$ and nodes $d_{t}, t\in T_{i}$, realising $t$ in $\Imc_{i}$, $i=1,2$, such that all pairs $d_{t},d_{t'}$ with $t,t'\in T_{1}\cup T_{2}$ are $\ALCOud(\Sigma)$ bisimilar. The maximal sets of such pairs that can make up the types realised in models $\Imc_{1},\Imc_{2}$ can be enumerated in double exponential time by formulating appropriate constraints and employing a recursive elimination procedure. These constraints have to be extended significantly compared to~\cite{ArtEtAl21} to deal with the cardinality constraints imposed by definite descriptions. The lower bound can be proved by adapting the lower bound proof for $\ALCOu$ in~\cite{ArtEtAl21}.
 \end{proof}
To show
Point~(4),
we make use of a recent undecidability proof in the context of conjunctive query inseparability of $\mathcal{ALC}$ KBs
given in~\cite{BotEtAl19}. In that paper it is shown that it is undecidable 
whether two $\mathcal{ALC}$ KBs entail the same conjunctive queries. To this end, an $\ALC$ ontology $\Omc_{1}$, an $\EL$-ontology $\Omc_{2}$, a signature $\Sigma$ consisting of concept and role names, and a single individual name $a$ are constructed such that it is undecidable whether 
$$
\Omc_{1} \cup \{r(a,a)\} \models C(a) \quad \Rightarrow \quad  \Omc_{2}\cup \{r(a,a)\}\models C(a)
$$
for all $\EL(\Sigma)$ concepts $C$. It turns out that only concepts $C$ of a particular form are relevant. In our reduction we create an \ALCO ontology $\Omc$ by taking the union of relativized versions of $\Omc_{1}$ and $\Omc_{2}$
and two individuals $a$ and $b$ such that the relativization of $\Omc_{1}$ acts on $a$ and the relativization of $\Omc_{2}$ on $b$. We add a fresh concept name $D_{a,b}$ to $\Sigma$ and add $D_{a,b}\equiv \{a\} \sqcup \{b\}$ to $\Omc$.
Then the set $\{a,b\}$ is trivially explicitly definable using an $\EL(\Sigma)$ concept under $\Omc$ and so $a$ is explicitly definable using an $\ELOud(\Sigma)$ concept under $\Omc$ iff it can be distinguished from $b$ using an $\ELOud(\Sigma)$ concept in the sense that
$\Omc \models C(a)$ and $\Omc\models \neg C(a)$. The undecidability of a
weaker form of distinguishability (does there exist an $\ELOud(\Sigma)$ concept such that
$\Omc \models C(a)$ and $\Omc\not\models C(a)$?) is immediate. To achieve
undecidability of the stronger form of indistinguishability we add 
further inclusions to $\Omc$ that make subtle use of the form of the relevant concepts $C$.

We come to the proof of Point~(5). 
It has already been observed in Point~(3) of Section~\ref{sec:observations} that for every $\ELOud$ ontology $\Omc$, $\Sigma \subseteq \sig{\Omc}$, and individual name $a$ that denotes w.r.t.~$\Omc$, there exists an $\ELOu(\Sigma)$ RE for $a$ under $\Omc$ iff there exists an $\ELOud(\Sigma)$ RE for $a$ under $\Omc$.
%
%
%
%
%
%
%
Thus, as we assume in Point~(5) that $a$ denotes
w.r.t. the ontology, it suffices to decide the existence of
$\ELOu(\Sigma)$ REs.
In fact, one can use Theorem~\ref{theo:eloufincanmod} and the simulation characterisation of $\ELOu$ concepts given in Lemma~\ref{thm:simequiv} to prove the following characterisation of the existence of
REs.
%
\begin{restatable}{lemma}{elourefexprchar}
	\label{thm:elourefexprchar}
	Given a signature $\Sigma$ and an $\ELOud$ ontology $\Omc = \Omc_{0} \cup \{ {\mathsf A} \equiv \{ a \} \}$, where $\Omc_{0}$ is an $\ELOud$ ontology in normal form, ${\mathsf A}$ is a
	concept name
	satisfiable w.r.t. $\Omc$,
	and $a$ is an individual name,
	the following
	are equivalent:
	\begin{enumerate}
		\item there does not exist an $\ELOu(\Sigma)$
		RE
		for $a$ under $\Omc$;
		\item
		there exist a model $\Jmc$ of $\Omc$ and $e \in \Delta^{\Jmc}$ such that $(\fincanmod, a^{\fincanmod}) \simul^{\ELOu}_{\Sigma} (\Jmc, e)$ and $e \neq a^{\Jmc}$.
	\end{enumerate}
\end{restatable}
We show in the appendix that Condition~(2) of Lemma~\ref{thm:elourefexprchar} can be checked in polynomial time in the size of $\Omc$. Thus, to show Point~(5), it suffices to recall that every $\ELOud$ ontology can be transformed
in polynomial time into a $\ELOud$ ontology in normal form that is a conservative extension of the original ontology, and observe that the existence of an
RE
over a subset of the signature
of an ontology is invariant under replacing the ontology by a conservative extension.

The complexity of
RE
existence remains
open if we do not assume the individual $a$ denotes (of course, if we assume total interpretations, this is trivially the case). 
We note that a \PTime{} upper bound cannot be shown via the PBDP: while
$\EL$ enjoys the PBDP~\cite{LutEtAl19}, neither $\ELOd$ nor $\ELOud$ enjoy the PBDP for individuals on partial interpretations.
To show this for  $\ELOud$ one can use an example provided in~\cite{MF21}.
The following example shows that
$\ELOd$ does not enjoy the PBDP for individuals on partial interpretations.

\begin{example}
Let $\Sigma = \{ b, B \}$ and $\Omc$ be the $\ELOd$ ontology:
\[
\{ a \} \sqsubseteq \{ b \}, \qquad
\{ a \} \sqsubseteq \exists r.B, \qquad
B \sqsubseteq \exists s.\{ a \}.
\]
On \emph{partial} interpretations, $a$ is implicitly definable from $\Sigma$ under $\Omc$, since $\Omc \models \{ a \} \equiv \{ b \} \sqcap \exists u.B$.
However, $a$ is not $\ELOd(\Sigma)$ explicitly definable under $\Omc$, as shown by the partial interpretations $\Imc, \Jmc$ that are models of $\Omc$ in the figure below, where we assume $a^{\Imc} = b$,
$a^{\Jmc}$ undefined,
$b^{\Imc} = b$, and $b^{\Jmc} = b$.
Non $\ELOd(\Sigma)$ explicit definability follows from the observation that $b$
satisfies exactly
the same $\ELOd(\Sigma)$ concepts in $\Imc$ and in $\Jmc$ (namely $\{b\}$).


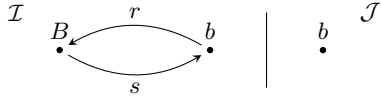
\begin{figure}[h]
\centering
\begin{tikzpicture}
\tikzset{
dot/.style = {draw, fill=black, circle, inner sep=0pt, outer sep=0pt, minimum size=2pt}
}

\draw (-0.75,0.5) node[label=west:$\mathcal{I}$] {};
\node (a) at (1.5,0) {};
\draw (a) node[dot, label=north:$b$] {};
\node (b) at (-0.5,0) {};
\draw (b) node[dot, label=north:$B$] {};

\path[->, >=stealth, bend right] (a) edge (b);
\path[->, >=stealth, bend right] (b) edge (a);
\node (x) at (0.5,0.5) {$r$};
\node (y) at (0.5,-0.5) {$s$};

\draw  (2.25,0.5) -- (2.25,-0.5);

\draw (4,0.5) node[label=west:$\mathcal{J}$] {};
\node (aa) at (3,0) {};
\draw (aa) node[dot, label=north:$b$] {};
%

\end{tikzpicture}
\label{fig:nopbdp}
\end{figure}

\end{example}


\section{Free DLs with Dual-Domain Semantics}
\label{sec:otherfreedl}

\newcommand{\fnot}{\sim\!}


Free logics
deal with terms that might fail to denote any existent object.
Semantically, they need to address the following
issues:
$(i)$ the formal
distinction between existent and non-existent objects;
$(ii)$ the truth of atomic 
formulas involving terms that do not refer to existent objects.
Concerning
$(i)$,
the two main options are
the so-called \emph{single-domain} and \emph{dual-domain semantics}.
Single-domain semantics are based on interpretations with a unique domain of objects,
representing the set of things that \emph{possibly} exist.
%
Dual-domain semantics introduce instead interpretations
based on two domains of objects: the \emph{outer domain},
representing
the set of all possible things; and the \emph{inner domain}, a subset of the outer domain
over which quantifiers are allowed to range,
containing
only those objects
that \emph{actually} exist.
%
%
Concerning
$(ii)$, three of the most prominent options are
the so-called
\emph{positive},
\emph{negative}
and
\emph{neutral} (or \emph{gapped}) \emph{semantics}.
Positive semantics allow atomic formulas with terms that refer to non-existent objects to be possibly \emph{true}.
%
Negative semantics
require
all atomic formulas with terms denoting non-existent objects, or non-denoting at all, to be \emph{false} by default.
%
Neutral semantics introduce a truth-value \emph{gap} for such formulas,
so that their truth value is
left
undefined~\cite{Nol20}.

%
While
partial interpretations are a kind of single-domain negative semantics,~\cite{NeuEtAl20}
consider
free DLs with nominals over dual-domain interpretations, presenting a positive, a negative and a gapped semantics.
%
This approach naturally
fits
scenarios
where
objects of the domain can start or cease to exist,
as
frequently considered in \emph{first-order modal} and \emph{temporal logics}~\cite{Gar01,GabEtAl03,BraGhi07,FitMen12,Gar13}.
Moreover, a positive semantics
allows one to
avoid inconsistencies
when reasoning in presence of data sources
that contradict the ontology,
a motivation
shared also by
\emph{inconsistency-tolerant DLs}~\cite{LemEtAl10,LemEtAl15},
by representing ``error'' individuals as non-existent objects.

%
To
show how the dual-domain semantics can be captured in our framework, we
present
the logic $\NKRd$~\cite{NeuEtAl20},
defined similarly to $\ALCOd$, with the addition of $\topex$ as a
concept name representing the set of {\em existing} objects.
Moreover, an \emph{$\NKRd$ assertion} is of the form $C(\tau)$,
$r(\tau_1,\tau_2)$, or $\tau_{1} = \tau_{2}$, with $C$ an $\NKRd$
concept, $r \in \NR$, and $\tau, \tau_1, \tau_2$ $\NKRd$ terms.  An
\emph{$\NKRd$ axiom} is either an $\NKRd$ CI, $C \sqsubseteq D$, with
$C,D$ $\NKRd$ concepts, or an $\NKRd$ assertion.
An \emph{$\NKRd$ formula}
is defined inductively
as follows,
with
$\beta$
$\NKRd$ axiom:
\[
	\p ::= \beta \mid \lnot (\p) \mid (\p \land \p).
\]

A \emph{dual-domain interpretation} is a triple
$I = (\OutDom^{I}, \InnDom^{I}, \cdot^{I})$, where $\OutDom^{I}$ is a
non-empty set, called \emph{outer domain} of $I$,
$\InnDom^{I} \subset \OutDom^{I}$
is a (possibly empty) set called
\emph{inner domain} of $I$, and $\cdot^{I}$ is a function mapping
every $A \in \NC$ to a subset of $\OutDom^{I}$, every $r \in \NR$ to a
subset of $\OutDom^{I} \times \OutDom^{I}$, and every $a \in \NI$ to
an element in $\OutDom^{I}$.
%
%
Given a dual-domain interpretation $I$, the \emph{extension} $C^{I}$
of an $\NKRd$ concept $C$ in $I$ is defined similarly to $\ALCOd$, with
the exception of:
\begin{gather*}
\topex^{I} = \InnDom^{I}, \\
(\exists r.C)^{I} = \{d \in \OutDom^{I} \mid \exists e \in \InnDom^{I} \cap C^{I}: (d,e) \in r^{I}\},
\end{gather*}
while, for a definite description $\defdes C$, the \emph{value of $\defdes C$ in $I$} is:
\begin{gather*}
		(\defdes C)^{I}  =
			\begin{cases}
				d, & \text{if} \ \InnDom^{I} \cap C^{I} = \{ d \}; \\
				d_{\defdes C}, & \text{with $d_{\defdes C} \in \OutDom^{I} \setminus \InnDom^{I}$ arbitrary, otherwise}.
			\end{cases}
\end{gather*}
Moreover, on dual-domain interpretations, we define the extension of $\NKRd$ nominals as follows: $\{ \tau \}^{I} = \{ \tau^{I} \}$. 

Given an $\NKRd$ formula $\p$ and a dual-domain interpretation $I$, we inductively define two different kinds of \emph{satisfaction relations} between $I$ and $\p$:
one
\emph{under positive semantics}, denoted by $\models^{+}$; and
one
\emph{under negative semantics}, denoted by $\models^{-}$.
In the following, $\circ \in \{+, -\}$.
\begin{align*}
	I \models^{\circ} C(\tau) & \ \text{iff} \
		\begin{cases}
			\tau^{I} \in C^{I}, & \,\,\,\,\,\,\,\, \text{if} \  \circ = + ; \\
			\tau^{I} \in \InnDom^{I} \ \text{and} \ \tau^{I} \in C^{I}, & \,\,\,\,\,\,\,\, \text{if} \ \circ = - ;\\
		\end{cases}
	\\
		I \models^{\circ} r(\tau_{1}, \tau_{2}) & \ \text{iff} \
		\begin{cases}
			(\tau_{1}^{I}, \tau_{2}^{I}) \in r^{I}, & \!\!\!\!\!\!\!\!\!\!\!\! \text{if} \  \circ = + ; \\
			\tau_{1}^{I}, \tau_{2}^{I} \in \InnDom^{I} \ \text{and} \ (\tau_{1}^{I}, \tau_{2}^{I}) \in r^{I}, &
		\end{cases}
		\\
		& \qquad\qquad\qquad\qquad\qquad\qquad\,\,\,\,\,\,\,\,\ \text{if} \ \circ = - ;\\
	I \models^{\circ} \tau_{1} = \tau_{2} & \ \text{iff} \
		\begin{cases}
			\tau_{1}^{I} = \tau_{2}^{I}, & \text{if} \  \circ = + ; \\
			\tau_{1}^{I}, \tau_{2}^{I} \in \InnDom^{I} \ \text{and} \ \tau_{1}^{I} = \tau_{2}^{I}, & \text{if} \ \circ = - ;\\
		\end{cases}
\end{align*}
\[
I \models^{\circ} C \sqsubseteq D \ \text{iff,} \ \text{for all $d \in \InnDom^{I}$, $d \in C^{I}$ implies $d \in D^{I}$,}
\]
\[
	I \models^{\circ} \lnot ( \psi ) \ \text{iff} \ I \not \models^{\circ} \psi, \
	I \models^{\circ} \psi \land \chi \ \text{iff} \ I \models^{\circ} \psi \ \text{and} \ I \models^{\circ} \chi.
\]
%
%
%
We say that $\p$ is \emph{satisfiable on dual-domain interpretations under positive}, respectively \emph{negative}, \emph{semantics} iff there exists a dual-domain interpretation $I$ such that
$I \models^{+} \p$, respectively $I \models^{-} \p$.

Differently from classical $\ALCO$ on total interpretations and from
$\ALCOud$ on partial interpretations (cf. Point~(2) in
Section~\ref{sec:observations}), $\NKRd$ assertions under these
semantics cannot be encoded into $\NKRd$ CIs.  Moreover, parentheses
in negated formulas $\lnot (\p)$ are not eliminable, since they
disambiguate between assertions of the form $\lnot C(\tau)$, with a
\emph{negated concept}, and \emph{negated assertions} of the form
$\lnot (C (\tau))$.  Indeed, these expressions have different
satisfaction conditions on dual-domain interpretations under negative
semantics: while $\lnot C(\tau)$ requires $\tau^{I}$ to be an element
of the inner domain in any of its models $I$, a formula like
$\lnot (C (\tau))$ is satisfied also in dual-domain interpretations
$I$ where $\tau^{I}$ is in the outer, but not in the inner, domain.

We now show
that
$\NKRd$
satisfiability
on dual-domain interpretations under either positive or negative
semantics
is polynomial time reducible to $\ALCOud$ reasoning on
partial interpretations.
The proof is reminiscent of the one from first-order modal logic, to
reduce varying to constant domain semantics~\cite{BraGhi07,FitMen12},
in that it exploits an \emph{existence concept} to represent the inner
domain on partial interpretations.
%
%
%
%
\begin{restatable}{theorem}{redddtopart}
\label{prop:redddtopart}
$\NKRd$ formula satisfiability on dual domain interpretations under either positive or negative semantics is polynomial time reducible to $\ALCOud$ ontology satisfiability.
\end{restatable}
Since the gapped semantics
allows for truth value gaps,
it is not covered by
our setting based
  on a two-valued semantics.
A comparison with this option is left as future work.

\section{Related Work}
\label{sec:relwork}

%
%

Definite descriptions
introduce
mild forms
of  \emph{cardinality constraints}, a set of constructors
with
a long tradition in DLs~\cite{BaaEtAl96,Tob00,BaaEck17,BaaBor19,BedEtAl20}
that allow to
constrain the number of elements in the extension of a concept.
The expressivity of many of these logics goes
far beyond the DLs proposed here, and novel reasoning tools are required.
In contrast,
reasoning in
our free DLs
can be reduced to
reasoning in standard DLs ($\ALCOud)$ or mild extensions ($\ELOud$).

Concerning
RE generation
tasks,
other DL-based approaches
have
studied
the problem of finding a concept to describe an element
with respect to a single interpretation given as input~\cite{AreEtAl08,AreEtAl11}. More expressive DLs, as well as a relaxed version of the closed-world assumption, are considered in~\cite{RenEtAl10a,RenEtAl10b}.
%


REs have also been proposed for several applications in ontology-based data management, such as query answering over KBs~\cite{BorEtAl16,BorEtAl17,TomWed19a}, identity resolution in ontology-based data access~\cite{TomWed18,TomWed19}, and identification problems in conceptual modelling~\cite{BorEtAL16b}.
%
The DLs considered in these papers are tractable languages tailored to efficient query answering in presence of functionality
and path-based identification
constraints.
%
%
In this approach,
DL concepts can serve as REs under a given KB if they contains exactly one element in all the models of the KB
and satisfy a correctness condition with respect to
a query.
%
%
They are not, however, directly treated as possibly non-denoting terms of the language.

Finally,
hybrid logics with non-denoting nominals have not received much attention in the literature, with the exception of~\cite{Han11} in the context of public announcement logics.
%
However,
formalisms involving definite descriptions
are
actively
investigated
in first-order modal logic~\cite{Ind18,OrlCor18,Ind20},
%
%
%
where
the
possible lack of
referents
for names and descriptions is usually paired with
\emph{non-rigid} denotation features, i.e., the ability to refer to different objects at different states.

\section{Discussion}
\label{sec:discuss}
%
%
We have introduced
DLs
with definite descriptions
on partial interpretations,
and investigated standard reasoning (satisfiability and entailment) and automated support for generating definite descriptions (RE existence). Many open
problems remain to be explored.
Regarding $\ELOud$, it is
open whether the \PTime{} upper bound for RE existence holds in general,
%
%
whether there is a polynomial size canonical model for
$\ELOud$-concepts and whether a satisfactory model-theoretic 
characterisation of its expressivity can be given. Also, RE existence 
is only a first step towards automated support for generating definite descriptions in practice.
This could be approached by exploring the shape and interpretability of definitions obtained from interpolants computed by
FO
theorem provers.
Finally, we intend to
extend our
free DLs with definite descriptions
with a
temporal
dimension~\cite{LutEtAl08},
for applications in
temporal conceptual modelling and
query answering over temporal DL ontologies~\cite{ArtEtAl14,ArtEtAl17},
where the interaction between lack of denotation and non-rigidity can be at stake.




\bibliographystyle{plain}
\bibliography{bibliography}

\clearpage

\begin{appendix}
\label{sec:appendix}

\section*{Appendix}

In the following, given a partial interpretation $\Int = (\Delta^{\Int}, \cdot^{\Int})$, we denote by $\dom(\cdot^{\Int})$ the \emph{domain of definition} of the partial function $\cdot^{\Int}$.

\subsection*{Proofs for Section~\ref{sec:complexity}}

We show that the $(\cdot)\tr$ translation in
Section~\ref{sec:complexity} is satisfiability preserving
(Theorem~\ref{prop:alcodtoalcou}).  This is a direct consequence of
Lemmas~\ref{prop-sat1} and~\ref{prop-sat2} presented
below. 

We first prove that, whenever an ontology 
is satisfiable, its translation is also satisfiable.
%
\begin{lemma}\label{prop-sat1}
  Let $\Omc$ be an $\ALCOud$ ontology.
  For every
  partial interpretation $\Int$ such that 
  $\Int \models \Omc$, there is
a total interpretation
  $\Int\tr$
 such that 
  $\Int\tr \models \Omc\tr$.
\end{lemma}
\begin{proof}
Suppose $\Imc= (\Delta^{\Int}, \cdot^{\Int})$,
with $\Delta^{\Int}$ denoted simply by $\Delta$ in the following, is
  a partial interpretation s.t. $\Int \models \Omc$.
  We construct
a total
  interpretation
  $\Int\tr= (\Delta^{\Imc\tr}, \cdot^{\Int\tr})$,
with
  $\Delta^{\Imc\tr} = \Delta^{\Imc}$, 
  such that
  $\Int\tr$ coincides with $\Int$ on
  all concept, role and individual names, except possibly from the fresh ones occurring in $\Omc\tr$, and for which we set the following,
 for every
  $\{ \tau \} \in \con{\Omc}$:
\begin{itemize}
%
\item if $\tau = b$ and
  $b \in \dom(\cdot^{\Int})$ then
  $a_{\tau}^{\Int\tr} = b ^{\Int}$ and
  $A_b^{\Int\tr} = \{b^{\Int}\}$ (otherwise, $a_{\tau}^{\Int\tr}$
  is arbitrary and $A_b^{\Int\tr} = \emptyset$);
\item if $\tau = \defdes B$ and
  $B^{\Int} \neq \emptyset$ then
  $a_{\tau}^{\Int\tr} = d$,
  for some $d \in B^{\Int}$ ($a_{\tau}^{\Int\tr}$ is arbitrary,
  otherwise).
  \end{itemize}
%
Finally, the individual names in $\NI \setminus \dom(\cdot^{\Int})$
are mapped arbitrarily by ${\Int\tr}$.  We first show the following
claim.
\begin{claim}\label{claim:tr1}
  For
  every $\{ \tau \} \in \con{\Omc}$ and every $d\in\Delta$,
  $d\in\{\tau\}^\Int$ iff $d\in(\{\tau\}\red)^{\Int\tr}$.
\end{claim}
\begin{proof}
  ($\Rightarrow$) We consider the two different forms of $\tau$.
  \begin{itemize}
  \item Let $\tau = b$ and $d\in \{\tau\}^\Int$. Then,
    $\{\tau\}^\Int = \{b^\Int\}$, $b \in \dom(\cdot^{\Int})$, and by
    construction, $a_{\tau}^{\Int\tr} = b^\Int$ and 
    $A_b^{\Int\tr} = \{a_{\tau}^{\Int\tr}\}$. So, 
    $(\{\tau\}\red)^{\Int\tr} = (\concleqone)^{\Int\tr} \cap
    A_b^{\Int\tr} = A_b^{\Int\tr} = \{b^\Int\} = \{d\}$.
  \item Let $\tau = \defdes B$, and $d\in \{\tau\}^\Int$. Then,
    $\{\tau\}^\Int = \{d\}$, $B^\Int = \{d\}$, and, by construction,
    $d = a_{\tau}^{\Int\tr}$, $B^{\Int\tr} = B^\Int= \{d\}$, and
    $(\{\tau\}\red)^{\Int\tr} = (\concleqone)^{\Int\tr} \cap
    B^{\Int\tr} = B^{\Int\tr} = \{d\}$.
  \end{itemize}
  ($\Leftarrow$) We consider the two different forms of $\tau$.
  \begin{itemize}
  \item Let $\tau = b$ and
    $d\in (\{\tau\}\red)^{\Int\tr} = (\concleqone)^{\Int\tr} \cap
    A_b^{\Int\tr}$. Then, $d\in A_b^{\Int\tr}$ and, by construction,
    $b \in \dom(\cdot^{\Int})$, $a_{\tau}^{\Int\tr} = b^\Int = d$. Thus, $\{\tau\}^\Int = \{b^\Int\} = \{d\}$.
  \item Let $\tau = \defdes B$, and
    $d\in (\{\tau\}\red)^{\Int\tr} = (\concleqone)^{\Int\tr} \cap
    B^{\Int\tr}$. Then, $d\in B^{\Int\tr}$ and, since
    $d\in (\concleqone)^{\Int\tr}$, we have that
    $B ^{\Int\tr} = \{d\} = B^\Int$. Thus, $\{\tau\}^{\Int} =\{d\}$.
    \qedhere
  \end{itemize}
\end{proof}


By construction of $\Int\tr$, one can see that
$\Int\tr\models \big(\{\tau\}^+ \sqsubseteq \forall u. (\{
a_{\tau} \} \Rightarrow \{ \tau \}^+) \big)$.
%
Moreover, without loss of generality, we can assume that $\Omc$ is in normal form.
We now show the following claim.

\begin{claim}
For every $C \sqsubseteq D \in \Omc$,
%
  \label{eqn:tr}
  $\Int\models C \sqsubseteq D \ \text{iff} \ \Int\tr\models (C \sqsubseteq D)\red$.
\end{claim}
\begin{proof}
%
We distinguish the following cases.
\begin{itemize}
\item Let $(C \sqsubseteq D) = (E \sqsubseteq F)$, where $E, F$ are
  $\ALC_{u}$ concepts.
  Since $(E \sqsubseteq F)\red = E \sqsubseteq F$ and $\Int\tr$ coincides with $\Int$ on all concept and role names (except possibly from the fresh ones occurring in the translation of a formula),
  an easy induction on $\ALC$ concepts shows that
  the statement of the claim
  holds in this case.
\item Let $(C \sqsubseteq D) = (\{ \tau \} \sqsubseteq A)$. By
  Claim~\ref{claim:tr1}, $\{\tau\}^\Int =
  (\{\tau\}\red)^{\Int\tr}$. Since $(\{ \tau \} \sqsubseteq A) = \{ \tau \}\red \sqsubseteq A$ and $\Int\tr$ coincides with $\Int$ on
 all concept names (except possibly from the fresh ones occurring in the translation of a formula),
  we have that
  the statement of the claim
  holds in this case.
\item Let $(C \sqsubseteq D) = (A \sqsubseteq \{ \tau \})$. The proof is similar to
  the previous case.
  \qedhere
\end{itemize}
%
\end{proof}
Since we have by assumption that $\Imc \models \Omc$, the previous claim concludes the proof of the lemma.
\end{proof}

We now show that, whenever an ontology $\Omc\tr$ is satisfiable, the original ontology $\Omc$ is also satisfiable.
%
\begin{lemma}\label{prop-sat2}
  Let $\Omc$ be an $\ALCOud$ ontology.
  For every
total
  interpretation $\Int$ such that 
  $\Int \models \Omc\tr$, there
  is
a partial interpretation $\Int'$ such that 
  $\Int' \models \Omc$.
\end{lemma}
\begin{proof}
Suppose that $\Int = (\Delta^{\Int}, \cdot^{\Int})$ and that $\Int \models \Omc\tr$.
%
Consider
a partial interpretation $\Int' = (\Delta^{\Int'}, \cdot^{\Int'})$, with $\Delta^{\Int'} = \Delta^{\Int}$, such that $\Int'$
coincides with $\Int$ on
all concept, role and individual names,
except possibly from the following.
For all $\{ \tau \} \in \con{\Omc}$:
\begin{itemize}
\item if $\tau = b$, and $d \in \Delta^{\Imc}$, then, if
  $A_{b}^{\Int} = \{ d \} $, we set $b^{\Imc'} = d$; otherwise,
  $b \not \in \dom(\cdot^{\Int'})$;
\item if $\tau = \defdes B$ and $d \in \Delta^{\Imc}$,
we have
  $\defdes B^{\Int'} = d$ iff $B^{\Int} = \{ d \}$.
\end{itemize}
%
Finally, $\Int'$ coincides with $\Int$ on individual names not occurring in $\Omc$.
Now we prove the following claim.
\begin{claim}\label{claim:tr2}
  For every
  $\{ \tau \} \in \con{\Omc}$
  and every $d \in \Delta^{\Imc}$,
  $d \in (\{\tau\}\red)^\Int$ iff $d \in \{\tau\}^{\Int'}$.
\end{claim}
\begin{proof}
  $(\Rightarrow)$ We consider the two different forms of $\tau$.
  \begin{itemize}
  \item Let $\tau = b$. If
    $d \in (\{ b \}\red)^{\Int} = (C^{\leq 1}_{b} \sqcap
    A_b)^{\Int}$, then $d \in A_b^{\Int}$ and
    $A_{b}^{\Int} \subseteq \{ a_{b}^{\Int} \}$. Thus, we have
    $d = a_{b}^{\Int}$. By definition of $\Int'$,
    we have $b^{\Int'} = a_{b}^{\Int}$, which implies $d \in \{ b^{\Int'} \}$, and
    hence $d \in \{ b \}^{\Int'}$.
  \item Let $\tau = \defdes B$, and
    $d\in (\{ \defdes B \}\red)^{\Int} = (C_{\defdes B}^{\leq 1} \sqcap
    B)^{\Int}$. We have, therefore, that $d\in B^{\Int}$ and
    $B^{\Int} \subseteq \{ a_{\defdes B}^{\Int} \}$.  Thus,
    $d = a_{\defdes B}^{\Int}$ and
    $B^{\Int} = \{ a_{\defdes B}^{\Int} \}$.  By definition of
    $\Int'$, we have $d = a_{\defdes B}^{\Int} = \defdes B^{\Int'}$,
    and hence $d \in \{ \defdes B \}^{\Int'}$.
  \end{itemize}
  $(\Leftarrow)$ We consider the two different forms of $\tau$.
  \begin{itemize}
  \item
  Let $\tau = b$, and $d \in \{ b \}^{\Int'}$, i.e.,
  $d = b^{\Int'}$.
  By construction of $\Int'$, then
  $A_{b}^{\Int} = \{ d \}$.  Since $\Int \models \Omc\tr$, we have in
  particular that
  $d \in A_{b}^{\Int}$
  implies
  $a_{b}^{\Int} \in A_{b}^{\Int}$.  Therefore,
  $A_{b}^{\Int} = \{d\} = \{ a_{b}^{\Int} \}$,
  and thus,
  $d \in (C^{\leq 1}_{b} \sqcap A_{b})^{\Int} = (\{ b \}\red)^{\Int}$.
\item Let $\tau = \defdes B$, and
  $d \in \{ \defdes B \}^{\Int'}$, i.e., $d = \defdes B^{\Int'}$.  By
  definition of $\Int'$, $B^{\Int} = \{ d \}$.  Since
  $\Int \models \Omc\tr$, we have in particular that
  $B^{\Int} \neq \emptyset$ implies
  $a_{\defdes B}^{\Int} \in B^{\Int}$.  Hence, 
  $B^{\Int} = \{ d \} = \{ a_{\defdes B}^{\Int} \}$, and thus,
  $d \in (C^{\leq 1}_{\defdes B} \sqcap B)^{\Int} = ( \{ \defdes B
  \}\red)^{\Int}$.
  \qedhere
  \end{itemize}
\end{proof}
\noindent
The last step is to show
the following claim.
%
\begin{claim}
  \label{eqn:tr2}
  For every $C \sqsubseteq D \in \Omc$, 
  $\Int \models (C \sqsubseteq D)\red \ \text{iff} \ \Int' \models C \sqsubseteq D$.
\end{claim}
\begin{proof}
Similar to
Claim~\ref{eqn:tr}.
\end{proof}
Since by assumption we have that $\Imc \models \Omc\tr$, the previous claim concludes the proof of the lemma.
\end{proof}


We now prove Lemma~\ref{lem:completion} in
Section~\ref{subsec:elo}.

\Eloentailment*

\begin{proof} 
  ($\Leftarrow$) Suppose the completion rules generate the sequences
  $S_0,\ldots,S_{m_1}$, $R_1,\ldots,R_{m_2}$, 
  $V_0,\ldots,V_{m_3}$, representing each update of $S,R,V$.
  Assume $S({\sf A})\cap \{B,\bot\}\neq \emptyset$. To
  show that $\Omc\models {\sf A}\sqsubseteq B$ we need to show that
  for all models $\Imc$ of $\Omc$ and for any $x\in\Delta^{\Imc}$, if
  $x\in{\sf A}^{\Imc}$ then $x\in B^{\Imc}$ which is a direct
  consequence of the following claim.
  \begin{claim}\label{cl:aux1}
    For all models $\Imc$ of $\Omc$, either
    ${\sf A}^{\Imc}=\emptyset$, or for all $n\geq 0$, all
    $r\in \Rmc_\Omc$, all $D\in\Bmc\Cmc_\Omc^+$,
      $E\in\Bmc\Cmc_\Omc$, all
    $B\in \Bmc\Cmc_\Omc\cap \NC$, and all $x\in \Delta^{\Imc}$, if
    $x\in E^\Imc$ then:
        \begin{enumerate}
    \item $D\in S_n(E)$ implies $x\in D^{\Imc}$;
    \item 
      $r\in R_n(E,D)$ implies there is $y\in\Delta^\Imc$ such that
      $(x,y)\in r^\Imc$ and $y\in D^\Imc$;
    \item $B^c\not\in V_n$ implies $|B^\Imc|\leq 1$, i.e.,
      $B^\Imc = \{\iota B\}^\Imc$;
      \item ${\sf A}\leadsto_{R_n} C$ implies there exists $y\in
        \Delta^{\Imc}$ such that $y\in C^\Imc$.
    \end{enumerate}
  \end{claim}
  \begin{proof}
    The proof is by induction on $n$. Let $\Imc$ be a model of
    $\Omc$ s.t. ${\sf A}^{\Imc}\neq\emptyset$. The base case ($n=0$)
    easily holds. Indeed, 
    Point~(1) holds because for all $C\in\Bmc\Cmc_\Omc$, we have $S_0(C)=\{C,\top\}$.
    The conditions in  Point~(2) 
    hold vacuously.
     Point~(3) holds because there is no $B^c$ s.t.  $B^c\not\in V_0$.
Finally, Point~(4) holds for $n=0$ because in this case,
by definition of $S_0$ and $R_0$, we have ${\sf A}\leadsto_{R_0} C$ 
only for $C={\sf A}$ or $C=\top$. So, assuming ${\sf A}^{\Imc}\neq\emptyset$,
then there exists $y\in \Delta^{\Imc}$ such that $y\in C^\Imc$.
          %
    Thus, Claim~\ref{cl:aux1} holds in the base case.         
          
    We prove Point~(1) assuming that $D\in S_n(E)\setminus S_{n-1}(E)$
    and considering the rules that add $D$ to $S_n(E)$.
    \begin{description}
    \item[${\sf R}_1, 
    {\sf R}_6,{\sf R}_7$.] These cases can be proved
      similarly to~\cite{BaaEtAl05}.
    \item[${\sf R}_3$.]
      In this case, there is $ \exists r.C \sqsubseteq D \in\Omc$ and
      $B\in\Bmc\Cmc_\Omc$ such that 
      $C\in S_{n-1}(B)$ and $r\in R_{n-1}(E,B)$.  By induction on
      Point~(2), if $r\in R_{n-1}(E,B)$ then there is
      $y\in\Delta^\Imc$ such that $(x,y)\in r^\Imc$ and $y\in B^\Imc$.
      If $y\in B^\Imc$ and $C\in S_{n-1}(B)$ then by the i.h.,
      $y\in C^\Imc$. Thus, $x\in (\exists r.C)^\Imc$ and since
      $\exists r.C \sqsubseteq D \in\Omc$ we obtain $x\in D^\Imc$.
         %
      \item[${\sf R}'_3$.] Let $\exists u.C \sqsubseteq D\in
        \Omc$  and ${\sf A}\leadsto_{R_{n-1}} C$. Then, by
        induction on Point~(4), there is $y\in C^\Imc$. Thus,
        $(x,y)\in u^\Imc$ and since $\exists u.C \sqsubseteq D\in
        \Omc$, we obtain $x\in D^\Imc$.
      \item[${\sf R}_4$.] Since ${\sf A}^{\Imc}\neq\emptyset$ and
        ${\sf A} \leadsto_{R_{n-1}} D$, then, by induction on Point~(4), 
        $D^{\Imc}\neq\emptyset$. Furthermore, there is a singleton
        $\{\tau\}$ s.t. $\{\tau\}\in S_{n-1}(E)$ and
        $\{\tau\}\in S_{n-1}(D)$. By induction on Point~(1),
        $x\in \{\tau\}^\Imc$ and, for any $y\in D^\Imc$ then
        $y\in \{\tau\}^\Imc$. Since $D^{\Imc}\neq\emptyset$, then
        $x=y$. Now, let $C\in S_{n-1}(D)$. By induction,
        $y\in C^\Imc$, and thus $x\in C^\Imc$. 
    \item[${\sf R}_5$.] 
      In this case, 
      there is $r\in R_{n-1}(E,D)$, and $\bot\in S_{n-1}(D)$.
      Thus, by induction, $\exists y\in\Delta^\Imc$ s.t. $(x,y)\in r^\Imc$ and
      $y\in D^\Imc$.  By induction (on Point~(1)),
      $y\in\bot^\Imc$. Thus we conclude that there is no $\Imc$ where
      $E^\Imc\neq \emptyset$.
      %
    \item[${\sf R}_9$.] In this case, there is
      $B\in S_{n-1}(E)$, and $B^c\not\in V_{n-1}$. By induction on Points~(2)-(3),
      $x\in B^\Imc$ with $B^\Imc = \{\iota B\}^\Imc$.
    \end{description}
    %
    We now show Point~(2).
    By ${\sf R}_2$, $r\in R_n(E,D)$ only if
    ${\sf A} \leadsto_{R_{n-1}} E$, $C\in S_{n-1}(E)$, and $C\sqsubseteq \exists r. D\in\Omc$. 
     By induction on Point~(1), if $C\in S_{n-1}(E)$ then $x\in C^\Imc$. As $C\sqsubseteq \exists r. D\in\Omc$ and \Imc is a model of \Omc, there is $y\in\Delta^\Imc$ such that $(x,y)\in r^\Imc$ and
      $y\in D^\Imc$.
      %
      We show Point~(3) by considering the rules that remove $B^c$
      from $V_n$.
      \begin{description}
      \item[${\sf R}_8$.] In this case, there is a singleton
        $\{\tau\}$ s.t. $\{\tau\}\in S_{n-1}(B)$, with $B\in\NC$. By
        induction, for any $x\in B^\Imc$, then $x\in
        \{\tau\}^\Imc$. Thus, $|B^\Imc|\leq 1$, as required.
      \item[${\sf R}_{10}$.] In this case, 
        ${\sf A}\leadsto_{R_{n-1}} C$ and $\{\iota B\}\in S_{n-1}(C)$.
        By induction on Point~(4), there exists $y\in \Delta^\Imc$
        such that  $y\in C^\Imc$. By induction on Point~(1),
      $y\in \{\iota B\}^\Imc$, and thus $|B^\Imc|\leq 1$, as required.
    \end{description}%

    We prove Point~(4) by induction on  $R_n$,  using the assumption that
    ${\sf A}^\Imc\neq \emptyset$.  
    If ${\sf A}\leadsto_{R_n} C$ then there are $C_1, \dots, C_k\in \Bmc\Cmc_\Omc$
    such that $C_1\in S_n({\sf A})$; 
  %
$r\in R_n(C_j,C_{j+1})$, for some $r\in \Rmc_\Omc$, 
  for all $1\leq j < k$;
$C\in S_n(C_k)$.
For $k>1$, ${\sf A}\leadsto_{R_{n-1}} C_{k-1}$
implies, 
by the inductive hypothesis, 
that there is $y\in\Delta^\Imc$ such that
        $y\in C_{k-1}^\Imc$.
By Point (2), if $r\in R_n(C_{k-1},C_{k})$, for some $r\in \Rmc_\Omc$,
then there is $z\in\Delta^\Imc$
  such that
      $(y,z)\in r^\Imc$ and $z\in C^\Imc_k$.
 By Point (1), if $z\in C^\Imc_k$ and $C\in S_n(C_k)$
 then $z\in C^\Imc$.
 Then, there is $z\in\Delta^\Imc$ such that
        $z\in C^\Imc$, as required.
    \end{proof}
    %
%
    ($\Rightarrow$) Here we show the contrapositive. That is,  
    $S({\sf A})\cap \{B,\bot\}= \emptyset$ implies
    $\Omc\not\models {\sf A}\sqsubseteq B$. Suppose the
    exhaustive application of the completion rules generate $S,R,V$.
%
    We now define a model \Imc of \Omc and show that there is
    $d\in\Delta^\Imc$ such that $d\in {\sf A}^\Imc$ but
    $d\not\in B^\Imc$. Let us define the following two
    sets: 
    \begin{align*}
      \ra & = \{C\in \Bmc\Cmc_\Omc \mid {\sf A}\leadsto_R C
            \},\\
      \ra^c & = \{C^c\in V\mid C\in\ra\cap\NC\},
    \end{align*}
    over which we define the relation $\thicksim$ where 
  \begin{itemize}
  \item $C \thicksim D$ iff $C=D$ or $\{\tau\}\in S(C)\cap S(D)$, for
    some term $\tau$.
  \end{itemize}
  Due to ${\sf R}_4$, one can show that:
  \begin{equation}
  C \thicksim D \text{ implies } S(C) = S(D).\label{eq:1}
  \end{equation}
  Thus, $\thicksim$ is an equivalence relation, denoted as $\eq{C}$.
  We define the partial interpretation $\Imc=(\Delta^\Imc, \cdot^\Imc)$ as follows:
  \begin{itemize}
  \item $\Delta^\Imc =\{\eq{C} \mid C\in \ra\cup\ra^c\}$;
  \item $r^\Imc =\{(\eq{C},\eq{D}) \in \Delta^\Imc\times
    \Delta^\Imc\mid \exists D'\in [D]. r\in R(C,D')\}$, 
    for all
    $r\in\NR\cap \Rmc_\Omc$;
  \item $D^\Imc = \{\eq{C}\in \Delta^\Imc\mid D\in S(C)\}$, for all $D\in \NC$;
  \item $a^\Imc =\eq{C}$, for some $C\in\ra$, if $\{a\}\in S(C)$, for all $a\in\NI$.
    %
  \end{itemize}
  Moreover, the universal role $u$ is mapped by $\cdot^\Imc$ to
  $\Delta^\Imc\times\Delta^\Imc$. According
    to the semantic definition
  of a definite description,   $(\iota A)^\Imc =\eq{C}$ if
  $A^\Imc=\{\eq{C}\}$, for all $A\in\NC$. When a term $\tau$ does not
  have a value in \Imc, we write 
  $\tau\not\in\dom(\cdot^\Imc)$. Thus, we extend \Imc to singletons in
  the following way:
  \begin{itemize}
  \item $\{\tau\}^\Imc = \{\tau^\Imc\}$, if $\tau\in\dom(\cdot^\Imc)$,
      otherwise $\{\tau\}^\Imc = \emptyset$.
  \end{itemize}
  %
  The interpretation of the singletons is well
  defined. Indeed, 
  if
  $a^\Imc = [C_1]$ and $a^\Imc = [C_2]$ then
  $\{a\}\in S(C_1)\cap S(C_2)$ and thus $C_1 \thicksim C_2$. So
  $[C_1] = [C_2]$. As for role names, the fact that their interpretation is
  well defined can be proved similarly to~\cite{BaaEtAl05}.
  The mapping of $\cdot^\Imc$ for the universal role is 
  also standard.
  %
  %
%
  \begin{claim}\label{cl:aux2}
    Let $C\in S(C')$. Then,
    \begin{enumerate}
    \item for any $D\in \Bmc\Cmc_\Omc^+$
        with $D\in S(C)$,
      $D\in S(C')$;
    \item for any $r\in \Rmc_\Omc$ with $r\in R(C,D)$,
      $r\in R(C',D)$. 
    \end{enumerate}
  \end{claim}
  \begin{proof}
    The proof is by induction on the generation of the sequences
    $S_0,\ldots,S_m$ and $R_0,\ldots,R_m$. For Point (1), we consider
    the rules that add $D$ to $S(C)$. The base case ($n=0$) easily
    follows.
    \begin{description}
    \item[${\sf R}_1$.] There is $B\sqcap E\sqsubseteq D \in \Omc$ and
      $B,E\in S_{n-1}(C)$. Then, by induction, $B,E\in S_{n-1}(C')$
      and, by an application of ${\sf R}_1$, $D\in S_{n}(C')$.
    \item[${\sf R}_4,{\sf R}_6,{\sf R}_7,{\sf R}_9$.] These cases can be proved
      similarly to ${\sf R}_1$.
    \item[${\sf R}_3$.] There is $\exists r.B\sqsubseteq D \in \Omc$,
      $B\in S_{n-1} (E)$, $r\in\Rmc_{n-1}(C,E)$, and ${\sf A}\leadsto_R
      E$. By induction, $r\in\Rmc_{n-1}(C',E)$ 
      and, by an application of ${\sf R}_3$, $D\in S_{n}(C')$.
      \item[${\sf R}_5$.] There  is $r\in R_{n-1}(C,D)$ and $\bot\in
        S_{n-1}(D)$.  By induction, $r\in\Rmc_{n-1}(C',D)$ and, by an
        application of ${\sf R}_5$, $\bot\in S_n(C')$.
    \end{description}
    As for Point (2), the base case easily holds since, for all
    $C,D\in  \Bmc\Cmc_\Omc\setminus{\{\sf A\}}$, we have that
      $R_0(C,D) = \{u\}$, while   $R_0({\sf A}, {\sf A}) =
      \{u,u_u\}$ and $S_0({\sf A}) = \{{\sf A},\top\}$.
    We thus consider the only rule that extends  
      $R$.
    \begin{description}
    \item[${\sf R}_2$.] There is $B\sqsubseteq \exists r.D \in \Omc$,
       $B\in S_{n-1}(C)$, and ${\sf A}\leadsto_R C$. By induction, $B\in S_{n-1}(C')$ and, by
      an application of ${\sf R}_2$, $r\in R_n(C',D)$. 
    \end{description}%
  \end{proof}%
  \noindent
  We now show three properties of elements $[D^c]\in \Delta^\Imc$ that
  are used to show Claim~\ref{cl:aux3}. The first is formulated as
  follows.
  \begin{align}
    \label{eq:2}
    \text{For } [D], [D^c]\in \Delta^\Imc, \text{we have } [D] \neq [D^c].
  \end{align}%
  Indeed, if by contradiction $[D] = [D^c]$, i.e., $D \thicksim D^c$,
  then there must be a term $\tau$ such that $\{\tau\} \in S(D^c)\cap S(D)$ and thus,
  by~${\sf R}_8$, $D^c\not\in V$. This contradicts the assumption
  that $[D^c]\in \Delta^\Imc$. The second property is stated  in
  the next claim.
  \begin{claim}\label{cl:pro3}
    For $[C^c]\in \Delta^\Imc$, 
    the following holds:
    \begin{enumerate}
    \item if $D\in S(C^c)$, with $D\neq C^c$, then $D\in S(C)$;
    \item if $r\in R(C^c,D)$, then $r\in R(C,D)$. 
    \end{enumerate}
  \end{claim}
  \begin{proof}
    The proof is by induction on 
    $S_0,\ldots,S_m$ and $R_0,\ldots,R_m$. For Point (1), we consider
    the rules that add $D$ to $S(C^c)$. The base case ($n=0$) easily
    follows.
    \begin{description}
    \item[${\sf R}_1$.] There is $C\sqcap B\sqsubseteq D \in \Omc$ and
      $C,B\in S_{n-1}(C^c)$. By induction, $C,B\in S_{n-1}(C)$, and,
      by an application of ${\sf R}_1$, $D\in S_n(C)$.
    \item[${\sf R}_3$.] There is $\exists r.E\sqsubseteq D \in \Omc$,
      $E\in S_{n-1}(B)$, $r\in R_{n-1}(C^c,B)$, and
      ${\sf A}\leadsto_R B$. By induction, $r\in R_{n-1}(C,B)$, and,
      by an application of ${\sf R}_3$, $D\in S_n(C)$.
    \item[${\sf R}'_3$.] Simply notice that if $C^c\in V$ then
      $C\in V$.
    \item[${\sf R}_5$.]  There is $r\in R_{n-1}(C^c,B)$ with
      $\bot \in S_{n-1}(B)$. By induction, $r\in R_{n-1}(C,B)$, and,
      by an application of ${\sf R}_5$, $\bot\in S_n(C)$.
    \item[${\sf R}_4$, ${\sf R}_6$.] These rules are not
      applicable. Indeed,  
      in both cases
      $\{\tau\}\in S_{n-1}(C^c)$ and, by induction,
      $\{\tau\}\in S_{n-1}(C)$. But, by ${\sf R}_8$, $C^c\not\in V$,
      contradicting the fact that $[C^c]\in \Delta^\Imc$.
    \item[${\sf R}_7$.] This rule is not applicable. Indeed, there
      must be $D\sqsubseteq \{\tau\}\in \Omc$ and $D\in
      S_{n-1}(C^c)$. Then, by induction, $D\in S_{n-1}(C)$ and, by
      applying ${\sf R}_7$, $\{\tau\}\in S_{n-1}(C)$. Again, by
      ${\sf R}_8$, $C^c\not\in V$, contradicting the fact that
      $[C^c]\in \Delta^\Imc$.
    \item[${\sf R}_9$.] This rule is not applicable. Indeed, there
      must be $B\in S_{n-1}(C^c)$ and $B^c\notin V_{n-1}$. Then, by
      induction, $B\in S_{n-1}(C)$ and, by applying ${\sf R}_9$,
      $\{\iota B\}\in S_n(C)$.  Again, by ${\sf R}_8$, $C^c\not\in V$,
      contradicting the fact that $[C^c]\in \Delta^\Imc$.
      \end{description}
      As for Point (2),  the base case holds, since, for all $C,D\in \Bmc\Cmc_\Omc$,
      we have that $R_0(C,D) = \emptyset$. 
      The only rule that adds an
      $r$ to the set $R$ is the rule ${\sf R}_2$ which adds $r$ to
      both $R(C^c,D)$ and $R(C,D)$.
%
      %
    \end{proof}
    The following is an immediate consequence of the above claim.
    \noindent
    \begin{corollary}\label{cor}
      Let $[C^c]\in \Delta^\Imc$, then it is never the case that
      $\{\tau\}\in S(C)$, for any term $\tau$.
    \end{corollary}
    \noindent
    Claim~\ref{cl:pro3-1} shows that $\bot$
    cannot label any node.
  \begin{claim}\label{cl:pro3-1}
    $\bot\not\in S(C)$, for all $C\in \ra\cup \ra^c$.
  \end{claim}
  \begin{proof}
    Let $C\in \ra$. The proof is by induction on the number, $k$, of
    concepts in the definition of ${\sf A}\leadsto_R C$, using the
    assumption that $\bot\not\in S({\sf A})$. The base case,
    $k=0,~ C\in S({\sf A})$, holds by the assumption that
    $\bot\not\in S({\sf A})$ and by Claim~\ref{cl:aux2}. Let the Claim
    holds for $k-1$, thus $\bot\not\in S(C_{k-1})$, and let
    $r\in R(C_{k-1},C_k)$. If, by absurd, $\bot\in S(C)$, for some
    $C\in S(C_k)$, then by Claim~\ref{cl:aux2}, $\bot\in S(C_k)$ and,
    by applying ${\sf
      R}_5$, $\bot\in S(C_{k-1})$ which is a contradiction.
    Let $C = D^c\in \ra^c$. Then, $D\in \ra$  and by Claim~\ref{cl:pro3},
    $\bot\not\in S(D^c)$, for otherwise we would contradict the fact
    that $\bot\not\in S(D)$, for any $D\in\ra$.
  \end{proof}
  We are now able to show the following claim.
  \begin{claim}\label{cl:aux3}
    For all $\eq{C}\in\Delta^\Imc$ and $D\in\Bmc\Cmc^{+}_\Omc$, we
    have that $\eq{C}\in D^\Imc\text{ iff }D\in S(C)$.
  \end{claim}
  \begin{proof}
    The case where $D$ is a concept name or of the form $\{a\}$ is by
    definition of \Imc.
    %
    If $D=\top$ then we are done since $\top\in S(C)$ for all $C$ with
    $\eq{C}\in\Delta^\Imc$.
    %
    If $D=\bot$, then, by Claim~\ref{cl:pro3-1}, $\bot\not\in S(C)$,
    for all $C\in \ra\cup \ra^c$.
    %
    The interesting case is when $D=\{\iota A\}$. 
    We now argue that the claim also holds in this case. 
    \begin{itemize}
    \item Suppose $D=\{\iota A\}$. If $\eq{C}\in \{\iota A\}^\Imc$
      then, by definition of \Imc, $\eq{C} = (\iota A)^\Imc$ and
      $A^\Imc=\{\eq{C}\}$. We first exclude that $C\in\ra^c$.  Indeed,
      if, by absurd, $A^\Imc=\{\eq{D^c}\}$, then, by definition of
      \Imc, $A\in S(D^c)$ and, by Claim~\ref{cl:pro3}, $A\in S(D)$
      which in turn implies that $[D]\in A^\Imc$. By
      property~\eqref{eq:2}, $[D] \neq [D^c]$, contradicting the fact
      that $A^\Imc$ is a singleton. Thus, $C\in\ra$ and, by definition
      of $A^\Imc$, $A\in S(C)$. Since $C\in\ra$, then, by the
      definition of $\leadsto_R$, $A\in\ra$, and thus,
      $[A]\in\Delta^\Imc$. Then, since $A\in S(A)$, by definition of
      \Imc, $[A]\in A^\Imc$. By similar considerations as above, we
      can show that $A^c\not\in V$ (for otherwise, $[A] \neq [A^c]$
      and they should both belong to $A^\Imc$). Since $A\in S(C)$ and
      $A^c\not\in V$, then, by~${\sf R}_9$, $\{\iota A\}\in S(C)$.
    \item Conversely, assume $\{\iota A\}\in S(C)$. By
      Corollary~\ref{cor}, $C$ cannot be of the form $D^c$.
      By Claim~\ref{cl:pro3},
      $\{\iota A\}\in S(D)$, and thus, by~${\sf R}_{10}$,
      $D^c\not\in V$, which contradicts the fact that $D^c\in V$.
    Thus, we can assume that $C\in\ra$, and then, by~${\sf R}_{10}$,
    $A^c\not\in V$. 
    Assume that $A^\Imc = \{[C]\}$. Then, by definition of \Imc,
    $(\iota A)^\Imc = [C]$, and so, $ [C]\in \{\iota A\}^\Imc$, as
    required. We thus show that $A^\Imc = \{[C]\}$. We first show that
    $[C]\in A^\Imc$. By assumption $\{\iota A\}\in S(C)$, while
    $\{\iota A\}\sqsubseteq A\in \Omc$. Then, by~${\sf R}_{6}$,
    $A\in S(C)$, and, by definition of \Imc, $[C]\in A^\Imc$. Now, by
    contradiction, let us assume that $[D]\in A^\Imc$, for some
    $D\not\in [C]$. By definition of \Imc, $A\in S(D)$, and
    since $A^c\not\in V$, by~${\sf R}_9$, $\{\iota A\}\in S(D)$. By
    definition of $\thicksim$, $D \thicksim C$. This contradicts
    $D\not\in [C]$.
    \end{itemize}
\end{proof}

We conclude the proof by showing that \Imc is a model of \Omc such
that $\eq{{\sf A}}\in {\sf A}^\Imc$ but
$\eq{{\sf A}}\not\in B^\Imc$. By definition of $\leadsto_R$, we
have that ${\sf A} \leadsto_R {\sf A}$, so
$\eq{{\sf A}}\in \Delta^\Imc$ and, by definition of \Imc and $S$, we
have that $\eq{{\sf A}}\in {\sf A}^\Imc$.  By assumption,
$B\not\in S({\sf A})$ and, by Claim~\ref{cl:aux3},
$\eq{{\sf A}}\not\in B^\Imc$.  It remains to show that \Imc is a
model of \Omc. We make a case distinction.
\begin{itemize}
\item $C\sqcap D\sqsubseteq E$. Let $\eq{B}\in (C\sqcap D)^\Imc$,
  then, $\eq{B}\in C^\Imc$ and $\eq{B}\in D^\Imc$. By
  Claim~\ref{cl:aux3}, we have $C,D\in S(B)$. By~${\sf R}_1$,
  $E \in S(B)$, and thus, by Claim~\ref{cl:aux3}, $\eq{B}\in
  E^\Imc$.
\item $C\sqsubseteq \exists r.D$.  Let
  $\eq{B}\in C^\Imc$, then, by Claim~\ref{cl:aux3} $C\in S(B)$. We
  first consider the case where $B\in\ra$. Then, $A\leadsto_RB$ and,
  by applying ${\sf R}_2$, $r\in R(B,D)$, which also implies
  $D\in\ra$. By definition of $r^\Imc$, we have that
  $(\eq{B},\eq{D})\in r^\Imc$. By definition of $S$, $D\in S(D)$ and,
  by Claim~\ref{cl:aux3}, $\eq{D}\in D^\Imc$. Thus,
  $\eq{B}\in (\exists r.D)^\Imc$.
  Let $B\in\ra^c$, i.e., $B$ is of the form $E^c$, with
  $E\in\ra$. Since $C\in S(E^c)$, by Claim~\ref{cl:pro3}, then,
  $C\in S(E)$.  By applying ${\sf R}_2$, $r\in R(E^c,D)$ and, by
  definition of $r^\Imc$, we have that $(\eq{E^c},\eq{D})\in
  r^\Imc$. As before, $\eq{D}\in D^\Imc$, and thus
  $\eq{B}=\eq{E^c}\in (\exists r.D)^\Imc$.
  %
\item $\exists r.C\sqsubseteq D$. Let $\eq{B}\in (\exists
  r.C)^\Imc$. Then, there is $\eq{F}\in\Delta^\Imc$ such that
  $(\eq{B},\eq{F})\in r^\Imc$ and $\eq{F}\in C^\Imc$.  Let $r\neq
  u$. Then, by $r^\Imc$ definition, there is $F'\in\eq{F}$ such that
  $r\in R(B,F')$, while, by Claim~\ref{cl:aux3}, $C\in S(F)$.
  By~\eqref{eq:1}, $S(F') = S(F)$, thus, $C\in S(F')$. We can now
  apply~${\sf R}_3$ by which $D\in S(B)$. So, by Claim~\ref{cl:aux3},
  $\eq{B}\in D^\Imc$.
%
  %
  Let $r=u$ and $\eq{B}\in (\exists
  u.C)^\Imc$. Then, there is $\eq{F}\in\Delta^\Imc$ such that
  $\eq{F}\in C^\Imc$, and, by Claim~\ref{cl:aux3}, $C\in S(F)$. We
  first consider the case where $\eq{F}\in\ra$. Then, $A\leadsto_R F$
  and thus $A\leadsto_R C$. By applying ${\sf R}'_3$, $D\in S(B)$ and,
  by Claim~\ref{cl:aux3}, $\eq{B}\in D^\Imc$. If $F = E^c\in\ra^c$,
  with $E\in\ra$, by Claim~\ref{cl:pro3}, since $C\in S(E^c)$, then
  $C\in S(E)$. Thus, as before $A\leadsto_R C$ and by applying  ${\sf
    R}'_3$ we can conclude that  $\eq{B}\in D^\Imc$.
\item $\{\tau\}\sqsubseteq D$. Let $\eq{B}\in \{\tau\}^\Imc$.  By
  Claim~\ref{cl:aux3}, $\{\tau\}\in S(B)$. By~${\sf R}_6$, $D\in S(B)$,
  and thus, by Claim~\ref{cl:aux3}, $\eq{B}\in D^\Imc$.
\item $D\sqsubseteq \{\tau\}$. Let $\eq{B}\in D^\Imc$ then, by
  Claim~\ref{cl:aux3}, $D\in S(B)$.  By~${\sf R}_7$,
  $\{\tau\}\in S(B)$, and,  
by Claim~\ref{cl:aux3}, $\eq{B}\in \{\tau\}^\Imc$.
\end{itemize}
This finishes the proof of the lemma. 
\end{proof}

We  
are now able to prove the following lemma,  used in the proof of Theorem~\ref{theo:eloufincanmod}.
In the following, we write $\Omc, {\sf A}\models \alpha$ meaning:
in all interpretations that satisfy \Omc and ${\sf A}$, we have that $\alpha$ holds.

\begin{lemma}\label{lem:invariantsELO}
  Let  $S, R$ be the label sets of a completed
  classification graph for $\Omc$ and ${\sf A}$, $C,D\in \Bmc\Cmc_{\Omc}$ and
  $E\in\Bmc\Cmc_{\Omc}^+$.
  The following invariants
  hold:
\begin{enumerate}
\item $E \in S(C)$ implies $\Omc, {\sf A}\models C \sqsubseteq E$
%
\item $r \in R(C, D)$ implies $\Omc, {\sf A}\models C \sqsubseteq \exists r.D$,
  for all $r\in\Rmc_\Omc$.
\end{enumerate}
\end{lemma}
\begin{proof}
  Point (1) follows from Claim~\ref{cl:aux1},~Point (1).
  As for Point (2), by ${\sf R}_2$, if $r\in R(C,D)$, then there exists
  $E\in \Bmc\Cmc_{\Omc}$ with $E\sqsubseteq \exists r.D\in\Omc$ and
  $E\in S(C)$.  Then, by Point~(1),
  $\Omc, {\sf A}\models C\sqsubseteq \exists r.D$.
\end{proof}

\eloufincanmod*
\begin{proof}
First, notice that $\fincanmod$ is the same as the interpretation $\Imc$ defined in the ($\Rightarrow$) direction proof of Lemma~\ref{lem:completion},
showing that $\fincanmod$ is
a model of $\Omc$ and {\sf A}.

  ($\Leftarrow$)
  We prove the following more general statement.
  \begin{claim}
  For every $[D] \in \Delta^{\fincanmod} \cap \Rmc_{{\mathsf A}}$ 
  and every $\ELOu$ concept $C$, we have that $\Omc \models D \sqsubseteq C$ implies $[D] \in C^{\fincanmod}$.
  \end{claim}
  \begin{proof}
  By ${\fincanmod}$
  construction, $[D] \in D^{\fincanmod}$.
  Since
  ${\fincanmod}$ is a model of $\Omc$, if $\Omc \models D \sqsubseteq C$
  then $[D] \in C^{\fincanmod}$.
  \end{proof}
  ($\Rightarrow$) 
  We first show Claim~\ref{aux}.  The proof is by induction on the
  construction of $C$.
  %
  \begin{claim}\label{aux}
  For every $[D] \in \Delta^{\fincanmod} \cap \Rmc_{{\mathsf A}}$ 
  and every $\ELOu$ concept $C$, if $[D] \in C^{\fincanmod}$ then $\Omc, {\sf A} \models D \sqsubseteq C$.
\end{claim}
\begin{proof} 
The base case $C = \bot$ is vacuously true. The base case $C = \top$
is obviously true.  For the other base cases, $C\in \NC$ and
$C = \{ a \}$, it
follows from 
Claim~\ref{cl:aux3} and Point~(1) of Lemma~\ref{lem:invariantsELO}.
%
We show the remaining cases.



$C = \exists r.E,~r\neq u$. If $[D] \in (\exists r.E)^{\fincanmod}$
then there exists $[F]$ such that $([D], [F]) \in r^{\fincanmod}$ and
$[F] \in E^{\fincanmod}$. By ${\fincanmod}$ construction, there exists
$F'\in [F]\cap \Rmc_{{\mathsf A}}$ s.t. $r\in R(D,F')$. By Point~(2) of
Lemma~\ref{lem:invariantsELO},
$\Omc, {\sf A}\models D\sqsubseteq \exists r.F'$. On the other hand,
$[F'] = [F]$, thus $[F'] \in E^{\fincanmod}$ and, by i.h.,
$\Omc, {\sf A}\models F'\sqsubseteq E$, and then,
$\Omc, {\sf A}\models D\sqsubseteq \exists r.E$.

$C = \exists u.E$.  If $[D] \in (\exists u.E)^{\fincanmod}$ then there
exists $[F]$ such that $[F] \in E^{\fincanmod}$.  By construction of
$\fincanmod$, we can also assume that $[F]\in\Rmc_{{\mathsf A}}$ and,
by i.h., $\Omc, {\sf A}\models F\sqsubseteq E$. To prove that
$\Omc, {\sf A}\models D\sqsubseteq \exists u.E$ it is enough to show
that $\Omc, {\sf A}\models D\sqsubseteq \exists u.F$. Since,
$[D], [F]\in \Rmc_{{\mathsf A}}$, then ${\sf A}\leadsto_\Rmc D$ and
${\sf A}\leadsto_\Rmc F$. Then, by Point~(2) of
Lemma~\ref{lem:invariantsELO},
$\Omc, {\sf A} \models {\sf A}\sqsubseteq \exists r_1...\exists r_m.D$
and
$\Omc, {\sf A} \models {\sf A}\sqsubseteq \exists r_1...\exists
r_n.F$, with $m,n\geq 0$. Thus we finally have
$\Omc, {\sf A} \models D\sqsubseteq \exists u.F$.

$C = E \sqcap F$. If $[D] \in (E \sqcap F)^{\fincanmod}$ then
$[D] \in E^{\fincanmod}$ and $[D] \in F^{\fincanmod}$. By i.h.,
$\Omc, {\sf A} \models D \sqsubseteq E$ and $\Omc, {\sf A} \models D \sqsubseteq F$,
i.e., $\Omc, {\sf A} \models D \sqsubseteq E \sqcap F$.
\end{proof}
If  $[{\sf A}] \in C^{\fincanmod}$ then, by Claim~\ref{aux}, we have that $\Omc, {\sf A} \models {\sf A} \sqsubseteq C$.
This happens iff $\Omc\models {\sf A} \sqsubseteq C$.
This ends the proof of Theorem~\ref{theo:eloufincanmod}.
\end{proof}



\subsection*{Proofs for Section~\ref{sec:alcoiotabisim}}

Before we proceed with the proofs of Theorems~\ref{thm:alcoiotabisimtoequiv} and~\ref{thm:alcofirstorder}, we provide the definitions of $\omega$-saturated partial interpretations,
defined for first-order formulas. 

Let $\Imc$ be a partial interpretation, regarded here also as a \emph{partial first-order interpretation} with identity.
The semantic definitions given for
the DLs with nominals considered in this paper
naturally extend to the first-order case, so that, in particular, the \emph{satisfaction} in $\Imc$ under a variable assignment $\assign$ for atomic first-order formulas
is given as follows:
\begin{align*}
	\Imc \models^{\assign}  P(t_1, \ldots, t_n) 
		\ \text{ iff } \ & 
		 \tvalue(t_{i}), \, 1 \leq i \leq n, \ \text{is defined and} \\
	  & (\tvalue(t_1), \ldots, \tvalue(t_n)) \in P^{\Imc},
		\\
		\Imc \models^{\assign}  t_1 =  t_2
		\ \text{ iff } \ &
		 \tvalue(t_{i}), \, i = 1, 2,  \ \text{is defined and} \\
		& \tvalue(t_1) = \tvalue(t_2),
\end{align*}
where $\tvalue(t)$ is the \emph{value} of a \emph{term} $t$ (that is, either a variable or an individual name in $\NI$) in $\Imc$ under the assignment $\sigma$, i.e., the image of the \emph{partial} function from the set of terms to $\Delta^{\Imc}$ such that
\begin{gather*}
\tvalue(t) =
\begin{cases}
\assign(x), & \text{ if } t \text{ is a variable } x; \\
a^{\Imc}, & \text{ if } t \text{ is } a \in \NI \cap \dom(\cdot^{\Imc}).
\end{cases}
\end{gather*}
Recall that $\dom(\cdot^{\Int}) \subseteq \NC\cup\NR\cup \NI$ is the domain of definition of $\cdot^\Imc$, with $\NC\cup\NR \subseteq \dom(\cdot^{\Int})$, and note that we do not include definite descriptions as terms of the first-order language.
Together with the other usual inductive clauses for Booleans and quantifiers, we obtain what is called in the literature a \emph{negative semantics} for first-order logic on partial interpretations~\cite{Leh02}.
%
We assume in this section that the sets $\NC \cup \NR \cup \NI$ and $\Delta^{\Imc}$ are disjoint.
Moreover, we consider each element $d \in \Delta^{\Imc}$ as an additional individual symbol such that $d^{\Imc} = d$.
%
Let $\Gamma$ be a set of first-order formulas with free variables among $x_{1}, \ldots, x_{n}$, predicate symbols from $\NC \cup \NR$, and individual symbols from $\NI \cup \Delta^{\Imc}$.
We say that $\Gamma$ is:
\begin{itemize}
	\item \emph{finitely realisable} in $\Imc$ iff, for every finite subset $\Gamma' \subseteq \Gamma$, there exists a variable assignment 
	 $\assign$ such that $\assign(x_{i}) \in \Delta^{\Imc}$, with $1 \leq i \leq n$, such that $\Imc \models^{\assign} \Gamma'$.
	\item \emph{realisable} in $\Imc$ iff there exists a variable assignment $\assign$ such that $\assign(x_{i}) \in \Delta^{\Imc}$, with $1 \leq i \leq n$, such that $\Imc \models^{\assign} \Gamma$
\end{itemize}
We say that $\Imc$ is \emph{$\omega$-saturated} iff, for every such set $\Gamma$ containing only finitely many individual symbols from $\Delta^{\Imc}$, the following holds:
if $\Gamma$ is finitely realisable in $\Imc$, then $\Gamma$ is realisable in $\Imc$.




In our proof of Theorem~\ref{thm:alcoiotabisimtoequiv},
we also use the following technical lemma.

\begin{lemma}
\label{lemma:bisimnotuni}
Let
$\Sigma$ be a signature. For every partial interpretation $\Imc$ and $d, d' \in \Delta^{\Imc}$
	\[
	\text{if} \
	(\Imc, d) \sim^{\ALCO}_{\Sigma} (\Imc, d'),
	\ \text{then} \
	(\Imc, d) \equiv^{\ALCOud}_{\Sigma} (\Imc, d').
	\]
%
\end{lemma}
\begin{proof}
We show that, for every $d, d' \in \Delta^{\Imc}$ such that $(\Imc, d) \sim^{\ALCO}_{\Sigma} (\Imc, d')$, it holds that, for every
$\ALCOud(\Sigma)$
concept $C$, $d \in C^{\Int}$ iff $d' \in C^{\Int}$.
The proof is by structural induction on $C$.
For the base cases $C = A$ and $C= \{ a \}$
and for the inductive cases $C = \lnot D$, $C = D \sqcap E$, $C = \exists r.D$, the proof is
straightforward.
We show the remaining cases.

Let $C = \exists u.D$, and suppose that $d \in \exists u.D^{\Imc}$. This implies that $\exists u.D^{\Imc} = \Delta^{\Imc}$, and thus $d' \in \exists u.D^{\Imc}$ as well. The converse direction is analogous.

Let $C = \{ \defdes D \}$, and suppose that $d \in \{ \defdes D \}^{\Imc}$. This is equivalent to $D^{\Imc} = \{ d \}$, and thus we have, in particular, $d \in D^{\Imc}$. By i.h., we obtain $d' \in D^{\Imc}$, meaning that $d' = d$, and hence $d' \in \{ \defdes D \}^{\Imc}$.
The converse direction is analogous.
%
\end{proof}

In the following, given a partial interpretation $\Imc$ and $d \in \Delta^{\Imc}$, we let $[x \mapsto d]$ stand for any variable assignment that maps $x$ to $d$ in $\Imc$.
Lemma~\ref{lemma:standardtr} below can be proved by induction on the structure of the concept.
\begin{restatable}{lemma}{standardtr}
\label{lemma:standardtr}
For every
$\ALCOud$
concept $C$, partial interpretation $\Int$ and $d \in \Delta^{\Int}$, we have that
$d \in C^{\Int}$ iff $\Int, [x \mapsto d] \models \sttr{x}{C}$.
\end{restatable}

We are now ready for the proof of Theorem~\ref{thm:alcoiotabisimtoequiv}.
%

\alcoiotabisimtoequiv*

\begin{proof}
(1)~Suppose there is an $\ALCOud(\Sigma)$ bisimulation $Z$ between $\Imc$ and $\Jmc$ such that $(d, e) \in Z$.
We show that, for every $\ALCOud$ concept $C$ with $\sig{C} \subseteq \Sigma$,
and every $u \in \Delta^{\Imc}$, $v \in \Delta^{\Jmc}$ such that $(u,v) \in Z$,
we have $u \in C^{\Int}$ iff $v \in C^{\Int}$.
The proof is by structural induction on $C$.
The base cases $C = A$ and $C = \{ a \}$, as well as the inductive cases for $C = \lnot D$, $C = D \sqcap E$, $C = \exists r . D$, and $C = \exists u. D$,
are as for $\ALCOu$~\cite[Theorem 4.1.2]{Ten05}.
We only need to prove the   case $C = \{ \defdes D \}$.
  
$(\Rightarrow)$
Suppose that $u \in \{ \defdes D \}^{\Imc}$. This means that $D^{\Imc} = \{ u \}$, and thus $u \in D^{\Imc}$.
By i.h., $v \in D^{\Jmc}$.
We want to show that $D^{\Jmc} = \{ v \}$.
Towards a contradiction, suppose there is $v' \neq v$ such that $v' \in D^{\Jmc}$.
We have two possibilities.

\begin{enumerate}
\item If $(\Jmc, v) \sim^{\ALCO}_{\Sigma} (\Jmc, v')$,
then by~($\defdes$)
there is $u' \neq u$ such that $(\Imc, u) \sim^{\ALCO}_{\Sigma} (\Imc, u')$.
By Lemma~\ref{lemma:bisimnotuni}, we have that $u' \in D^{\Imc}$, which is a contradiction.

\item If $(\Jmc, v) \not \sim^{\ALCO}_{\Sigma} (\Jmc, v')$, we can distinguish two cases.

\begin{enumerate}
	\item[2.1.] Suppose there exists $v'' \neq v$ such that $(\Jmc, v) \sim^{\ALCO}_{\Sigma} (\Jmc, v'')$. Again by~($\defdes$),
	there is $u'' \neq u$ such that $(\Imc, u) \sim^{\ALCO}_{\Sigma} (\Imc, u'')$. This, by Lemma~\ref{lemma:bisimnotuni}, implies that $u'' \in D^{\Imc}$, which is impossible.
	\item[2.2.] Suppose there is no $v'' \neq v$ such that $(\Jmc, v) \sim^{\ALCO}_{\Sigma} (\Jmc, v'')$. 
	By totality of $Z$,
	we have that for every $w \in \Delta^{\Jmc}$ there is $t \in \Delta^\Imc$ such that
	$(t, w) \in Z$.
	Thus, in particular,
	$(u', v') \in Z$,
	for some $u' \in \Delta^{\Imc}$.
	Suppose that $u' = u$.
	Since $Z$ is an $\ALCO(\Sigma)$ bisimulation,
	by transitivity of $\ALCO(\Sigma)$ bisimilarity, 
	$(\Jmc, v) \sim^{\ALCO}_{\Sigma} (\Jmc, v')$, contradicting our hypothesis.
	%
	Suppose then that $u' \neq u$. By i.h.,
	$u' \in D^{\Imc}$, which is again a contradiction.
\end{enumerate}
\end{enumerate}
In conclusion, there is no $v' \neq v$ such that $v' \in D^{\Jmc}$, and thus $v \in \{ \defdes D \}^{\Jmc}$.  
The direction $(\Leftarrow)$ is obtained analogously, by
using Conditions~($\defdes$)
and
totality of $Z$.



%

%

(2)~Let
$S = \{ (u, v) \in \Delta^{\Imc} \times \Delta^{\Jmc} \mid (\Imc, u) \equiv^{\ALCOud}_{\Sigma} (\Jmc, v) \},$
and suppose that $(d,e) \in S$.
We want to show that if \Imc and \Jmc are $\omega$-saturated 
then $S$ is an $\ALCOud(\Sigma)$ bisimulation between $\Int$ and $\Jmc$, and thus
$(\Imc, d) \sim^{\ALCOud}_{\Sigma} (\Jmc, e)$, as required.
%
For
Conditions~(\textit{atom}),~(\textit{back}),~(\textit{forth}) and totality,
the proof is analogous to the ones for $\ALCO$ and $\ALCOu$~\cite{AreEtAl01,Ten05}. 
%
We now prove that $S$ satisfies also Condition~($\defdes$).

	For the $(\Rightarrow)$ direction, suppose there exists $d' \in \Delta^{\Imc}$ such that $d \neq d'$ and $(\Imc, d) \sim^{\ALCO}_{\Sigma} (\Imc, d')$. Let $u$ be such an element, and take an individual variable $x$.
	Consider the set of first-order formulas
	$
	\Gamma = \{ \lnot( e = x) \} \cup T_{u},
	$
	with $T_{u} = \{ \sttr{x}{C} \mid C \in \tp^{\Imc}_{\Lmc}(u) \}$,
	and let $\Gamma'$ be the set $\Gamma$ with $e$ replaced by $d$.
	The assignment $\assign(x) = u$ makes $\Gamma'$ realisable in $\Imc$.
	
	We first show that, since $\Jmc$ is $\omega$-saturated and $(\Imc, d) \equiv^{\Lmc}_{\Sigma} (\Jmc, e)$, the set $\Gamma$ is realisable in $\Jmc$.
	Let $\Gamma_{0}$ be a finite subset of $\Gamma$ (without loss of generality, we assume that $\lnot( e = x) \in \Gamma_{0}$), and let $\Gamma'_{0}$ be the set $\Gamma_{0}$ with $e$ replaced by $d$.
	Since $\Gamma'_{0} \subseteq \Gamma'$, we have that also $\Gamma'_{0}$ is realisable in $\Imc$ by $\assign(x) = u$.
	Now consider the $\Lmc(\Sigma)$ concept
$
		C_{0} = \bigsqcap_{\sttr{x}{C} \in \Gamma'_{0}} C,
$
i.e., the conjunction of $\Lmc(\Sigma)$ concepts with standard translation in $\Gamma'_{0}$.
We have that
$\Imc, [x \mapsto u] \models \lnot (d = x) \land \sttr{x}{C_{0}}$,
and thus, by Lemma~\ref{lemma:standardtr}, $u \in C_{0}^{\Imc}$.
Given that $d \neq u$ and $(\Imc, d) \sim^{\ALCO}_{\Sigma} (\Imc, u)$, by Lemma~\ref{lemma:bisimnotuni} we have
$(\Imc, d) \equiv^{\Lmc}_{\Sigma} (\Imc, u)$.
From Lemma~\ref{lemma:standardtr},
we obtain
$
\Imc, [y \mapsto d, x \mapsto u] \models \sttr{y}{C_{0}} \land \sttr{x}{C_{0}} \land \lnot(y = x),
$
which implies 
$
\Imc, [y \mapsto d] \models \sttr{y}{\lnot \{ \defdes C_{0} \} }.
$
The previous step means $d \in (\lnot \{ \defdes C_{0} \})^{\Imc}$, and since
$(d,e) \in S$,
we have that $e \in (\lnot \{ \defdes C_{0} \})^{\Jmc}$.
Moreover, $e \in C_{0}^{\Jmc}$ (because $d \in C_{0}^{\Imc}$ and
$(d,e) \in S$),
hence there exists $v \in \Delta^{\Jmc}$ such that $e \neq v$ and $v \in C_{0}^{\Jmc}$.
In conclusion,
$
\Jmc, [x \mapsto v] \models \lnot(e = x) \land \sttr{x}{C_{0}},
$
meaning that $\Gamma_{0}$ is realisable in $\Jmc$.
Given that $\Jmc$ is $\omega$-saturated, this implies that $\Gamma$ is realisable in $\Jmc$, as required.

The final step is to show that there exists $e' \in \Delta^{\Jmc}$ such that $e \neq e'$ and
$(\Jmc, e) \sim^{\ALCO}_{\Sigma} (\Jmc, e')$.
Let $\assign(x) = v$ be an assignment realising $\Gamma$ in $\Jmc$.
Since $T_{u} \subseteq \Gamma$, we have that $(\Imc, u) \equiv^{\ALCOud}_{\Sigma} (\Jmc, v)$, and so in particular $(\Imc, u) \equiv^{\ALCO}_{\Sigma} (\Jmc, v)$.
By hypothesis, we have that $(\Imc, d) \sim^{\ALCO}_{\Sigma} (\Imc, u)$ (hence, by Theorem~\ref{thm:alcoiotabisimtoequiv}, $(\Imc, d) \equiv^{\ALCO}_{\Sigma} (\Imc, u)$), and
$(d,e) \in S$
(thus, in particular, $(\Imc, d) \equiv^{\ALCO}_{\Sigma} (\Jmc, e)$).
By transitivity of $\equiv^{\ALCO}_{\Sigma}$, we obtain
$(\Jmc, e) \equiv^{\ALCO}_{\Sigma} (\Jmc, v)$.
Thanks to the proof for $\ALCO$ adapted to the case of partial interpretations~\cite{AreEtAl01,Ten05},
since $\Jmc$ is $\omega$-saturated,
this implies
$(\Jmc, e) \sim^{\ALCO}_{\Sigma} (\Jmc, v)$.

The $(\Leftarrow)$ direction is analogous.
\qedhere

\end{proof}

 The following is an adaptation of well-known results from first-order model theory 
 to the case of partial interpretations.
 
 \begin{theorem}
 \label{thm:omegasatexist}
  A set $\Gamma$ of first-order formulas is satisfiable on partial interpretations iff every finite subset of $\Gamma$ is satisfiable on partial interpretations. Moreover, 
 for every partial interpretation $\Int$, there exists a partial interpretation $\Int^{\ast}$ that is $\omega$-saturated and satisfies the same first-order sentences as $\Int$ (i.e., is \emph{elementarily equivalent} to $\Int$).
 \end{theorem}

\alcofirstorder*
\begin{proof}
The proof is an adaptation 
of the one in~\cite[Theorem 4]{LutEtAl11}
to the case of
$\ALCOud$ on partial interpretations.
The implication (1)~$\Rightarrow$~(2) follows from Theorem~\ref{thm:alcoiotabisimtoequiv} and Lemma~\ref{lemma:standardtr}.
We now show the direction (2)~$\Rightarrow$~(1).
Assume that $\p(x)$ is invariant under $\sim^{\ALCOud}_{\Sigma}$, and suppose towards a contradiction there there is no $\ALCOud(\Sigma)$ concept $C$ such that $\Imc \models \sttr{x}{C} \leftrightarrow \p(x)$, for every partial interpretation $\Imc$.
Define the set
$
\conseq{\p(x)} = \{ \sttr{x}{C} \mid \p(x) \models \sttr{x}{C}, C \ \ALCOud \ \text{concept}  \}.
$
By compactness of first-order logic on partial interpretations (Theorem~\ref{thm:omegasatexist}), the set
$\conseq{\p} \cup \{\lnot \p(x) \}$
is satisfiable.
Let $\Int^{-}$ be a partial interpretation satisfying $\conseq{\p} \cup \{\lnot \p(x) \}$ under the assignment $\assign$ such that $\assign(x) = d$.
By Theorem~\ref{thm:omegasatexist}, we can assume
without loss of generality
that $\Imc^{-}$ is $\omega$-saturated.

We now claim the following.

\begin{claim}
\label{claim:phitypesat}
The set
$\{ \p(x) \} \cup \{ \sttr{x}{C} \mid C \in \tp^{\Imc^{-}}_{\ALCOud}(d) \}$ is satisfiable.
\end{claim}
\begin{proof}
Suppose towards a contradiction that the statement does not hold.
By compactness (Theorem~\ref{thm:omegasatexist}), there is a finite set $\Gamma \subseteq \tp^{\Imc^{-}}_{\ALCOud}(d)$ such that $\{ \p(x) \} \cup \{ \sttr{x}{C} \mid C \in \Gamma \}$ is unsatisfiable.
This means that
$\models \p(x) \to \lnot \bigwedge_{C \in \Gamma} \sttr{x}{C}$,
and thus
$\lnot \bigwedge_{C \in \Gamma} \sttr{x}{C} \in \conseq{\p(x)}$.
However, since $\Imc, [x \mapsto d] \models \conseq{\p(x)}$, we have that 
$\conseq{\p(x)} \subseteq \{ \sttr{x}{C} \mid C \in \tp^{\Imc^{-}}_{\ALCOud}(d) \}$, hence a contradiction.
\end{proof}

Consider then an $\omega$-saturated partial interpretation $\Imc^{+}$ satisfying
$\{ \p(x) \} \cup \{ \sttr{x}{C} \mid C \in \tp^{\Imc^{-}}_{\ALCOud}(d) \}$
under the assignment $\assign'(x) = e$.
By definition and Lemma~\ref{lemma:standardtr}, we have that $(\Imc^{-}, d) \equiv^{\ALCOud}_{\Sigma} (\Imc^{+}, e)$.
Since $\Imc^{-}$ and $\Imc^{+}$ are both $\omega$-saturated, by Theorem~\ref{thm:equivomegatobisim} we get $(\Imc^{-}, d) \sim^{\ALCOud}_{\Sigma} (\Imc^{+}, e)$.
However, $\Imc^{-}, [x \mapsto d] \not \models \p(x)$, whereas $\Imc^{+}, [x \mapsto e] \models \p(x)$, contrary to the hypothesis that $\p(x)$ is invariant under $\sim^{\ALCOud}_{\Sigma}$.
\end{proof}

\elousimequiv*
\begin{proof}
The proof generalises the one for $\EL$ given in~\cite[Lemma 23]{LutEtAl11}, so to cover the cases of nominals and the universal roles on partial interpretations.

(1)~Assume that $(\Imc, d) \simul^{\ELOu}_{\Sigma} (\Jmc, e)$, i.e., there exists an $\ELOu(\Sigma)$ simulation $Z$ from $\Imc$ to $\Jmc$ such that $(d, e) \in Z$.
We show that, for every $\ELOu$ concept $C$ with $\sig{C} \subseteq \Sigma$,
and every $u \in \Delta^{\Imc}$, $v \in \Delta^{\Jmc}$ such that $(u,v) \in Z$,
we have that $u \in C^{\Int}$ implies $v \in C^{\Int}$.
The proof is by structural induction on $C$.
The base cases $C = A$,
$C = \bot$, $C = \top$,
are as in $\EL$, while
the base case of $C = \{ a \}$ follows immediately from
Condition
($\textit{atom}_{R}$)
for individual names.
The inductive cases $C = D \sqcap E$ and $C = \exists r . D$,
are as in $\EL$~\cite[Lemma 23]{LutEtAl11}.
For
the inductive case of $C = \exists u. D$, suppose that $u \in \exists u.D^{\Imc}$. This means that there exists $u' \in D^{\Imc}$. Since $Z$ satisfies left-totality,
there exists $v' \in \Delta^{\Jmc}$ such that $(u', v') \in Z$. By i.h., we obtain that $v' \in D^{\Jmc}$, and thus $v \in \exists u.D^{\Jmc}$.
This concludes the inductive proof. From the assumption that $(d,e) \in Z$, we obtain $\tp^{\Imc}_{\ELOu(\Sigma)}(d) \subseteq  \tp^{\Jmc}_{\ELOu(\Sigma)}(e)$.

(2)~Assume $\tp^{\Imc}_{\ELOu(\Sigma)}(d) \subseteq  \tp^{\Jmc}_{\ELOu(\Sigma)}(e)$ and that
$\Jmc$ is $\omega$-saturated.
Consider the relation $Z = \{ (u, v) \in \Delta^{\Imc} \times \Delta^{\Jmc} \mid  \tp^{\Imc}_{\Lmc(\Sigma)}(u) \subseteq \tp^{\Jmc}_{\Lmc(\Sigma)}(v) \}$.
We show that $Z$ is an $\ELOu(\Sigma)$ simulation from $\Imc$ to $\Jmc$.
%
Conditions
$(\textit{atom}_{R})$
for concept names
and $(\textit{forth})$ are as in the $\EL$ case, while Condition
$(\textit{atom}_{R})$
for individual names
follows immediately from the definition of $Z$.
It remains to show
that $Z$ satisfies
left-totality.
%
Suppose, towards a contradiction, that there exists a $w \in \Delta^{\Imc}$
such that, for every $t \in \Delta^{\Jmc}$, $(w,t) \not \in Z$,
meaning that, for every $t \in \Delta^{\Jmc}$, there exists $C \in \tp^{\Imc}_{\ELOu(\Sigma)}(w)$ so that $C \not \in \tp^{\Jmc}_{\ELOu(\Sigma)}(t)$.
%
%
Since $\Jmc$ is $\omega$-saturated,
it can be seen that there exists an $\ELOu(\Sigma)$ concept $C_{0}$ such that $\exists u.C_{0} \not \in \tp^{\Jmc}_{\ELOu(\Sigma)}(e)$, while $\exists u.C_{0} \in \tp^{\Imc}_{\ELOu(\Sigma)}(d)$, contradicting the assumption that $\tp^{\Imc}_{\ELOu(\Sigma)}(d) \subseteq \tp^{\Jmc}_{\ELOu(\Sigma)}(e)$.
\qedhere
\end{proof}

For the proofs
below,
we will use the following notation.
%
Given an $\Lmc$ ontology $\Omc$, a set of $\Lmc$ concepts $\Gamma$ and an $\Lmc$ concept $C$, we write $\Omc \cup \Gamma \models C$ iff,
for every partial interpretation $\Imc$ such that $\Imc \models \Omc$, the following holds, for every $d \in \Delta^{\Imc}$: if $d \in D^{\Imc}$, for all $D \in \Gamma$, then $d \in C^{\Imc}$.

\expdefchar*
\begin{proof}
%
We show that
$a$ is explicitly $\ALCOud(\Sigma)$ definable under $\Omc$ iff there are no pointed partial interpretations $(\Imc, d)$ and $(\Jmc, e)$ such that $\Imc$ and $\Jmc$ are models of $\Omc$, $d = a^{\Imc}$, $e \neq a^{\Jmc}$ and 
$(\Imc, d) \sim^{\ALCOud}_{\Sigma} (\Jmc, e)$.

$(\Rightarrow)$
Assume $a$ is explicitly $\ALCOud(\Sigma)$ definable under $\Omc$,
i.e., there exists an  $\ALCOud(\Sigma)$ concept $C$ such that $\Omc \models \{ a \} \equiv C$.
We have that 
$\{ a \}^{\Imc} = C^{\Imc}$, for every partial interpretation $\Int$ that is a model of $\Omc$.
Thus, if $(\Imc, d) \sim^{\Lmc}_{\Sigma} (\Jmc, e)$,
for a partial interpretation $\Jmc$ that is a model of $\Omc$, and $d = a^{\Imc}$,
we also have that $d \in C^{\Imc}$ and, by Theorem~\ref{thm:alcoiotabisimtoequiv}, Point~(1), that $e \in C^{\Jmc}$.
Therefore, $e = a^{\Jmc}$ and so $\Omc, \{ a \}$ and $\Omc, \lnot \{ a \}$ are not jointly consistent modulo $\ALCOud(\Sigma)$ bisimulations.

$(\Leftarrow)$
Conversely, assume that $a$ is not explicitly $\ALCOud(\Sigma)$ definable under $\Omc$, and let 
\[
\Gamma^{\Omc}_{\Sigma}(a) = \{ D \mid \Omc \models \{ a \} \sqsubseteq D, D \ \ALCOud(\Sigma) \ \text{concept} \}.
\]
As $a$ is not explicitly $\ALCOud(\Sigma)$ definable under $\Omc$, for every $D \in \Gamma^{\Omc}_{\Sigma}(a)$, we have that
$\Omc \not \models D \sqsubseteq \{ a \}$.
Since $\Gamma^{\Omc}_{\Sigma}(a)$ is closed under conjunctions, by compactness of first-order logic on partial interpretations~(of which $\ALCOud$ is a fragment, cf. Theorem~\ref{thm:omegasatexist}) we obtain that
$\Omc \cup \Gamma^{\Omc}_{\Sigma}(a) \not\models \{ a \}$, i.e., there exist a partial interpretation $\Jmc$ and an $e \in \Delta^{\Jmc}$ such that $\Jmc \models \Omc$ and $e \in D^{\Jmc}$, for all $D \in \Gamma^{\Omc}_{\Sigma}(a)$, but $e \neq a^{\Jmc}$.
Now, consider the $\ALCOud(\Sigma)$ type $t^{\Jmc}_{\ALCOud(\Sigma)}(e)$ of $e$ in $\Jmc$.
It can be seen that there exist a partial interpretation $\Imc$
and a $d \in \Delta^{\Imc}$ such that $\Imc \models \Omc$ and $d \in E^{\Imc}$, for all $E \in t^{\Jmc}_{\ALCOud(\Sigma)}(e) \cup \{ a \}$.
Indeed, towards a contradiction, suppose otherwise.
This means that $\Omc \cup t^{\Jmc}_{\ALCOud(\Sigma)}(e) \models \neg \{ a \}$, and thus, by
compactness of first-order logic on partial interpretations
$\Omc \models F \sqsubseteq \neg \{ a \}$,
for a concept $F \in t^{\Jmc}_{\ALCOud(\Sigma)}(e)$.
The previous step implies that $\lnot F \in \Gamma^{\Omc}_{\Sigma}(a)$, and thus $e \in (\lnot F)^{\Jmc}$, contradicting the fact that $F \in t^{\Jmc}_{\ALCOud(\Sigma)}(e)$.
Therefore, we obtain $(\Imc, d) \equiv^{\ALCOud}_{\Sigma} (\Jmc, e)$.
Moreover, we can assume without loss of generality (cf. Theorem~\ref{thm:omegasatexist}) that both $\Imc$ and $\Jmc$
are $\omega$-saturated.
By Theorem~\ref{thm:alcoiotabisimtoequiv}, Point~(2), it then follows that $(\Imc, d) \sim^{\ALCOud}_{\Sigma} (\Jmc, e)$, where $d = a^{\Imc}$ and $e \neq a^{\Jmc}$.
\end{proof}

\bothdir*

%
\begin{proof}
We adapt the proof of~\cite[Theorem 8]{ArtEtAl21} to cover the cases of partial interpretations and of DLs with definite descriptions.

We start with the upper bound. We show a slighly more general result by proving the following: for $\ALCOud$ ontologies $\Omc_{1}$ and $\Omc_{2}$ and $\ALCOud$ concepts $C_{1}$ and $C_{2}$ and $\Sigma$ a signature it is in \TwoExpTime{} to decide whether $\Omc_{1},C_{1}$ and $\Omc_{2},C_{2}$ are jointly consistent modulo $\ALCOud(\Sigma)$-bisimulations in the sense
that there are pointed models $(\Imc,d)$ and $(\Jmc,r)$ such that $\Imc$ is a model of $\Omc_{1}$ and $d\in C_{1}^{\Imc}$ and $\Jmc$ is a model of $\Omc_{2}$ and $e\in C_{2}^{\Jmc}$ such that $(\Imc,d) \sim_{\Sigma}^{\ALCOud} (\Jmc,e)$. 
Let $\Xi$ denote the closure under single negation of the set of subconcepts of concepts in $\Omc_{1}$, $\Omc_{2}$, $C_{1}$, and $C_{2}$. A \emph{$\Xi$-type $t$} is a maximal subset of $\Xi$
such that there exists a partial interpretation $\Imc$ and $d\in \Delta^{\Imc}$
with $t=\text{tp}_{\Xi}(\Imc,d)$,
where
$$
\text{tp}_{\Xi}(\Imc,d) = \{ C\in \Xi\mid d\in C^{\Imc}\}
$$
is the $\Xi$-type realized at $d$ in $\Imc$. Let 
$T(\Xi)$ denote the set of all $\Xi$-types.
Given a role name $r$, a pair $(t_{1},t_{2})$ of $\Xi$-types $t_{1},t_{2}$ is \emph{$r$-coherent}, in symbols $t_{1} \rightsquigarrow_{r} t_{2}$,
if there exists
a model $\Imc$ and $(d_{1},d_{2})\in r^{\Imc}$ such that
$t_{i} = \text{tp}_{\Xi}(\Imc,d_{i})$ for $i=1,2$.

We encode models using pairs $(T_{1},T_{2})\in 2^{T(\Xi)} \times 2^{T(\Xi)}$
such that there are models $\Imc_{1}$ and $\Imc_{2}$ of $\Omc_{1}$ and $\Omc_{2}$, respectively, such that for every $t\in T_{i}$ there exists a node $d_{t}\in \Delta^{\Imc_{i}}$ that realizes $t$ and such that all $d_{t}$, $t\in T_{1}\cup T_{2}$, are $\ALCOud(\Sigma)$-bisimilar. Thus, we will define ``good'' sets $\mathcal{S}$ of such pairs and the domain of the model $\Imc_{i}$ will consist of tuples $(t,(T_{1},T_{2}))$ with $t\in T_{i}$ and $(T_{1},T_{2})\in \mathcal{S}$. Certain copies of such tuples will be needed to ensure that $\defdes C$ is satisfied just in case that $C$ is satisfied exactly once and also to obtain an $\ALCOud(\Sigma)$-bisimulation.  

To deal with copies some notation is needed. Assume $\mathcal{S}\subseteq 2^{T(\Xi)} \times 2^{T(\Xi)}$ is given.
Call a pair $(t,(T_{1},T_2))$ with $(T_{1},T_{2})\in \mathcal{S}$ and $t\in T_{i}$ an \emph{$\mathcal{S}_{i}$-node}. ``Good'' set $\mathcal{S}$ will correspond to models $\Imc_{1}$ and $\Imc_{2}$ such that $\Imc_{i}$'s domain consists of $\mathcal{S}_{i}$-nodes and some additional copies if needed. To make this precise, we say that an $\mathcal{S}_{i}$-node $(t,(T_{1},T_{2}))$ with $t\in T_{i}$ is \emph{a direct $i$-singleton} if there exists a concept $C$ of the form $\{a\}$ or $\{\defdes D\}$ in $t$. $(t,(T_{1},T_{2}))$ is \emph{an indirect $i$-singleton} if there exists $t'$ such that $T_{i}=\{t\}$, $T_{j}=\{t'\}$, where $\{i,j\}=\{1,2\}$, and $(t',(T_{1},T_{2}))$ is a direct $j$-singleton. $(t,(T_{1},T_{2}))$ is an \emph{$i$-singleton} if it is an indirect or direct $i$-singleton. Any non-$i$-singleton admits copies, but $i$-singletons do not admit copies.

Consider $\mathcal{S} \subseteq 2^{T(\Xi)} \times 2^{T(\Xi)}$. We denote by $p_{i}(\mathcal{S})$ the set of all types in $\bigcup_{(T_{1},T_{2})\in \mathcal{S}}T_{i}$. We write $(T_{1},T_{2})\rightsquigarrow_{r} (T_{1}',T_{2}')$
if for every $t\in T_{i}$ there exists $t'\in T_{i}'$ such that $t \rightsquigarrow_{r} t'$.

We now introduce relevant properties of $\mathcal{S}$ that ensure that the models discussed above exist.

We begin with properties related to the semantics of the universal role.
$\mathcal{S}$ is called \emph{good for the universal role} if any concept of the form $\exists u.D$ or $\forall u.D$ is either in all types in $p_{i}(\mathcal{S})$ or in no type in $p_{i}(\mathcal{S})$. Moreover, if $\exists u.D\in t$ for all $t\in p_{i}(\mathcal{S})$, then there exists $t$ with $D\in t\in p_{i}(\mathcal{S})$. 

We next come to nominals.
$\mathcal{S}$ is called \emph{good for nominals} if
for every $\{a\}$ in $\Xi$ and $i=1,2$, there exists at most one type $t_{a}^{i}$ with $\{a\}\in t_{a}^{i}\in p_{i}(\mathcal{S})$ and at most one pair $(T_{1},T_{2})\in \mathcal{S}$ with $t_{a}^{i}\in T_{i}$. Moreover,  
\begin{itemize}
	\item if $a\in \Sigma$, then that pair takes the form $(\{t_{a}^{1}\},\{t_{a}^{2}\})$;
	\item if $a\not\in \Sigma$, then either $T_{i}$ contains at least two types, or $T_{1}$ and $T_{2}$ are both singletons. 
\end{itemize}
We say that $\mathcal{S}$ is \emph{good for definite descriptions} if
for every $\{\defdes C\}$ in $\Xi$ and $i=1,2$ the following case distinction holds:
\begin{enumerate}
\item there exists exactly one type $t_{\defdes C}^{i}$ with $\{\defdes C\}\in t_{\defdes C}^{i}\in p_{i}(\mathcal{S})$ and exactly one pair $(T_{1},T_{2})\in \mathcal{S}$ with $t_{\defdes C}^{i}\in T_{i}$. $t_{\defdes C}^{i}$ is then also the only type in  $p_{i}(\mathcal{S})$ containing $C$. Moreover,
\begin{itemize}
	\item if $\text{sig}(C) \subseteq \Sigma$, then the pair $(T_{1},T_{2})\in \mathcal{S}$ with $t_{\defdes C}^{i}\in T_{i}$ takes the form $(\{t_{\defdes C}^{1}\},\{t_{\defdes C}^{2}\})$;
	\item if $\text{sig}(C) \not\subseteq \Sigma$, then either $T_{i}$ contains at least two types, or $T_{1}$ and $T_{2}$ are both singletons.  
\end{itemize}
\item there is no type in $p_{i}(\mathcal{S})$ containing $\{\defdes C\}$. Then there is either no type in $p_{i}(\mathcal{S})$ containing $C$ or there are at least two types in $p_{i}(\mathcal{S})$ containing $C$ or the single type $t$ in $p_{i}(\mathcal{S})$ containing $C$ has the following property: no $\mathcal{S}_{i}$-node of the form $(t,(T_{1},T_{2}))$ is an $i$-singleton, or there are at least two such $\mathcal{S}_{i}$-nodes.
\end{enumerate} 

We call $\mathcal{S} \subseteq 2^{T(\Xi)} \times 2^{T(\Xi)}$ \emph{good for $\Omc_{1},\Omc_{2},\Sigma$} if $(T_{1},T_{2})\in \mathcal{S}$ implies $T_{i}\not=\emptyset$ for $i=1,2$, all $t\in T_{i}$ are satisfiable in models of $\Omc_{i}$, for $i=1,2$, $\mathcal{S}$ is good for the universal role, nominals and definite descriptions,
and the following conditions hold:
\begin{enumerate}
\item \emph{$\Sigma$-concept name coherence}: for any concept name
    $A\in \Sigma$ and $(T_{1},T_{2})\in \mathcal{S}$, $A\in t$ iff
    $A\in t'$ for all $t,t'\in T_{1}\cup T_{2}$;
\item \emph{Existential saturation}: for $i=1,2$, if
$(T_{1},T_{2})\in \mathcal{S}$ and $\exists r.C\in t\in
T_{i}$, then there
exists $(T_{1}',T_{2}')\in \mathcal{S}$ such that
there exists $t'\in T_{i}'$ with $C\in t'$ and $t\rightsquigarrow_{r}t'$.

\item \emph{$\Sigma$-existential saturation}: for $i=1,2$, if
    $(T_{1},T_{2})\in \mathcal{S}$ and $\exists r.C\in t\in
    T_{i}$, where $r$ is a role name in $\Sigma$, then there
    exist $(T_{1}',T_{2}')\in \mathcal{S}$ such that
    $(T_{1},T_{2})\rightsquigarrow_{r} (T_{1}',T_{2}')$ and
    there exists $t'\in T_{i}'$ with $C\in t'$ such that
    $t\rightsquigarrow_{r}t'$.
\end{enumerate}
We now show the following claims.
\begin{claim}
\label{cla:jointbisimequiv}
The following conditions are equivalent:
\begin{enumerate}[label=$(\roman*)$, align=left, leftmargin=*]
	\item $\Omc_{1}, C_{1}$ and $\Omc_{2}, C_{2}$ are jointly consistent modulo $\ALCOud(\Sigma)$ bisimulations;
	\item there exists a set $\mathcal{S}$ that is
good for $\Omc_{1},\Omc_{2}$, $\Sigma$ such that $C_{1}\in t_{1}$ and $C_{2}\in T_{2}$ for some $t_1 \in T_{1}$ and $t_{2}\in T_{2}$ with $(T_{1},T_{2})\in \mathcal{S}$.
\end{enumerate}
\end{claim}
\begin{proof}
$(\Rightarrow)$ Let $\Imc_{1},d_{1}\sim^{\ALCOud}_{\Sigma} \Imc_{2},d_{2}$ for models $\Imc_{1}$ of $\Omc_{1}$ and $\Imc_{2}$ of $\Omc_{2}$ such that $d_{1},d_{2}$ realize, respectively, $\Xi$-types $t_{1},t_{2}$ and $C_{1}\in t_{1}, C_{2}\in t_{2}$.
	Define $\mathcal{S}$ by setting $(T_{1},T_{2})\in \mathcal{S}$ if there is $d\in \Delta^{\Imc_{i}}$ for some $i\in \{1,2\}$ such that 
	\[
	T_{j}= \{ \text{tp}_{\Xi}(\Imc_{j},e) \mid e\in \Delta^{\Imc_{j}},
	\Imc_{i},d \sim^{\ALCOud}_{\Sigma} \Imc_{j},e\},
	\]
	for $j=1,2$. We then say that $(T_{1},T_{2})$ is induced by $d$ in
	$\Imc_{i}$.
	It can be shown that $\Smc$ is good for $\Omc_{1}, \Omc_{2}, \Sigma$.	
	
$(\Leftarrow)$
Assume that $\mathcal{S}$ is good for $\Omc_{1}, \Omc_{2}, \Sigma$ and we have $\Xi$-types $s_{1},s_{2}$ with $C_{1}\in s_{1}$ and $C_{2}\in s_{2}$ such that $s_{1}\in S_{1}$ and $s_{2}\in S_{2}$, for some $(S_{1},S_{2})\in \mathcal{S}$.
	We construct interpretations $\Imc_{1}$ and $\Imc_{2}$ as follows. Take for any $\mathcal{S}_{i}$-node $(t,p)$ that is not an $i$-singleton two distinct copies $(t,p)_{1}$ and $(t,p_{2})$ of $(t,p)$. If $(t,p)$ is an $i$-singleton then we set $(t,p)_{1}=(t,p)_{2}=(t,p)$. Now let, for $i=1,2$: 
\begin{align*}
	\Delta^{\Imc_{i}} &= \{ (t,p)_{1},(t,p)_{2} \mid \text{$(t,p)$ is an $\mathcal{S}_i$-node}\}; \\ 
	r^{\Imc_{i}} &=\{((t,p)_{l},(t',p')_{k})\in \Delta^{\Imc_{i}}\times \Delta^{\Imc_{i}} \mid p\rightsquigarrow_{r}p', t\rightsquigarrow_{r}t'\}, \\
	 &  \text{ for} \ r\in \Sigma; \\
	r^{\Imc_{i}} & = \{((t,p)_{l},(t',p')_{k})\in \Delta^{\Imc_{i}}\times \Delta^{\Imc_{i}}\mid t\rightsquigarrow_{r}t'\} \\
	 &  \text{for} \ r\not\in \Sigma; \\
	A^{\Imc_{i}} & = \{(t,p)_{l}\in \Delta^{\Imc_{i}}\mid A\in t\};  \\
a^{\Imc_{i}} & = (t,(T_{1},T_{2}))_{1}\in \Delta^{\Imc_{i}}, \text{for} \ a\in t\in T_{i}. 
\end{align*}
One can show by induction that, for $i=1,2$ and all $D\in \Xi$
and $(t,p)_{l}\in \Delta^{\Imc_{i}}$: $D\in t$ iff $(t,p)_{l}\in D^{\Imc_{i}}$.
Thus, $\Imc_{i}$ is a model of $\Omc_{i}$ for $i=1,2$. Moreover, let
	$$
	S= \{((t_{1},p_{1})_{l},(t_{2},p_{2})_{k})\in \Delta^{\Imc_{1}}\times \Delta^{\Imc_{2}} \mid p_{1}=p_{2}\}.
	$$
	We have that $S$ is an $\ALCOud(\Sigma)$ bisimulation between $\Imc_{1}$ and $\Imc_{2}$ witnessing that $\Omc_{1}, C_{1}$ and $\Omc_{2},C_{2}$ are jointly consistent modulo $\ALCOud(\Sigma)$ bisimulations.
	\end{proof}

\begin{claim}
It is decidable in double  exponential time whether there exists a set $\mathcal{S}$ that is
good for $\Omc_{1},\Omc_{2}, \Sigma$ such that $C_{i}\in t_{i}$ for some 
$t_i \in T_{i}$ with $(T_{1},T_{2})\in \mathcal{S}$.
\end{claim}
\begin{proof}
Given $\Omc_{1}$, $\Omc_{2}$, $C_{1}$, $C_{2}$, and $\Sigma$, we can enumerate in double exponential time the maximal sets
$\mathcal{U} \subseteq 2^{T(\Xi)\times T(\Xi)}$
that are good for $\Omc_{1}, \Omc_{2}, \Sigma$, by proceeding as follows.
We first list all maximal sets $\mathcal{U} \subseteq 2^{T(\Xi)\times T(\Xi)}$
such that 
\begin{itemize}
	\item $(T_{1},T_{2})\in \mathcal{U}$ implies $T_{i}\not=\emptyset$ for $i=1,2$, \item $(T_{1},T_{2})\in \mathcal{U}$ implies that all $t\in T_{i}$ are satisfiable in models of $\Omc_{i}$, for $i=1,2$, 
	\item $\mathcal{U}$ is good for the universal role, nominals and definite descriptions.
\end{itemize}
This can be done in at most double exponential time. 
Next, we recursively eliminate from any such $\mathcal{U}$ all pairs that are not $\Sigma$-concept name coherent, existentially saturated, or $\Sigma$-existentially saturated. Consider the largest fixpoint $\mathcal{S}_{0}\subseteq \mathcal{U}$ of this procedure (which runs in at most double exponential time) and check whether it
is good for $\Omc_{1}, \Omc_{2}, \Sigma$. The good $\mathcal{S}_{0}$ provide a complete list of all maximal $\mathcal{U} \subseteq 2^{T(\Xi)\times T(\Xi)}$ that are good for $\Omc_{1}, \Omc_{2}, \Sigma$.
It remains to check whether any of these $\mathcal{S}_{0}$ contains $(T_{1},T_{2})$ with $C_{1}\in t_{1}\in T_{1}$ and $C_{2}\in t_{2}\in T_{2}$ for some $t_{1},t_{2}$.
\end{proof}
For the lower bound, it is shown in~\cite{ArtEtAl21} that it is \TwoExpTime-hard
to decide for $\mathcal{ALCO}$ ontologies $\Omc$ with a single individual name $a$ and signatures $\Sigma\subseteq \Sigma_{\Omc}\setminus\{a\}$ whether $\Omc,\{a\}$ and $\Omc,\neg\{a\}$ are jointly consistent modulo $\mathcal{ALCO}_{u}(\Sigma)$ bisimulations. Then it is easy to see that 
for any such $\Omc$ and $\Sigma$, $\Omc,\{a\}$ and $\Omc,\neg\{a\}$ are jointly consistent modulo $\mathcal{ALCO}_{u}(\Sigma)$ bisimulations iff $\Omc,\{a\}$ and $\Omc,\neg\{a\}$ are jointly consistent modulo $\ALCOud(\Sigma)$ bisimulations which implies the desired result.
\end{proof}

%


\elourefexprchar*
\begin{proof}
%
For (1)~$\Rightarrow$~(2), assume that Condition~(1) holds.
Let 
\[
\Gamma^{\Omc}_{\Sigma}(a) = \{ C \mid \Omc \models \{ a \} \sqsubseteq C, C \ \ELOu(\Sigma) \ \text{concept} \}.
\]
As there does not exists an $\ELOu(\Sigma)$
RE
for $a$ under $\Omc$, for every $C \in \Gamma^{\Omc}_{\Sigma}(a)$ we have that
$\Omc \not \models C \sqsubseteq \{ a \}$.
Since $\Gamma^{\Omc}_{\Sigma}(a)$ is closed under conjunctions, by compactness of
FO
on partial interpretations~(of which $\ELOu$ is a fragment, cf. Theorem~\ref{thm:omegasatexist}) we obtain that
$\Omc \cup \Gamma^{\Omc}_{\Sigma}(a) \not \models \{ a \}$, i.e.,
there exist a partial interpretation $\Jmc$ and an $e \in \Delta^{\Jmc}$ such that $\Jmc \models \Omc$ and $e \in C^{\Jmc}$, for every $C \in \Gamma^{\Omc}_{\Sigma}(a)$, but $e \neq a^{\Jmc}$.
W.l.o.g.,
we can assume that $\Jmc$ is $\omega$-saturated.
Moreover, given that ${\mathsf A}$ is satisfiable w.r.t. $\Omc$, we have that $\fincanmod$ is defined, and since $\Omc \models {\mathsf A} \equiv \{ a \}$ it holds that
$a^{\fincanmod} = [{\mathsf A}] \in \Delta^{\fincanmod}$.
Now suppose, towards a contradiction, that 
$\tp^{\fincanmod}_{\ELOu(\Sigma)}(a^{\fincanmod}) \not \subseteq  \tp^{\Jmc}_{\ELOu(\Sigma)}(e)$, meaning that there exists an $\ELOu(\Sigma)$ concept $F$ such that $a^{\fincanmod} \in F^{\fincanmod}$, but $e \not \in F^{\Jmc}$.
Hence, by
Theorem~\ref{theo:eloufincanmod},
we obtain that $\Omc \models \{ a \} \sqsubseteq F$, and thus $F \in  \Gamma^{\Omc}_{\Sigma}(a)$, while $e \not \in F^{\Jmc}$, contradicting the fact that $e \in C^{\Jmc}$, for every $C \in \Gamma^{\Omc}_{\Sigma}(a)$.
Therefore, $\tp^{\fincanmod}_{\ELOu(\Sigma)}(a^{\fincanmod}) \subseteq  \tp^{\Jmc}_{\ELOu(\Sigma)}(e)$.
By Theorem~\ref{thm:simequiv}, Point~(2), we have that $(\fincanmod, a^{\fincanmod}) \simul^{\ELOu}_{\Sigma} (\Jmc, e)$, with $e \neq a^{\Jmc}$, as required.

For (2)~$\Rightarrow$~(1), assume that Condition~(2) holds, i.e., there exist a model $\Jmc$ of $\Omc$ and $e \in \Delta^{\Jmc}$ such that $(\fincanmod, a^{\fincanmod}) \simul^{\ELOu}_{\Sigma} (\Jmc, e)$ and $e \neq a^{\Jmc}$ ($\fincanmod$ is defined, since ${\mathsf A}$ is satisfiable w.r.t. $\Omc$).
Now suppose, towards a contradiction, that there exists an $\ELOu(\Sigma)$ concept $C$ such that $\Omc \models \{ a \} \equiv C$.
Since
$\Omc \models {\mathsf A} \equiv \{ a \}$, the previous step is equivalent to
$\Omc \models {\mathsf A } \equiv C$.
Thus,
by Theorem~\ref{theo:eloufincanmod},
we have in particular $[ {\mathsf A } ] \in C^{\fincanmod}$.
Given that
$\Omc \models {\mathsf A} \equiv \{ a \}$,
we have $a^{\fincanmod} = [ {\mathsf A }]$, and thus the previous step implies
$a^{\fincanmod} \in C^{\fincanmod}$.
Having assumed $(\fincanmod, a^{\fincanmod}) \simul^{\ELOu}_{\Sigma} (\Jmc, e)$,
from Theorem~\ref{thm:simequiv}, Point~(1), we obtain $e \in C^{\Jmc}$, with $e \neq a^{\Jmc}$,
contradicting the fact that, since $\Jmc \models \Omc$, $\{ a \}^{\Jmc} = C^{\Jmc}$.
\end{proof}

\begin{restatable}{lemma}{elourefexpden}
\label{thm:elourefexpden}
Condition~(2) of Lemma~\ref{thm:elourefexprchar} can be checked in polynomial time.
\end{restatable}
%
%
\begin{proof}
We first require the following notion.
Given the $\ELOud$ finite canonical model $\fincanmod$ of
$\Omc = \Omc_{0} \cup \{ {\mathsf A}\equiv \{ a \} \}$ and
${\mathsf A}$,
consider a fresh concept name $X_{d}$, for each $d \in \Delta^{\fincanmod}$.
Let $\Sigma \subseteq \sig{\Omc}$ be a signature. We define the \emph{$\ELOu(\Sigma)$ diagram $\Dmc_{\fincanmod}$} of $\fincanmod$ as the ontology consisting of the following CIs:
\begin{itemize}
	\item $X_{d} \sqsubseteq A$, for every $A\in \Sigma$ and $d \in A^{\fincanmod}$;
	 \item $X_{b^{\fincanmod}}  \sqsubseteq \{b\}$, for every $b\in \Sigma$;
	\item $X_{d} \sqsubseteq \exists r.X_{d'}$, for every $r\in \Sigma$ and $(d,d')\in r^{\fincanmod}$;
	\item $X_{d} \sqsubseteq \exists u.X_{d'}$, for every $d, d' \in \Delta^{\fincanmod}$.
\end{itemize}

Then, we require the following claim.
\begin{claim}
\label{cla:fincanmodx}
For every $d \in \Delta^{\fincanmod}$ and every $\ELOu(\Sigma)$ concept $C$, it holds that
\[
d \in C^{\fincanmod}
\ \text{implies} \
\Omc \cup \Dmc_{\fincanmod} \models X_{d} \sqsubseteq C.
\]
\end{claim}
\begin{proof}
The proof is by induction on the construction of $C$.
The base cases are as follows.
\begin{itemize}
\item For $C = \bot$, the statement is vacously true.
\item For $C = \top$, the statement is obviously true.
\item For $C = A$, the statement follows immediately from the
  definition of $\Dmc_{\fincanmod}$.
\item For $C = \{ b \}$, by construction of $\fincanmod$ we have that
  $d \in \{ b \}^{\fincanmod}$ iff $d = b^{\fincanmod}$. Since by
  definition
  $X_{b^{\fincanmod}} \sqsubseteq \{ b \} \in \Dmc_{\fincanmod}$, we
  have
  $\Omc \cup \Dmc_{\fincanmod} \models X_{b^{\fincanmod}} \sqsubseteq
  \{ b \}$.
\end{itemize}
The inductive cases are as follows.
\begin{itemize}
\item For $C = D \sqcap E$, suppose that
  $d \in (D \sqcap E)^{\fincanmod}$, meaning that
  $d \in D^{\fincanmod}$ and $d \in E^{\fincanmod}$. By i.h.,
  this
  implies $\Omc \cup \Dmc_{\fincanmod} \models X_{d} \sqsubseteq D$
  and $\Omc \cup \Dmc_{\fincanmod} \models X_{d} \sqsubseteq E$.
  Thus,
  $\Omc \cup \Dmc_{\fincanmod} \models X_{d} \sqsubseteq D \sqcap E$.
\item For $C = \exists r.D$, suppose that
  $d \in (\exists r.D)^{\fincanmod}$. This means that there exists
  $d' \in \Delta^{\fincanmod}$ such that $(d,d') \in r^{\fincanmod}$
  and $d' \in D^{\fincanmod}$.  By i.h., 
  $\Omc \cup \Dmc_{\fincanmod} \models X_{d'} \sqsubseteq D$.  By
  construction of $\Dmc_{\fincanmod}$, we have that
  $\Omc \cup \Dmc_{\fincanmod} \models X_{d} \sqsubseteq \exists
  r.X_{d'}$, since $(d,d') \in r^{\fincanmod}$.  Thus,
  $\Omc \cup \Dmc_{\fincanmod} \models X_{d} \sqsubseteq \exists r.D$.
\item For $C = \exists u.D$, similarly to the previous case, suppose
  that $d \in (\exists u.D)^{\fincanmod}$. Then, there exists
  $d' \in \Delta^{\fincanmod}$ such that $d' \in D^{\fincanmod}$.  By
  i.h., 
  $\Omc \cup \Dmc_{\fincanmod} \models X_{d'} \sqsubseteq D$.  By
  construction of $\Dmc_{\fincanmod}$, we have
  $\Omc \cup \Dmc_{\fincanmod} \models X_{d} \sqsubseteq \exists
  u.X_{d'}$.  Thus,
  $\Omc \cup \Dmc_{\fincanmod} \models X_{d} \sqsubseteq \exists u.D$.
  \qedhere
\end{itemize}
\end{proof}


We now show that Condition~(2) of Lemma~\ref{thm:elourefexprchar} holds iff 
\begin{equation*}
	\Omc \cup \Dmc_{\fincanmod} \not \models X_{a^{\fincanmod}} \sqsubseteq \{ a \}.
\end{equation*}

For the $(\Rightarrow)$ direction, assume that Condition~(2) of Lemma~\ref{thm:elourefexprchar} holds, meaning that
there exist a model $\Jmc$ of $\Omc$ and $e \in \Delta^{\Jmc}$ such that $(\fincanmod, a^{\fincanmod}) \simul^{\ELOu}_{\Sigma} (\Jmc, e)$ and $e \neq a^{\Jmc}$.
Let $Z \subseteq \Delta^{\fincanmod} \times \Delta^{\Jmc}$ be such an $\ELOu(\Sigma)$ simulation from $\fincanmod$ to $\Jmc$.
We define the partial interpretation $\Jmc'$ as $\Jmc$, except that, for every fresh concept name $X_{d}$ occurring in $\Dmc_{\fincanmod}$, we set $v \in X^{\Jmc'}_{d}$ iff $(d,v) \in Z$.
We have that
$\Jmc'$ is still a model of $\Omc$ with some $e \in \Delta^{\Jmc'}$ such that $(\fincanmod, a^{\fincanmod}) \simul^{\ELOu}_{\Sigma} (\Jmc', e)$ and $e \neq a^{\Jmc'}$.
Moreover, by definition of $\Jmc'$, we have $e \in X^{\Jmc'}_{a^{\fincanmod}}$, and it can be checked that $\Jmc' \models \Dmc_{\fincanmod}$.
Hence, $\Jmc'$ is as required, and $\Omc \cup \Dmc_{\fincanmod} \not \models X_{a^{\fincanmod}} \sqsubseteq \{ a \}$.

%
%
%
%
%
%

For the $(\Leftarrow)$ direction,
assume that
$\Omc \cup \Dmc_{\fincanmod} \not \models X_{a^{\fincanmod}}
\sqsubseteq \{ a \}$.  This means that there exist a model $\Jmc$ of
$\Omc \cup \Dmc_{\fincanmod}$ and an $e \in \Delta^{\Jmc}$ such that
$e \in X^{\Jmc}_{a^{\fincanmod}}$, but $e \neq a^{\Jmc}$.  We can
assume w.l.o.g. that $\Jmc$ is $\omega$-saturated.
We now prove that $(\fincanmod, a^{\fincanmod})
\simul^{\ELOu}_{\Sigma} (\Jmc, e)$.
Suppose that
$a^{\fincanmod} \in C^{\fincanmod}$.  By Claim~\ref{cla:fincanmodx},
this implies
$\Omc \cup \Dmc_{\fincanmod} \models X_{a^{\fincanmod}} \sqsubseteq
C$.  Since $\Jmc \models \Omc \cup \Dmc_{\fincanmod}$ and
$e \in X^{\Jmc}_{a^{\fincanmod}}$,
we obtain that $e \in C^{\Jmc}$.
Hence,
$\tp^{\fincanmod}_{\ELOu(\Sigma)}(a^{\fincanmod}) \subseteq \tp^{\Jmc}_{\ELOu(\Sigma)}(e)$.
By Theorem~\ref{thm:elousimequiv}, Point~(2), the relation
$Z = \{ (u, v) \in \Delta^{\fincanmod} \times \Delta^{\Jmc} \mid  \tp^{\fincanmod}_{\ELOu(\Sigma)}(u) \subseteq \tp^{\Jmc}_{\ELOu(\Sigma)}(v) \}$
is thus an $\ELOu(\Sigma)$ simulation from $\fincanmod$ to $\Jmc$ such that $(a^{\fincanmod}, e) \in Z$,
with
$e \neq a^{\Jmc}$,
as required.
%

Since we proved in Theorem~\ref{th-eloud-sat} that entailment in
$\ELOud$ can be decided in polynomial time,
we have finished the proof.
%
\end{proof}

We formulate Point~(4) of Theorem~\ref{thm:mainrefexp} in the
following theorem.
\begin{theorem}
$(\ALCOud,\ELOud)$
RE
existence is undecidable.
\end{theorem}
\begin{proof}
	In this proof we identify any $\EL$-concept $C$ with a tree-shaped conjunctive query $q_{C}(x)$ with a single answer variable $x$ and we use graph-theoretic terminology when speaking about $q_{C}(x)$. In~\cite{BotEtAl19}, $\ALC$ ontologies $\Omc_{1}$, sets $\Amc$ of atomic assertions, and signatures $\Sigma_{0}$ and $\Sigma$ are constructed such that 
\begin{enumerate}
	\item $\Omc_{1}$ uses only a single role name $r$ and only contains CIs of the from $A \sqsubseteq C$, where $C$ is constructed using existential restrictions, conjunction, and disjunction;
	\item $\mathcal{A}$ takes the form $\{r(a,a), F_{1}(a),\ldots,F_{n}(a)\}$ for a single individual name $a$ and some concept names $F_{1},\ldots,F_{n}$;
	\item $\Sigma_{0}$ contains concept names only and $\Sigma\supseteq \Sigma_{0}$ additionally contains $r$ and a concept name ${\it end}$ not in $\Sigma_{0}$;
\end{enumerate}
and such that it is undecidable whether there exists an 
$\EL(\Sigma)$-concept $D$ of the form 
$$
\exists r.(B_{1} \sqcap \exists r.(B_{2} \sqcap \cdots \exists r.(B_{n} \sqcap \exists r.{\it end}))\cdots))
$$
with $B_{1},\ldots,B_{n}\in \Sigma_{0}$ and $\Omc,\Amc\models D(a)$. We say
that such $\EL(\Sigma)$-concepts are \emph{without $\Sigma_{0}$-gaps}. Observe that the query $q_{D}(x)$ takes the form
\begin{eqnarray*}
	q_{D}(x) & = & \exists x_{1}\cdots \exists x_{n} \exists x_{n+1} (r(x,x_{1}) \wedge B_{1}(x_{1}) \wedge \cdots \wedge \\
	         &  &  
B_{n}(x_{n}) \wedge r(x_{n},x_{n+1}) \wedge {\it end}(x_{n+1}))
\end{eqnarray*}
and thus being without $\Sigma_{0}$-gaps means that there exists an $r$-path from
the answer variable $x$ to an {\it end}-node in $q_{D}(x)$ such that all nodes on the path are decorated with at least one concept name in $\Sigma_{0}$. It is also shown that it is undecidable whether any CQ containing $q_{D}(x)$ as a subquery is entailed.
 
Now, in~\cite{BotEtAl19}, also an $\EL$-ontology $\Omc_{2}$ is constructed such that 
\begin{enumerate}
	\item $\Omc_{2}$ also only uses the single role name $r$ and additional concept 
	names;
	\item $\Omc_{2},\Amc\not\models D(a)$ for any $\EL(\Sigma)$-concept $D$ without $\Sigma_{0}$-gaps;
	\item If $D$ is an $\EL(\Sigma)$-concept such that $q_{D}(x)$ does not contain any subquery without $\Sigma_{0}$-gaps, then $\Omc_{1},\Amc\models D(a)$ implies $\Omc_{2},\Amc \models D(a)$. This even holds for CQs over $\Sigma$.
	\end{enumerate}
We use $\Omc_{1},\Omc_{2}$, $\Amc$, $\Sigma$, and $\Sigma_{0}$ to construct
an $\ALCO$-ontology $\Omc$ with individual names $a,b$ and $\Sigma'=\Sigma\cup \{D_{a,b}\}$ such that it is undecidable whether there exists an $\ELOud(\Sigma')$-concept $C$ such that $\Omc\models \{a\} \equiv C$. We achieve this by ensuring that there exists such a $\ELOud(\Sigma')$-concept $C$ iff there exists an $\EL(\Sigma)$-concept $D$ without $\Sigma_{0}$-gaps such that $\Omc_{1},\Amc\models D(a)$.

To construct $\Omc$ we first ensure that $\Omc_{1}$ and $\Omc_{2}$ do not interfere with each other and thus relativize $\Omc_{i}$ using a fresh concept name $A_{i}$, for $i=1,2$. The respective relativizations together with $\{a\} \sqsubseteq A_{1}$ and $\{b\} \sqsubseteq A_{2}$ are denoted $\Omc_{i}^{A_{i}}$, $i=1,2$.
Next we encode the assertions of $\Amc$ in the ontology $\Omc$, at both $a$ and  $b$: 
\begin{itemize}
	\item $\{a\}\sqsubseteq \exists r.\{a\}$ and $\{a\} \sqsubseteq F_{i}$ for $F_{i}(a) \in \Amc$.
	\item $\{b\}\sqsubseteq \exists r.\{b\}$ and $\{b\} \sqsubseteq F_{i}$ for $F_{i}(a) \in \Amc$.
\end{itemize}
Now we make sure that \emph{no} $\ELOud(\Sigma)$-concept without $\Sigma_{0}$-gaps can be satisfied at $b$. Thus we add for a fresh concept name $P$:
$$
\{b\} \sqsubseteq P, \quad P \sqcap \forall r.((\bigsqcap_{B \in \Sigma_{0}}B) \rightarrow P),
\quad P \sqcap \exists r.{\it end} \sqsubseteq \bot
$$
Finally we add $D_{a,b}\equiv \{a\} \sqcup \{b\}$. The ontology $\Omc$ is now defined as the union of $\Omc_{1}^{A_{1}}$ and $\Omc_{2}^{A_{2}}$ and the inclusions introduced above.

\medskip
\noindent
Claim. $\Omc_{1},\Amc\models D(a)$ for some $\EL(\Sigma)$-concept $D$ without $\Sigma_{0}$-gaps iff there exists a $\ELOud(\Sigma')$-concept $D$ such that $\Omc\models \{a\}\equiv D$.

\medskip
\noindent

Assume $\Omc,\Amc\models D(a)$ for some $\EL(\Sigma)$-concept $D$ without $\Sigma_{0}$-gaps. Let $C = D \sqcap D_{a,b}$. We show that
$\Omc\models \{a\}\equiv C$. $\Omc\models \{a\}\sqsubseteq D \sqcap D_{a,b}$
follows directly from the construction. To show that $\Omc\models D \sqcap D_{a,b}
\sqsubseteq \{a\}$ observe that it suffices to prove that $\Omc \models \{b\}\sqsubseteq \neg D$. But this follows by the inclusion of the CIs with $P$
in $\Omc$.

Conversely, assume that $\Omc_{1},\Amc\not\models D(a)$ for any $\EL(\Sigma)$-concept $D$ without $\Sigma_{0}$-gaps. Assume for a proof by contradiction that
there exists an $\ELOud(\Sigma')$-concept $D$ such that $\Omc\models \{a\}\equiv D$. As definite descriptions and the universal role clearly cannot help to define $\{a\}$, we may assume that $D$ does not use them. Consider $q_{D}(x)$ and obtain the query $q_{D}'(x)$ by identifying all variables on the path from 
$x$ to any variable $y$ with $D_{a,b}(y)$ in $q_{D}(x)$ (including $y$). Take the $\EL(\Sigma)$-concept $D'$ corresponding to $q_{D}'(x)$ (with $D_{a,b}$ and $r(x,x)$ removed if they occur in $q_{D}'(x)$. Also let $\Gamma$ be the
set of concept names $B$ with $B(x)$ a member of $q_{D}'(x)$. Then, by assumption, 
\begin{itemize}
\item $\Omc_{1},\Amc\models D'(a)$ and 
\item $\Omc_{1},\Amc \models B(a)$ for all $B\in \Gamma$.
\end{itemize}
It follows from Point~1 that $q_{D'}$ does not contain any subqueries without $\Sigma_{0}$-gaps. Thus, $\Omc_{2},\Amc\models D'(a)$ and so $\Omc\models \{b\} \sqsubseteq D'$. From Point~2 we obtain $\Omc_{2},\Amc\models B(a)$ for all $B\in \Gamma$ and thus $\Omc\models \{b\} \sqsubseteq B$ for all $B\in \Gamma$. From $\{b\} \sqsubseteq \exists r.\{b\}$ and $\{b\} \sqsubseteq D_{a,b}$ we obtain
$\Omc\models \{b\} \sqsubseteq D$ which contradicts the assumption that $\Omc \models \{a\}\equiv D$ since $\{b\} \sqcap \neg \{a\}$ is satisfiable w.r.t.~$\Omc$.
\end{proof}
The reduction shows that the undecidability result also holds for any language between $\EL$ and $\ELOud$.


\subsection*{Proofs for Section~\ref{sec:otherfreedl}}

For the proofs in this section, we introduce the following notation.
%
An \emph{$\ALCOud$ formula} is defined inductively as
\[
\p ::= \alpha 
\mid \neg (\p)  \mid (\p \land \p),
\]
where~$\alpha$ is an $\ALCOud$ axiom. The satisfaction conditions under a partial interpretation for axioms are given as in Section~\ref{sec:freedl}, and those for Booleans connectives are as usual. The notion of an $\ALCOud$ formula being \emph{satisfiable on partial interpretations} is defined accordingly.
Given
a formula $\p$,
$\con{\p}$ is the set of all concepts occurring in $\p$, and $\sub{\p}$
is the set of all subformulas of $\p$, defined as usual.
%
%
By
a
well-known internalisation technique
via
the universal role~\cite{BaaEtAl03a,Rud11} and what observed in Point~(2) of Section~\ref{sec:observations}, it is possible to encode $\ALCOud$ formulas into $\ALCOud$ CIs under partial interpretations.
Thus, in our setting, formulas are just syntactic sugar, introduced to
simplify the reduction below.
%
However,
similarly to $\NKRd$,
parentheses in
negated
formulas
$\lnot (\p)$
are not eliminable,
since they
disambiguate
between 
expressions
of the form $\lnot C(\tau)$, i.e.,
assertions
with a
\emph{negated concept},
and \emph{negated assertions} of the form $\lnot (C (\tau))$.
Indeed,
these
expressions
have different satisfaction conditions on partial interpretations:
while $\lnot C(\tau)$ requires $\tau$ to denote in any of its models,
a formula like $\lnot (C (\tau))$ is satisfied also in partial interpretations where $\tau$ does not denote.
Thus, $\lnot C (\tau)$ is encoded by the $\ALCOud$ CIs
$\top \sqsubseteq \exists u.\{ \tau \}$, $\{ \tau \} \sqsubseteq \lnot C$, while $\lnot (C (\tau) )$ is encoded by the CI $\top \sqsubseteq \forall u. \lnot \{ \tau \} \sqcup \exists u.(\{ \tau \} \sqcap \lnot C)$.

%

Given an $\NKRd$ concept $C$ and a formula $\p$, let $\Ex$ be a fresh concept name, representing the \emph{existence concept}, that does not occur in $C$ and $\p$.
We introduce the translation $\cdot^{\exrel}$,
mapping an $\NKRd$ concept $C$ to an $\ALCOud$ concept $C^{\exrel}$ as follows:
%
\begin{gather*}
	\topex^{\exrel} = \Ex, \
	A^{\exrel} = A, \
	(\lnot C)^{\exrel} = \lnot C^{\exrel}, \\
		(\exists r. C)^{\exrel} = \exists r.(\Ex \sqcap C^{\exrel}), \
	(C \sqcap D)^{\exrel} = C^{\exrel} \sqcap D^{\exrel}, \\
	{\small
	\{ \tau \}^{\exrel}  = 
	\begin{cases}
		\{ a \}, & \!\!\!\!\!\!\!\!\!\!\!\!\!\!\!\!\!\!\!\!\!\!\!\!\!\!  \text{if $\tau = a$}; \\
		\{ \defdes (\Ex \sqcap C^{\exrel}) \} \sqcup (\lnot \exists u.\{ \defdes (\Ex \sqcap C^{\exrel}) \} \sqcap \{  a_{\defdes C} \}), &
	\end{cases}
	}
	\\
	{\small
	\qquad\qquad\qquad\qquad\qquad\qquad\qquad\qquad\qquad\qquad\qquad\ \text{if $\tau = \defdes C$};
	}
\end{gather*}
%

We then define two
functions, $\cdot^{+}$ and $\cdot^{-}$, mapping an $\NKRd$ formula $\p$ to an $\ALCOud$ formula $\p^{+}$, respectively $\p^{-}$. In the following, $\circ \in \{ +, - \}$.
{\small{
\begin{align*}
	C(\tau)^{\circ} & =
		\begin{cases}
			\{ \tau \}^{\exrel} \sqsubseteq C^{\exrel}, & \ \text{if} \  \circ = + ; \\
			\{ \tau \}^{\exrel} \sqsubseteq \Ex \sqcap C^{\exrel}, & \ \text{if} \ \circ = - ;\\
		\end{cases}
		\\
	r(\tau_{1}, \tau_{2})^{\circ} & =
		\begin{cases}
			\{ \tau_{1} \}^{\exrel} \sqsubseteq \exists r.\{\tau_{2}\}^{\exrel}, & \ \text{if} \  \circ = + ; \\
			 \{ \tau_{1} \}^{\exrel} \sqsubseteq \Ex \sqcap \exists r. (\Ex \sqcap \{\tau_{2}\}^{\exrel}), & \ \text{if} \ \circ = - ;\\
		\end{cases}
		\\
	(\tau_{1} = \tau_{2})^{\circ} & =
		\begin{cases}
			\{ \tau_{1} \}^{\exrel} \equiv \{\tau_{2}\}^{\exrel}, & \ \text{if} \  \circ = + ; \\
			  \{ \tau_{1} \}^{\exrel} \sqcap \Ex \equiv \{\tau_{2}\}^{\exrel} \sqcap \Ex, & \ \text{if} \ \circ = - ;\\
		\end{cases}	
\end{align*}
}}
\[
(C \sqsubseteq D)^{\circ} = \Ex \sqcap C^{\exrel} \sqsubseteq D^{\exrel},
\]
\[
		(\lnot (\psi))^{\circ} = \lnot (\psi^{\circ}), \
		(\psi \land \chi)^{\circ} = \psi^{\circ} \land \chi^{\circ}.
\]
Finally, for an $\NKRd$ formula $\p$,
and $\circ \in \{ +, - \}$,
we define the $\ALCOud$ translation $\p^{\boxcircle}$
in the following way:
\begin{align*}
\p^{\boxcircle} =
	\p^{\circ} \land
	& \bigwedge_{\{ a \} \in \con{\p}} \!\!\!
	\top( a )  \
	\land \bigwedge_{\{ \defdes D \} \in \con{\p}} \!\!\! \lnot \Ex( a_{\defdes D}).
\end{align*}

%

To prove 
Theorem~\ref{prop:redddtopart} we show the following lemma.

\begin{restatable}{lemma}{alcoiotadualdompartint}
\label{prop:alcoiotadualdompartint}
Let $\p$ be an
$\NKRd$
formula, and let $\circ \in \{ +, - \}$. There exists a dual domain interpretation $I$ such that $I \models^{\circ} \p$
iff
there exists a partial interpretation $\Imc$ such that $\Imc \models \p^{\boxcircle}$.
\end{restatable}
\begin{proof}
To prove the statement, we will show the following.
\begin{enumerate}
[label=$(\roman*)$, align=left, leftmargin=*, labelsep=0cm]
	\item For every dual-domain interpretation $I$ such that $I \models^{\circ} \p$, there exists a partial interpretation $\Imc$ such that
$
	\Imc \models \p^{\circ} \land \bigwedge_{\{ a \} \in \con{\p}} \top \sqsubseteq \exists u.\{ a \}  \land \bigwedge_{\{ \defdes D \} \in \con{\p}} (\top \sqsubseteq \exists u. \{ a_{\defdes D} \} \land \{ a_{\defdes D} \} \sqsubseteq \lnot \Ex).
$
\item[$(ii)$] For every partial interpretation $\Imc$ such that
  $ \Imc \models \p^{\circ} \land \bigwedge_{\{ a \} \in \con{\p}}
  \top \sqsubseteq \exists u.\{ a \} \land \bigwedge_{\{ \defdes D \}
    \in \con{\p}} (\top \sqsubseteq \exists u. \{ a_{\defdes D} \}
  \land \{ a_{\defdes D} \} \sqsubseteq \lnot \Ex), $ there exists a
  dual-domain interpretation $I$ such that $I \models^{\circ} \p$.
\end{enumerate}

$(i)$ Let $I = (\OutDom^{I}, \InnDom^{I}, \cdot^{I})$ be a dual-domain
interpretation such that $I \models^{\circ} \p$, with
$\circ \in \{ +, - \}$.
We construct a partial interpretation
$\Imc = (\Delta^{\Imc}, \cdot^{\Imc})$ by taking
$\Delta^{\Imc} = \OutDom^{I}$, $A^{\Imc} = A^{I}$, for all
$A \in \NC \setminus \{ \Ex \}$, $\Ex^{\Imc} = \InnDom^{I}$,
$r^{\Imc} = r^{I}$, for all $r \in \NR$. Moreover,
for
every $\{ \defdes D \} \in \con{\p}$, we consider a fresh individual
name $a_{\defdes D}\in\NI$ not occurring in $\p$ and we set
$a_{\defdes D}^{\Imc} = d_{\defdes D}\in \OutDom^{I} \setminus
\InnDom^{I}$, while for the remaining individual names we define
$a^{\Imc} = a^{I}$.
%
%
First, we observe that, since $\cdot^{I}$ is total on $\NI$,
$\Imc \models \top \sqsubseteq \exists u.\{ a \}$, for every
$\{ a \} \in \con{\p}$, meaning that every $a$ that occurs in $\p$
denotes in $\Imc$.
%
In addition, for every $\{ \defdes D \} \in \con{\p}$, we have by
definition of $\Imc$ that
$\Imc \models \top \sqsubseteq \exists u. \{ a_{\defdes D} \} \land \{
a_{\defdes D} \} \sqsubseteq \lnot \Ex$.  We now require the following
claim.
\begin{claim}
  \label{claim:alcoiotaconsat}
  For all $C \in \con{\p}$ and $d \in \Delta^{I}$, $d \in C^{I}$ iff
  $d \in (C^{\exrel})^{\Imc}$.
\end{claim}
\begin{proof}
The proof is by induction on $C$.
The base cases $C = A$
and
$C = \topex$
come immediately from the definition of $\Imc$.
The inductive cases $C = \lnot D$ and $C = D \sqcap E$ are straightforward.

For $C = \exists r.D$, we have the following.
By definition of $(\exists r.D)^{I}$, we have that
$d \in (\exists r.D)^{I}$ iff there exists $e \in \InnDom^{I}$ such
that $(d,e) \in r^{I}$ and $e \in D^{I}$. By i.h. and definition of
$\Imc$, this holds iff there exists $e \in \Ex^{\Imc}$ such that
$(d,e) \in r^{\Imc}$ and $e \in (D^{\exrel})^{\Imc}$, or,
equivalently, iff there exists $e \in \Delta^{\Imc}$ such that
$(d,e) \in r^{\Imc}$ and $e \in (\Ex \sqcap
(D^{\exrel}))^{\Imc}$. This means that
$d \in (\exists r.(\Ex \sqcap D^{\exrel}))^{\Imc}$, which is the same
as $d \in ((\exists r.D)^{\exrel})^{\Imc}$.


For $C = \{ \tau \}$, we consider the two forms of $\tau$.
\begin{itemize}
\item Let $\tau = a$. We have that $d \in \{ a \}^{I}$ iff
  $d = a^{I}$. By definition of $\Imc$, the previous step is
  equivalent to $d = a^{\Imc}$. This means that
  $d \in \{ a \}^{\Imc}$, i.e., by definition of $\cdot^{\exrel}$,
  $d \in (\{ a \}^{\exrel})^{\Imc}$.
\item Let $\tau = \defdes D$.  We have that
  $d \in \{ \defdes D \}^{I}$ iff $d = (\defdes D)^{I}$, meaning that
  either $\InnDom^{I} \cap D^{I} = \{ d \}$, or $\InnDom^{I} \cap D^{I} \neq \{ e \}$,
  for every $e \in \Delta^{I}$,
  and
  $d = d_{\defdes D} \in \OutDom^{I} \setminus \InnDom^{I}$.  By
  i.h. and definition of $\Imc$, the previous step is equivalent to:
  either $(\Ex \sqcap D^{\exrel})^{\Imc} = \{ d \}$, or
  $(\Ex \sqcap D^{\exrel})^{\Imc} \neq \{ e \}$, for every $e \in \Delta^{\Imc}$,
  and
  $d = a_{\defdes D}^{\Imc} \in \Delta^{\Imc} \setminus \Ex^{\Imc}$.
  This means the following: either $d = \defdes (\Ex \sqcap D^{\exrel})^{\Imc}$,
  or $\defdes (\Ex \sqcap D^{\exrel})$ does not denote in $\Imc$ and
  $d = a_{\defdes D}^{\Imc}$.  Equivalently,
  $d \in (\{ \defdes (\Ex \sqcap  D^{\exrel}) \} \sqcup (\lnot \exists u.\{ \defdes
  (\Ex \sqcap D^{\exrel}) \} \sqcap \{ a_{\defdes D} \}))^{\Imc}$.  By definition
  of $\cdot^{\exrel}$, this means that
  $d \in (\{ \defdes D \}^{\exrel})^{\Imc}$.
%
%
\qedhere
\end{itemize} 
\end{proof}

The next claim follows from
the proof of Claim~\ref{claim:alcoiotaconsat}.
%
We use the following notation: let
$\tau^{\exrel} = a$, if $\tau = a$, and
$\tau^{\exrel} = \defdes (\Ex \sqcap D^{\exrel})$, if $\tau = \defdes D$.

\begin{claim}
\label{cla:alcoiotaterm}
The following equivalences hold.
\begin{enumerate}
[label=$(\roman*)$, align=left, leftmargin=*, labelsep=0cm]
\item For every $\{ \defdes D \} \in \con{\p}$ and
  $d \in \OutDom^{I}$,
  $d = (\defdes D)^{I}$ iff either $d = \defdes (\Ex \sqcap D^{\exrel})^{\Imc}$, or
  $\defdes (\Ex \sqcap D^{\exrel})$ does not denote in $\Imc$
and $d = a^{\Imc}_{\defdes D}$.
\item For every $\{ \tau \} \in \con{\p}$ and $d \in \OutDom^{I}$,
  $d = \tau^{I} \in \InnDom^{I}$ iff
  $d = (\tau^{\exrel})^{\Imc} \in \Ex^{\Imc}$.
\end{enumerate}
\end{claim}

Using the two claims above, we can show the following.
\begin{claim}
  \label{claim:alcoiotaforsat}
  For all $\psi \in \sub{\p}$, $I \models^{\circ} \psi$ iff
  $\Imc \models \psi^{\circ}$, with $\circ \in \{ +, - \}$.
\end{claim}
\begin{proof}

The proof is by induction on $\psi$.
The base cases are as follows.
\begin{itemize}
\item For $\psi = (C \sqsubseteq D)$, we have the following:
  $I \models^{\circ} C \sqsubseteq D$ iff, for all
  $d \in \InnDom^{I}$, $d \in C^{I}$ implies $d \in D^{I}$. By
  definition of $\Imc$ and Claim~\ref{claim:alcoiotaconsat},
  this holds iff, for all $d \in \Ex^{\Imc}$,
  $d \in (C^{\exrel})^{\Imc}$ implies $d \in (D^{\exrel})^{\Imc}$.
  Equivalently, for all $d \in \Delta^{\Imc}$, if
  $d \in (\Ex \sqcap C^{\exrel})^{\Imc}$, then
  $d \in (D^{\exrel})^{\Imc}$. The previous step means that
  $\Imc \models \Ex \sqcap C^{\exrel} \sqsubseteq D^{\exrel}$, that
  is, $\Imc \models (C \sqsubseteq D)^{\circ}$.
\end{itemize}
We now distinguish the cases of positive and negative semantics.
For $\circ = +$, we have the following.
\begin{itemize}
\item For $\psi = C(\tau)$, we consider the two forms of $\tau$.
  \begin{itemize}
  \item Let $\tau = a$. We have that $I \models^{+} C(a)$ iff
    $a^{I} \in C^{I}$.
    By definition of $\Imc$, $a$ denotes in $\Imc$ and $a^{\Imc} = a^{I}$, thus, by
    Claim~\ref{claim:alcoiotaconsat}, the previous step means
    $a^{\Imc} \in (C^{\exrel})^{\Imc}$.  Equivalently, thus,
    $\Imc \models \{ a \} \sqsubseteq C^{\exrel}$, i.e.,
    $\Imc \models C(a)^{+}$.
%
  \item Let $\tau = \defdes D$.
  We have that
    $I \models^{+} C(\defdes D)$ holds iff
        $(\defdes D)^{I} \in C^{I}$.
    By definition of $\Imc$ and Claim~\ref{cla:alcoiotaterm}, we have
    that $d =(\defdes D)^{I}$ iff either
    $d = \defdes  (\Ex \sqcap D^{\exrel})^{\Imc}$, or $\defdes  (\Ex \sqcap D^{\exrel})$ does
    not denote in $\Imc$
and $d = a^{\Imc}_{\defdes D}$. Furthermore, by
Claim~\ref{claim:alcoiotaconsat}, $d\in C^I$ iff $d \in (C^{\exrel})^{\Imc}$.
Since $a_{\defdes D}$ denotes in $\Imc$, the previous step means that
$\Imc \models \{ \defdes (\Ex \sqcap D^{\exrel}) \} \sqcup (\lnot \exists u.\{ (\Ex \sqcap \defdes D^{\exrel}) \}
\sqcap \{ a_{\defdes D} \}) \sqsubseteq C^{\exrel}$.
By definition of $\cdot^{\exrel}$, this is equivalent to:
$\Imc \models (\{ \defdes D \}^{\exrel} \sqsubseteq C^{\exrel})$,
i.e., by definition of $\cdot^{+}$,
$\Imc \models C( \defdes D)^{+}$.
	\end{itemize}
\item For $\psi = r(\tau_{1},\tau_{2})$,  we consider the two forms of $\tau_{1}, \tau_{2}$.

\begin{itemize}
\item Let $\tau_{1} = a, \tau_{2} = b$. We have that
  $I \models^{+} r(a, b)$
  holds iff $(a^{I}, b^{I}) \in r^{I}$. By definition of $\Imc$, the
  previous step is equivalent to $(a^{\Imc}, b^{\Imc}) \in r^{\Imc}$,
  and since both $a, b$ denote in $\Imc$ this means that
  $\Imc \models \{ a \} \sqsubseteq \exists r.\{ b \}$, i.e.,
  $\Imc \models r(a, b)^{+}$.
\item Let $\tau_{1} = \defdes D, \tau_{2} = \defdes E$.
	%
  We have that $I \models^{+} r(\defdes D, \defdes E)$ holds iff
  $((\defdes D)^{I}, (\defdes E)^{I}) \in r^{I}$.  Let
  $d = (\defdes D)^I$ and $e = (\defdes E)^I$. By definition of $\Imc$
  and Claim~\ref{cla:alcoiotaterm}, we have that
  $d = (\defdes D)^I$ iff either $d = \defdes (\Ex \sqcap D^{\exrel})^{\Imc}$ or
  $\defdes (\Ex \sqcap D^{\exrel})$ does not denote in $\Imc$ and
  $d = a_{\defdes D}^{\Imc}$; similarly for $e = (\defdes E)^I$.
  By definition of $\Imc$, $(d,e) \in r^{\Imc}$.  Since, by definition
  of $\Imc$, we have that both $a_{\defdes D}$ and $a_{\defdes E}$
  denote in $\Imc$, the previous step is equivalent to
  $\Imc \models \{ \defdes (\Ex \sqcap D^{\exrel}) \} \sqcup (\lnot \exists u. \{
  \defdes (\Ex \sqcap D^{\exrel}) \} \sqcap \{ a_{\defdes D} \}) \sqsubseteq
  \exists r. ( \{ \defdes (\Ex \sqcap E^{\exrel} ) \} \sqcup (\lnot \exists u. \{
  \defdes (\Ex \sqcap E^{\exrel} ) \} \sqcap \{ a_{\defdes E} \}))$.  This means
  $\Imc \models (\{ \defdes D \}^{\exrel}\sqsubseteq \exists r.\{
  \defdes E \}^{\exrel})$, i.e.,
  $\Imc \models r(\defdes D, \defdes E)^{+}$.
	
\item Let $\tau_{1} = a, \tau_{2} = \defdes D$.
	%
  We have that $I \models^{+} r(a, \defdes D)$ holds iff
  $(a^{I}, (\defdes D)^{I}) \in r^{I}$.  Let $d = (\defdes D)^I$ . By definition of $\Imc$ and
  Claim~\ref{cla:alcoiotaterm}, we have that
  $d = (\defdes D)^I$ iff either $d = \defdes (\Ex \sqcap D^{\exrel})^{\Imc}$, or
  $\defdes (\Ex \sqcap D^{\exrel})$ does not denote in $\Imc$ and
  $d = a_{\defdes D}^{\Imc}$. By definition of $\Imc$,
  $(a^I,d) \in r^{\Imc}$.  Since, by definition of $\Imc$, we have that
  both $a$ and $a_{\defdes D}$ denote in $\Imc$, the previous step is
  equivalent to
  $\Imc \models \{ a \}^{\exrel} \sqsubseteq \exists r. ( \{ \defdes
  (\Ex \sqcap D^{\exrel} ) \} \sqcup (\lnot \exists u. \{ \defdes (\Ex \sqcap D^{\exrel}) \}
  \sqcap \{ a_{\defdes D} \}))$.  This means
  $\Imc \models (\{ a \}^{\exrel} \sqsubseteq \exists r.\{ \defdes D
  \}^{\exrel})$, i.e., $\Imc \models r(a, \defdes D)^{+}$.
\item Let $\tau_{1} = \defdes D, \tau_{2} = a$. This case can be shown
  similarly to the previous one.
\end{itemize}
\item For $\psi = (\tau_{1} = \tau_{2})$, we consider the two forms of
  $\tau_{1}, \tau_{2}$.
\begin{itemize}
\item Let $\tau_{1} = a, \tau_{2} = b$. We have that
  $I \models^{+} a = b$
  holds iff $a^{I} = b^{I}$. By definition of $\Imc$, the previous
  step means $a^{\Imc} = b^{\Imc}$, and since both $a, b$ denote in
  $\Imc$ this is equivalent to $\Imc \models \{ a \} \equiv \{ b \}$,
  i.e., $\Imc \models (a = b)^{+}$.
			%
\item Let $\tau_{1} = \defdes D, \tau_{2} = \defdes E$.
We have that $I \models^{+} \defdes D = \defdes E$ holds
  iff $(\defdes D)^{I} = (\defdes E)^{I}$.  Let $d = (\defdes D)^I$ and
  $e = (\defdes E)^I$. By definition of $\Imc$ and
  Claim~\ref{cla:alcoiotaterm}, we have that
  $d = (\defdes D)^I$ iff either $d = \defdes (\Ex \sqcap D^{\exrel})^{\Imc}$ or
  $\defdes (\Ex \sqcap D^{\exrel}$ does not denote in $\Imc$ and
  $d = a_{\defdes D}^{\Imc}$,
  and similarly for $e$.  Since, by definition of $\Imc$, we have that
  both $a_{\defdes D}$ and $a_{\defdes E}$ denote in $\Imc$, and
  $d=e$, then the previous step is equivalent to
  $\Imc \models \{ \defdes (\Ex \sqcap D^{\exrel}) \} \sqcup (\lnot \exists u. \{
  \defdes (\Ex \sqcap D^{\exrel}) \} \sqcap \{ a_{\defdes D} \})
  \equiv
  \{ \defdes (\Ex \sqcap E^{\exrel}) \} \sqcup (\lnot \exists u. \{ (\Ex \sqcap \defdes E^{\exrel}) \}
  \sqcap \{ a_{\defdes E} \})$.  This means
  $\Imc \models (\{ \defdes D \}^{\exrel} \equiv \{ \defdes E
  \}^{\exrel})$, i.e., $\Imc \models (\defdes D = \defdes E)^{+}$.
\item Let $\tau_{1} = a, \tau_{2} = \defdes D$.
  We have that $I \models^{+} a = \defdes D$ holds iff
  $a^{I} = (\defdes D)^{I}$.  Let $d = (\defdes D)^I$. By definition of
  $\Imc$ and Claim~\ref{cla:alcoiotaterm}, we have that
  $d = (\defdes D)^I$ iff either
  $d = \defdes (\Ex \sqcap D^{\exrel})^{\Imc}$, or $\defdes (\Ex \sqcap D^{\exrel})$ does not
  denote in $\Imc$ and $d = a_{\defdes D}^{\Imc}$.
  Since, by definition of $\Imc$, we have that both $a$ and
  $a_{\defdes D}$ denote in $\Imc$, and $a^\Imc = d$, the
  previous step is equivalent to
  $\Imc \models \{ a \}^{\exrel} \equiv \{ \defdes (\Ex \sqcap  D^{\exrel} _ \}
  \sqcup (\lnot \exists u. \{ \defdes (\Ex \sqcap D^{\exrel}) \} \sqcap \{
  a_{\defdes D} \})$.  This means
  $\Imc \models (\{ a \}^{\exrel} \equiv\{ \defdes D \}^{\exrel})$,
  i.e., $\Imc \models (a = \defdes D)^{+}$.
	\item Let $\tau_{1} = \defdes D, \tau_{2} = a$. This case can be shown similarly to the previous one.
\end{itemize}
\end{itemize}
For $\circ = -$, we have the following.
\begin{itemize}
\item For $\psi = C(\tau)$, $I \models^{-} C(\tau)$
  iff $\tau^{I} \in \InnDom^{I}$ and $\tau^{I} \in C^{I}$. 
   By Claim~\ref{claim:alcoiotaconsat} and
Claim~\ref{cla:alcoiotaterm},
   this means
  $(\tau^{\exrel})^{\Imc} \in (\Ex \sqcap C^{\exrel})^{\Imc}$. 
We show that the previous step is equivalent to $\Imc \models C(\tau)^{-}$, i.e.,
  $\Imc \models \{ \tau \}^{\exrel} \sqsubseteq \Ex \sqcap C^{\exrel}$.
  The $(\Rightarrow)$ direction is clear. For $(\Leftarrow)$, suppose that $\Imc \models \{ \tau \}^{\exrel} \sqsubseteq \Ex \sqcap C^{\exrel}$.
  If $\tau = a$, we have by definition of $\Imc$ that $a$ denotes in $\Imc$, hence $a^{\Imc} \in (\Ex \sqcap C^{\exrel})^{\Imc}$, as required.
  If $\tau = \defdes D$, suppose towards a contradiction that $\defdes (\Ex \sqcap D^{\exrel})$ does not denote in $\Imc$.
  Since
  $\Imc \models \{ \defdes (\Ex \sqcap D^{\exrel}) \} \sqcup (\lnot \exists u. \{ \defdes (\Ex \sqcap D^{\exrel}) \} \sqcap \{ a_{\defdes D} \}) \sqsubseteq \Ex \sqcap C^{\exrel}$
  and $a_{\defdes D}$ denotes in $\Imc$, we obtain that $a_{\defdes D}^{\Imc} \in \Ex^{\Imc}$, contradicting the fact that, by definition of $\Imc$, $a_{\defdes D}^{\Imc}  \in \Delta^{\Imc} \setminus \Ex^{\Imc}$. Therefore, $\defdes (\Ex \sqcap D^{\exrel})$ denotes in $\Imc$ and $\defdes (\Ex \sqcap D^{\exrel})^{\Imc} \in (\Ex \sqcap C^{\exrel})^{\Imc}$.
    
  %
			%
\item For $\psi = r(\tau_{1}, \tau_{2})$,
  $I \models^{-} r(\tau_{1},\tau_{2})$
  iff $\tau_{1}^{I}, \tau_{2}^{I} \in \InnDom^{I}$ and
  $(\tau_{1}^{I}, \tau_{2}^{I}) \in r^{I}$.
  By
  Claim~\ref{cla:alcoiotaterm},
  this means $(\tau_{1}^{\exrel})^{\Imc}, (\tau_{2}^{\exrel})^{\Imc} \in \Ex^{\Imc}$ and
  $((\tau_{1}^{\exrel})^{\Imc}, (\tau_{2}^{\exrel})^{\Imc}) \in r^{\Imc}$.
  Similarly to the previous case, it can be seen that we have
  equivalently
  $\Imc \models \{ \tau_{1} \}^{\exrel} \sqsubseteq \Ex \sqcap \exists r.(\Ex
  \sqcap \{ \tau_{2} \}^{\exrel})$, i.e.,
  $\Imc \models r(\tau_{1}, \tau_{2})^{-}$
			%
\item For $\psi = \tau_{1} = \tau_{2}$,
  $I \models^{-} \tau_{1} = \tau_{2}$
  iff $\tau_{1}^{I}, \tau_{2}^{I} \in \InnDom^{I}$ and
  $\tau_{1}^{I} = \tau_{2}^{I}$. By definition of $\Imc$, this means
  $\tau_{1}^{\Imc}, \tau_{2}^{\Imc} \in \Ex^{\Imc}$ and
  $\tau_{1}^{\Imc} = \tau_{2}^{\Imc}$.
  Similarly to the previous case, it can be seen that we have
  equivalently
  $\Imc \models \{ \tau_{1} \}^{\exrel} \sqcap \Ex \equiv \{ \tau_{2} \}^{\exrel} \sqcap
  \Ex$, i.e., $\Imc \models (\tau_{1} = \tau_{2})^{-}$.
%
\end{itemize}
The inductive cases of $\psi = \lnot ( \chi)$ and $\psi = (\chi \land \zeta)$, under positive or negative semantics, are straightforward.
\qedhere
\end{proof}

By Claim~\ref{claim:alcoiotaforsat} and the assumption that $I \models^{\circ} \p$ we can conclude $\Imc \models \p^{\circ} \land
\bigwedge_{\{ a \} \in \con{\p}} \top \sqsubseteq \exists u. \{ a \} \land \bigwedge_{\{ \defdes D \} \in \con{\p}} (\top \sqsubseteq \exists u. \{ a_{\defdes D} \} \land \{ a_{\defdes D} \} \sqsubseteq \lnot \Ex)$.

$(ii)$
Consider a partial interpretation
$\Imc = (\Delta^{\Imc}, \cdot^{\Imc})$ such that
$\Imc \models \p^{\circ} \land \bigwedge_{\{ a \} \in \con{\p}}
\top \sqsubseteq \exists u. \{ a \} \land \bigwedge_{\{ \defdes D \} \in \con{\p}} (\top \sqsubseteq \exists u. \{ a_{\defdes D} \}  \land \{ a_{\defdes D} \} \sqsubseteq \lnot \Ex)
$, where $\circ \in \{ +, - \}$.
We define a dual-domain interpretation $I = (\OutDom^{I}, \InnDom^{I},
\cdot^{I})$ by taking $\OutDom^{I} = \Delta^{\Imc}$, $\InnDom^{I} =
\Ex^{\Imc}$, $\topex^{I} = \Ex^{\Imc}$, $A^{I} = A^{\Imc}$, for all $A \in \NC$,
$r^{I} = r^{\Imc}$, for all $r \in \NR$, and $a^{I} = a^{\Imc}$, if
$a \in \NI$ denotes in $\Imc$, while $a^{I}$ is defined arbitrarily,
otherwise.  Moreover, for every $\defdes D$ such that
$\{ \defdes D \} \in \con{\p}$, we set
$d_{\defdes D} = a_{\defdes D}^{I}$ (while, for
$\{ \defdes D \} \not \in \con{\p}$, we set
$d_{\defdes D} \in \Delta^{I} \setminus \InnDom^{I}$ arbitrarily).
Note that, since
$\Imc \models \top \sqsubseteq \exists u. \{ a_{\defdes D} \} \land \{
a_{\defdes D} \} \sqsubseteq \lnot \Ex$, for every
$\{ \defdes D \} \in \con{\p}$,
we have
that $a_{\defdes D}$ denotes in $\Imc$, hence
$a_{\defdes D}^{I} = a_{\defdes D}^{\Imc}$, and
$a_{\defdes D}^{\Imc} \in \Delta^{\Imc} \setminus \Ex^{\Imc}$.
Therefore,
$d_{\defdes D} = a_{\defdes D}^{I}\in \OutDom^{I}\setminus
\InnDom^{I}$ is well-defined.  In addition, for all
$\{ a \} \in \con{\p}$,
we have that $a$ denotes in $\Imc$, since
$\Imc \models \top \sqsubseteq \exists u. \{ a \}$, and thus, by $I$
definition, $a^{I} = a^{\Imc}$.
Similarly to the proof of Claim~\ref{claim:alcoiotaconsat}, it can be seen that, for all $C \in \con{\p}$ and $d \in \Delta^{I}$, $d \in C^{I}$ iff $d \in (C^{\exrel})^{\Imc}$.
Moreover, as in the proof of Claim~\ref{claim:alcoiotaforsat}, we have again that, for all $\psi \in \sub{\p}$, $I \models^{\circ} \psi$ iff $\Imc \models \psi^{\circ}$, with $\circ \in \{ +, - \}$.
Since by assumption it holds that $\Imc \models \p^{\circ}$, we obtain $I \models^{\circ} \p$.
\qedhere
\end{proof}

Theorem~\ref{prop:redddtopart} below follows immediately from the previous lemma,
together with the observation that $\ALCOud$ formulas can be encoded as $\ALCOud$ ontologies.

\redddtopart*

%
%
%
%
%


%




\end{appendix}


\end{document}